\newif\ifshort
\shortfalse 

\newif\ifhighlight

\highlighttrue

\documentclass{ecai} 

\usepackage{latexsym}
\usepackage{amssymb}
\usepackage{amsmath}
\usepackage{amsthm} 
\usepackage{dsfont}
\usepackage{booktabs} 
\usepackage{paralist}
\usepackage{enumitem}
\usepackage{graphicx}
\usepackage{color}
\usepackage[colorlinks]{hyperref}
\hypersetup{
    linkcolor=blue,        
    citecolor=red,         
}
\usepackage{thm-restate}
\usepackage[title]{appendix}

\usepackage{fdsymbol}
\usepackage{cleveref} 
\usepackage{multirow}
\usepackage{makecell}
\usepackage[textsize=tiny,colorinlistoftodos,textwidth=1.5cm,linecolor=green!70!black, backgroundcolor=green!10, bordercolor=black,disable=true]{todonotes}
\usepackage{mdframed}

\newtheorem{theorem}{Theorem}

\newtheorem{corollary}{Corollary}
\newtheorem{proposition}{Proposition}

\newtheorem{claim}{Claim}

\theoremstyle{definition}
\newtheorem{remark}{Remark}

\newenvironment{claimproof}[1][\proofname]
{\proof[#1]}
{\endproof}

\usetikzlibrary{decorations.pathreplacing}

\newcommand{\BibTeX}{B\kern-.05em{\sc i\kern-.025em b}\kern-.08em\TeX}

\crefname{table}{Table}{Tables}
\crefname{figure}{Figure}{Figures}
\crefname{theorem}{Theorem}{Theorems}
\crefname{corollary}{Corollary}{Corollaries}
\crefname{observation}{Observation}{Observations}
\crefname{lemma}{Lemma}{Lemmas}
\crefname{example}{Example}{Examples}
\crefname{reduction}{Reduction}{Reductions}
\crefname{construction}{Construction}{Constructions}
\crefname{subsection}{Subsection}{Subsections}
\crefname{section}{Section}{Sections}
\crefname{claim}{Claim}{Claims}
\crefname{proposition}{Proposition}{Propositions}
\crefname{algorithm}{Algorithm}{Algorithm}
\crefname{definition}{Definition}{Definitions}
\crefname{openquestion}{Open Question}{Open Questions}
\crefname{obs}{Observation}{Observations}
\crefname{remark}{Remark}{Remarks}

\newcommand{\finalutil}{\ensuremath{u^*}}
\newcommand{\sumutil}{\ensuremath{u_{\Sigma}}}
\newcommand{\sumreqs}{\ensuremath{r_{\Sigma}}}
\newcommand{\sumcaps}{\ensuremath{c_{\Sigma}}}
\newcommand{\maxcap}{\ensuremath{c_{\max}}}
\newcommand{\maxutilf}{\ensuremath{u_{\max}}}
\newcommand{\maxreq}{\ensuremath{r_{\max}}}
\newcommand{\maxcapl}{\ensuremath{m^{\mathcal{L}}_c}}
\newcommand{\sumcapmins}{\ensuremath{\underline{c}_{\Sigma}}}
\newcommand{\fa}{f}
\newcommand{\loc}{p}
\newcommand{\ser}{s}
\newcommand{\FA}{F}
\newcommand{\LOC}{P}
\newcommand{\SER}{S}
\newcommand{\noser}{t}
\newcommand{\NN}{\mathds{N}}
\newcommand{\ZZ}{\mathds{Z}}

\newcommand{\reqf}{\boldsymbol{r}}
\newcommand{\capl}{\bar{\boldsymbol{c}}}
\newcommand{\caplmin}{\underline{\boldsymbol{c}}}
\newcommand{\profit}{\boldsymbol{u}}

\newcommand{\tieno}{{{n_{\sim}}}}

\newcommand{\xx}{\boldsymbol{x}}
\newcommand{\yy}{\boldsymbol{y}}

\newcommand{\asg}{\sigma}

\newcommand{\blocksize}{\rho}
\newcommand{\blockbnd}{\rho}
\newcommand{\blockbig}{\rho_m^\star}
\newcommand{\superblock}{homogeneous $\blocksize$-block}
\newcommand{\mvin}{{\hspace{0.8pt}\textup{in}}}
\newcommand{\mvout}{{\hspace{0.8pt}\textup{out}}}
\newcommand{\curr}{\textup{curr}}
\newcommand{\nxt}{\textup{next}}
\newcommand{\wtD}{\widetilde{D}}
\newcommand{\mytilde}[1]{\stackrel{\textstyle\sim}{\smash{#1}\rule{0pt}{0.6ex}}}
\newcommand{\mysubtilde}[1]{\stackrel{\textstyle\sim}{\smash{#1}\rule{0pt}{0.2ex}}}
\newcommand\myprime{\mkern-3.5mu\raise0.6ex\hbox{$\scriptscriptstyle\prime$}}
\def\altm{\mytilde{m}}
\def\altsubm{\mysubtilde{m}}
\def\altsubmprime{\mysubtilde{m}\hspace{0.5pt}\myprime}
\def\true{\texttt{true}}
\def\false{\texttt{false}}
\def\R{\mathcal{R}}
\def\vecx{\mathbf{x}}
\def\vecy{\mathbf{y}}
\def\status{\boldsymbol{\ell}}
\def\statQ{\lambda}
\def\Q{\Lambda}
\def\chainpath{\mathcal{P}^\star}
\makeatletter
\newcommand{\leqnomode}{\tagsleft@true\let\veqno\@@leqno}
\makeatother
\def\UU{\mathcal{U}}
\def\TT{\mathcal{T}}

\newcommand{\Ablock}{A^{\blocksize}}
\newcommand{\Ablockopt}{\bar{A}^{\blocksize}}
\newcommand{\Aconfig}[1]{A_r(#1)}
\newcommand{\altk}{k'}

\newcommand{\decprob}[3]{%

  \smallskip
  {  \centering
    \noindent\begin{minipage}{0.94\linewidth}%
      \textsc{#1}\\
      \textbf{Input:} #2\\
      \textbf{Question:} #3 
    \end{minipage}%
    \smallskip
    
    \par
  }
}

\newcommand{\dectask}[3]{%

  \smallskip   
  {  \centering
    \begin{minipage}{0.94\linewidth}%
      \textsc{#1}\\
      \textbf{Input:} #2\\
      \textbf{Task:} #3 
    \end{minipage}%
    \smallskip
    
    \par}
}

\newcommand{\RR}{\textsc{RR}}
\newcommand{\RefRes}{\textsc{Refugee Resettlement}}
\newcommand{\maxutil}{\textsc{MaxUtil}}
\newcommand{\pareto}{\textsc{Pareto}}
\newcommand{\feasible}{\textsc{Feasible}}
\newcommand{\feasibleRR}{\feasible-\RR}
\newcommand{\paretoRR}{\pareto-\RR}
\newcommand{\maxutilRR}{\maxutil-\RR}
\newcommand{\identical}{equal} 
\newcommand{\uniform}{equal} 
\newcommand{\binary}{binary}

\newcommand{\indifferent}{equal} 
\newcommand{\dichotomous}{dichotomous}

\newcommand{\acceptable}{acceptable}
\newcommand{\zerobound}{no lower quotas}

\DeclareMathOperator{\util}{util}
\DeclareMathOperator{\satur}{load}

\newcommand{\myemph}[1]{\emph{#1}}
\ifhighlight
\renewcommand{\myemph}[1]{{\color{green!40!black}\emph{#1}}}
\fi

\DeclareMathOperator{\res}{res}

\newcommand{\argmax}{\arg \max}
\newcommand{\argres}{\arg \res}

\newcommand{\config}{\ensuremath{c^{\FA}}}
\newcommand{\configs}{\ensuremath{C^{\FA}}}

\newcommand{\conref}[1]{Constraint~\eqref{#1}}

\newcommand{\loctype}{\ensuremath{\tau^{\LOC}}}
\newcommand{\loctypes}{\ensuremath{\mathcal{T}^{\LOC}}}

\newcommand{\NP}{NP}
\newcommand{\NPh}{\NP-hard}
\newcommand{\FPT}{FPT}
\newcommand{\PP}{P}
\newcommand{\XP}{XP}
\newcommand{\Wone}{W[1]}
\newcommand{\Woneh}{\Wone-hard}
\newcommand{\Wtwo}{W[2]}

\newcommand{\mypara}[1]{\smallskip
\noindent \textbf{#1}}

\newcommand{\todoS}[1]{\todo[linecolor=green!70!black, backgroundcolor=green!10]{S: #1}} 
\newcommand{\todoI}[1]{\todo[linecolor=blue!70!black, backgroundcolor=blue!10]{I: #1}} 
\newcommand{\todoHinline}[1]{
    \begin{mdframed}[backgroundcolor=orange!20]
    \scriptsize \color{orange!40!black}H: #1
  \end{mdframed}
}
\renewcommand{\todoHinline}[1]{
}

\newcommand{\appsymb}{$\star$}

\usepackage{etoolbox} 
\newcommand{\toappendix}[1]{%
  \gappto{\appendixtext}{
    {#1}
   }
}

\newcommand{\appendixproofwithstatement}[3]{%
  \gappto{\appendixtext}{
    \subsection{Proof of \cref{#1}}\label{proof:#1}
    #2
    \begin{proof}
    #3\end{proof}
  }
}

\newcommand{\appendixproof}[2]{%
  \gappto{\appendixtext}{
    \subsection{Proof of \cref{#1}}\label{proof:#1}
    #2
    }
}

\newcommand{\appendixcorrectnessproofwithstatement}[4]{%
  #1  
  \gappto{\appendixtext}{
    \subsection{Correctness of the Construction in the Proof
      of \cref{#2}}\label{proof:#2}
    {\normalfont\emph{#3}}

    #4
    }
}

\newcommand{\appendixsection}[1]{%
  \gappto{\appendixtext}{
    \section{Additional material for Section~\ref{#1}}
    \label{appsec:#1}
  }
}

\begin{document}


\begin{frontmatter}


\paperid{1227} 


\newcommand{\probORR}{
  Optimal Refugee Resettlement
}
\title{%
  Parameterized Algorithms for\\ \probORR
}


\author[A]{\fnms{Jiehua}~\snm{Chen}\orcid{0000-0002-8163-1327}\footnote{Equal contribution.}\thanks{Corresponding Author. Email: jiehua.chen@tuwien.ac.at.}}
\author[B,C]{\fnms{Ildik\'o}~\snm{Schlotter}\orcid{0000-0002-0114-8280}\footnotemark}
\author[A]{\fnms{Sofia}~\snm{Simola}\orcid{0000-0001-7941-0018}\footnotemark}

\address[A]{TU Wien, Austria}
\address[B]{HUN-REN Centre for Economic and Regional Studies, Hungary}
\address[C]{Budapest University of Technology and Economics, Hungary}


\begin{abstract}
We study variants of the \probORR\ problem where a set~$\FA$ of refugee families need to be allocated to a set~$\LOC$ of possible places of resettlement in a feasible and optimal way. 
Feasibility issues emerge from the assumption that each family requires certain services (such as accommodation, school seats, or medical assistance), while there is an upper and, possibly, a lower quota on the number of service units provided at a given place. 
Besides studying the problem of finding a feasible assignment, we also investigate two natural optimization variants. 
In the first one, we allow families to express preferences over~$P$, and we aim for a Pareto-optimal assignment. 
In a more general setting, families can attribute utilities to each place in~$P$, and the task is to find a feasible assignment with maximum total utilities.
We study the computational complexity of all three variants in a multivariate fashion using the framework of parameterized complexity. 
We provide fixed-parameter tractable algorithms for a handful of natural parameterizations, and complement these tractable cases with tight intractability results.
\end{abstract}

\end{frontmatter}


\section{Introduction}

At the 2023 Global Refugee Forum, the UN High Commissioner for Refugees reported that \emph{114 million} people are currently displaced due to persecution, human rights violations, violence, and wars, and made a direct appeal to everyone to join forces to help refugees find protection.\footnote{\url{https://www.unhcr.org/global-refugee-forum-2023}}
This immense number highlights the critical need for effective resettlement strategies that cater to diverse populations.

Refugee resettlement involves not just relocating individuals but also families, each with distinct needs and service requirements ranging from accommodation to education and medical assistance.
\citet{delacretaz2023matching} and \citet{AAMTT21RR} propose a \emph{multi-dimensional} and \emph{multiple knapsack} model to address these challenges.
Their model takes into account the specific needs of \emph{refugee families} who require a range of services, as well as the capacity constraints of potential hosting places that have specific upper and lower quotas on the services they can offer.
The goal is to determine 
a \myemph{feasible} assignment from the families to the places which satisfies the specific needs of the families while ensuring that no place is over- or under-subscribed according to its capacity constraints.
Additionally, the model may include a \myemph{utility score} for each family--place pair which estimates the ``profit'' that a family may contribute to a place; such profit could be for example the employment outcome.
Optimizing the assignment means finding a feasible assignment that yields a maximum total utility.

If we care about the welfare and choices of the refugee families,
we may allow them to express \emph{preferences} over places which they find acceptable~\cite{delacretaz2023matching}.
A standard optimality criterion in such a case is Pareto-optimality, which means that we aim for a feasible assignment for which no other feasible assignment can make one family better off without making another worse off.

Unfortunately, it is computationally intractable (i.e., \NPh) to determine whether a feasible assignment exists~\cite{AzizGaspersSunWalsh2019aamas}.
Similarly, it is \NPh\ to find a feasible assignment with maximum 
total utility or one that is Pareto-optimal, even if there are no lower quotas~\cite{aziz2018stability,gurski2019knapsack,AAMTT21RR}.
To tackle these complexities, we examine the parameterized complexity of the three computational problems for refugee resettlement that we study, \feasibleRR, \maxutilRR\, and \paretoRR, and provide parameterized algorithms for them.
We focus on canonical parameters such as the number of places ($m$), the number of refugee families ($n$), the number of services ($\noser$), and the desired utility ($\finalutil$).
We also consider additional parameters that are motivated from real-life scenarios, including the maximum number~$\maxreq$ of units required by a family per service and the maximum utility~$\maxutilf$ a family can contribute. 
The service units can reasonably be assumed to be small integers in practical situations when a family's requirements describe their need for housing (e.g., number of beds or bedrooms) or education (e.g., the number of school seats or kindergarten places). 
Our study provides new insights into the parameterized complexities of these problems, presenting fixed-parameter (FPT) algorithms for several natural parameterizations, and contrasting these with strong intractability results. See \cref{table:results} for an overview.
We summarize our main contributions as follows.

\mypara{Single service.} We develop an \FPT\ algorithm w.r.t.~$\maxreq$ for \feasibleRR; the algorithm also applies to \maxutilRR\ and \paretoRR\ when all families have the same utilities for all places (\emph{equal utilities}) or 
are indifferent between all of them (\emph{equal preferences}), respectively; see \cref{thm:onetboundedrmaxilp}.  
  The main idea is to group all families together that have the same requirements, and group all places together with the same lower and upper quotas.
  Then, we observe that either the upper quotas are small (i.e., bounded by a function in~$\maxreq$) so we can brute-force search all possible partitions of the families into different places, or 
  there is a so-called \myemph{\superblock} (see \cref{sec:singleservice} for the formal definition) that can be exchanged across the places, which enables us to replace the upper quota of each place with a value bounded by a function of~$\maxreq$.
  In this way, we bound the number of groups of families and places, and can use integer linear programming (ILP) to obtain an FPT algorithm for~$\maxreq$.

We also propose an FPT algorithm for the combined parameter~$m+\maxreq$ for the general case when families may have different utilities or preferences (\cref{thm:onetboundedrmax+m}).
 The generalized algorithm additionally uses the idea that \myemph{\superblock s} can be exchanged across places, and combines dynamic programming  with color-coding~\cite{AYZ95} to find an optimal solution in FPT time.

\mypara{Multiple services.} 
In \cref{thm:fptservmaxreq}, we extend the FPT algorithm of \cref{thm:onetboundedrmaxilp} for the setting of equal preferences or utilities to multiple service types 
by combining the parameters~$\maxreq$ and $\noser$, the number of services; we use the technique of $N$-fold integer programming~\cite{HOR-2013}. 
We present a more general FPT algorithm for \pareto-\RR\ with  parameter $\maxreq+\noser+m$ which also solves \maxutil-\RR\ if the number of different utility values is bounded (\cref{thm:fptservmaxcaploc}); this result relies on Lenstra's result on solving ILPs with bounded dimension~\cite{lenstra1983integer}.
Contrasting our algorithmic results, we prove that \pareto- and \maxutil-\RR\ are both \NPh\ already for three places, even if there are no lower quotas, all  upper quotas are~$1$, and families have equal preferences or utilities, respectively; see \cref{thm:pareto_loc_maxcap}.
\todoHinline{Please fill in the citations.}

\mypara{Related work.}
The model we study is the same as that of Ahani et al.~\cite{AAMTT21RR}. They formulate \maxutil-\RR\ via Integer Linear Programming (ILP) and study its performance.
The same model without lower quotas has attracted previous study:
It was introduced in a working paper by Delacr{\'e}taz et al.~\cite{delacretaz2016refugee} (
see also~\cite{delacretaz2023matching}).
The paper provides an algorithm for finding a Pareto-efficient matching when the preferences are strict, and also studies other stability concepts.
Aziz et al.~\cite{aziz2018stability} show that \emph{finding} a Pareto-optimal assignment is NP-hard even when the families are indifferent between places, and study a few other stability concepts.
Nguyen et al.~\cite{nguyen2021stability} use fractional matchings to find group-stable assignments which violate the quotas only a little.
None of the works above focuses on the parameterized complexity of the problems.

As already mentioned by \citet{AAMTT21RR}, the \maxutilRR\ problem is a generalization of the  \textsc{Multiple/Multidimensional Knapsack} problem~\cite{AAMTT21RR}.
The parameterized complexity of the latter has been studied by \citet{gurski2019knapsack}, and several of our hardness-results are obtained either directly from them or from modifications of their reductions.
\textsc{Multiple/Multidimensional Knapsack} however has neither lower quotas nor different profits for items depending on which knapsack they are placed in. They also assume the sizes and profits are encoded in binary, whereas we assume they are encoded in unary.
Hence, their parameterized algorithms are not directly applicable to our problems.
\feasibleRR\ generalizes \textsc{Bin Packing}~\cite{JKMS-binpacking-2013} and hence \textsc{Simple Multidimensional Partitioned Subset Sum}~\cite{ganian23group}; note that the latter two problems are equivalent.
Since \textsc{Bin Packing} is \Wone-hard w.r.t.\ the number of bins and the bins correspond to the places in our setting, \Wone-hardness for \feasibleRR\ follows; see \cref{prop:wtbinpacking}.



The problem can be seen as an extension of different classical matching problems.
We can model \textsc{Matching with diversity constraints}~\cite{BiroFleinerIrvingManlove2010,HamadaIwamaMiyazaki2016algorithmica,CGH2021DiverSM,AzizGaspersSunWalsh2019aamas,abdulkadirouglu2005college,kurata2017controlled} by using services as types. 
In the case where we have a single service, the problem can be seen as a variant of \textsc{Matching with Sizes}~\cite{biro-mcdermid-matching-sizes,mcdermid2010keeping}, where the service requirements correspond to the sizes.

Refugee resettlement has also been studied in the literature under other types of models: 
Online setting~\cite{andersson2018dynamic,ahani2023dynamic,bansak2024outcome},
one-to-one housing~\cite{andersson2020assigning},
preferences based on weighted vectors~\cite{xepapadeas2022refugee}, hedonic games~\cite{kuckuck2019refugee}, and
placing refugees on a graph~\cite{knop2023host,lisowski2023swap,schierreich2023anonymous}.

\mypara{Paper structure.} In \cref{sec:prelim}, we formally define \RR.
We investigate the case when there is only one service and when there are multiple services in \cref{sec:singleservice} and \cref{sec:moreservices}, respectively.
In \cref{sec:moreservices}, we first look at the \feasibleRR\ problem, followed by \paretoRR\ the problem, and finally the \maxutilRR\ problem.
We conclude with a discussion on potential areas for future research in \cref{sec:conclusion}.

\section{Preliminaries}\label{sec:prelim}
\appendixsection{sec:prelim}
For an integer~$z$, we use $[z]$ to denote the set~$\{1,2,\dots,z\}$.
Given two vectors~$\xx$ and $\yy$ of same length, we write \myemph{$\xx \le \yy$} if for each coordinate~$i$ it holds that $\xx[i]\le \yy[j]$.

%
An \myemph{instance} of \RR\ is a tuple~$(\FA, \LOC, \SER, (\reqf_i)_{\fa_i\in \FA}, (\caplmin_j,\capl_j)_{\loc_j\in \LOC})$ with the following information.
\begin{compactitem}[--]
  \item $\FA$ denotes a set of $n$ refugee \myemph{families} with $\FA=\{\fa_1, \dots, \fa_{n}\}$,
  \item $\LOC$ denotes a set of $m$ \myemph{places} with $\LOC = \{\loc_1, \dots, \loc_{m}\}$, and
  \item $\SER$ denotes a set of $\noser$ \myemph{services} $\SER = \{\ser_1, \dots, \ser_{\noser}\}$, such that 
  \item each family~$\fa_i\in \FA$ has a \myemph{requirement} vector~$\reqf_i \in \NN^\noser$ where, for every $\ser_k \in \SER$, the value $\reqf_i[k]$ determines how many units of service~$\ser_k$ the family~$\fa_i$ requires, and
  \item each place $\loc_j \in \LOC$ has two vectors~$\caplmin_j, \capl_j\in \NN^{\noser}$, denoted as \myemph{lower quota} and \myemph{upper quota} which indicate  
  for every service~$\ser_k \in \SER$, the minimum and maximum number of units place~$\loc_j$ can provide. 
Non-zero lower quotas for places may for example follow from an obligation for a place to house at least a certain number of refugees. If the lower quota of every place is a zero-vector, then we say that the instance has \myemph{\zerobound}.
\end{compactitem}

\mypara{Assignments.} Given an instance of \RefRes, an \myemph{assignment} is a function $\asg \colon \FA \to \LOC \cup \{\bot\}$; we say that
$\fa_i \in \FA$ is \myemph{assigned} to a place $\loc_j \in \LOC$ if $\asg(\fa_i) = \loc_j$, and
$\fa_i$ is \myemph{unassigned} if $\asg(\fa_i) = \bot$.
We define the \myemph{load} vector of a place~$\loc_j \in \LOC$ under~$\asg$ as~$\satur(\loc_j,\asg)\coloneqq \sum_{\fa_i \in \asg^{-1}(\ell_j)}  \reqf_i[k]$; for each service~$\ser_k\in \SER$, $\satur(\loc_j, \asg)[k]$ denotes the number of units that are required by the refugees that are assigned to~$\loc_j$. 
An assignment is \myemph{complete} if it does not leave any families unassigned.
An assignment is \myemph{feasible} if for every place~$\loc_j \in \LOC$ the load vector is within the lower and upper quota, i.e., $ \caplmin_j \le \satur(\loc_j,\asg)\le \capl_j$.
Place $\loc_j$ can \myemph{accommodate} a set of families~$\FA' \subseteq \FA$
if $\sum_{\fa_{i'} \in \FA'}  \reqf_{i'}[k] \leq \capl_j[k]$ for each~$\ser_k \in \SER$.

\mypara{Utilities.} Each family may contribute a certain utility to each place. 
To model this, each family $\fa_i \in \FA$ expresses an integral \myemph{utility} vector~${\profit_i \in \ZZ^m}$, where for every~$\loc_j \in \LOC$, the value~$\profit_i[j]$ indicates the utility of family~$\fa_i$ if assigned to~$\loc_j$.
Note that we also allow negative utilities, but it will be evident that all hardness results hold even if the utilities are non-negative.
Given an assignment $\asg$, we define the \myemph{(total) utility} of the assignment as the sum of all utilities obtained by the families, i.e.,~\myemph{$\util(\asg)$} $=\sum_{\loc_j \in \LOC} \sum_{\fa_i \in \asg^{-1}(\loc_j)} \profit_i[j]$. 
We consider two special kinds of utility vectors.
We say that the families have \myemph{\identical\ utilities} if all utility values~$\profit_i[j]$ are equal and positive over all families~$\fa_i \in \FA$ and places~$\loc_j \in \LOC$,
and families have \myemph{\binary} utilities if each utility value is either zero or one. 

\mypara{Preferences and Pareto-optimal assignments.} 
Each family $\fa_i \in \FA$ may only find a subset of places \myemph{acceptable} and may have a \myemph{preference list~$\succeq_i$} over the acceptable places, i.e., a 
weak order over a subset of $\LOC$. For a family~$\fa_i$ and two places~$\loc$ and~$\loc'$ in its preference list, $\loc \succeq_i \loc'$ means that $\fa_i$ \myemph{weakly prefers}~$\loc$ to~$\loc'$. 
If $\loc \succeq_i \loc'$  and $\loc' \succeq_i \loc$, then we write $\loc \sim_i \loc'$ and say that $\fa_i$ is \myemph{indifferent} between~$\loc$ and~$\loc'$. 
We write 
$\loc \succ_i \loc'$ to denote that $\fa_i$ \myemph{(strictly) prefers} $\loc$ to~$\loc'$, meaning that $\loc \succeq_i \loc'$ but $\loc \not\succeq_i \loc'$.
%
If the preference list of~$\fa_i$ contains~$\loc$, then~$\fa_i$ finds $\loc$ \myemph{acceptable}.
We assume that each family~$\fa_i$ prefers being assigned to some place in his preference list over being unassigned; accordingly, we write \myemph{$\loc \succ_i \bot$}.
An assignment is \myemph{\acceptable} if every family is either unassigned or assigned to a place it finds acceptable.

We also define \myemph{\indifferent} and \myemph{\dichotomous} preferences:
If every family finds every place acceptable and is additionally indifferent between them, the preferences are \indifferent. If every family is indifferent between every place it finds acceptable, the preferences are \dichotomous.


A feasible and \acceptable\ assignment $\asg$ is \myemph{Pareto-optimal} if it admits no \myemph{Pareto-improvement}, that is, a feasible and \acceptable\ assignment $\asg'$ such that $\asg'(\fa_i) \succeq_i \asg(\fa_i)$ for every $\fa_i \in \FA$ and there exists at least one family $\fa_{i'} \in \FA$ such that $\asg'(\fa_{i'}) \succ_{i'} \asg(\fa_{i'})$. %

\toappendix{
\subsection{Example for \RefRes}
\label{sec:app:example}

Consider an instance of \RefRes\	 with four families, $f_1,f_2,f_3,$ and $f_4$, two places $\loc_1$ and $\loc_2$. Families may require two kinds of services: housing, which is determined directly by the number of people in the family, and school seats for their children. 

The requirement of the families is as follows.

{\centering
\smallskip
\begin{tabular}{l@{\hspace{3pt}}l}
$f_1$: & $\reqf_1=(4,2)$, i.e., a family of four with two school-age children;  \\
$f_2$: & $\reqf_2=(2,0)$, i.e., a family of two, no children; \\
$f_3$: & $\reqf_3=(6,2)$, i.e., a family of six with two school-age children;  \\
$f_4$: & $\reqf_4=(3,1)$, i.e., a family of three with one school-age child.  
\end{tabular}
\smallskip
}

Place~$\loc_1$ offers housing for at most 10 people, while place~$\loc_2$ can house at most 8 people. 
To promote diversity, both places aim to enroll at least two refugee children in their schools, and both have three school seats to offer. 
Hence, the lower and upper quotas are as follows:

\smallskip
\begin{tabular}{l@{\hspace{3pt}}l}
$\loc_1$: & $\caplmin_1=(0,2)$ and $ \capl_1=(10,3)$; \\
$\loc_2$: & $\caplmin_2=(0,2)$ and $ \capl_2=(8,3)$.
\end{tabular}
\smallskip

Both places are acceptable to each refugee  family, and the preferences of the families are defined as

\smallskip
\begin{tabular}{l@{\hspace{3pt}}l}
$f_1$: & $\loc_1 \succ \loc_2$; \\
$f_2$: &  $\loc_2 \succ \loc_1$; \\
$f_3$: & $\loc_2 \succ \loc_1$; \\
$f_4$: & $\loc_2 \succ \loc_1$.
\end{tabular}
\smallskip

Alternatively, families may have the following utility vectors:

\smallskip
\begin{tabular}{l@{\hspace{3pt}}l}
$f_1$: & $\profit_1=(2,1)$; \\
$f_2$: & $\profit_2=(1,2)$; \\
$f_3$: & $\profit_3=(1,2)$; \\
$f_4$: & $\profit_4=(1,2)$.
\end{tabular}
\smallskip

Consider the assignment $\asg$ that assigns families $f_2$ and $f_3$ to~$\loc_1$, and assigns families~$f_1$ and~$f_4$ to~$\loc_2$. It is clear that 
\begin{linenomath*}
\begin{align*}
(0,2)=\caplmin_1 & \leq \sum_{\fa_i \in \asg^{-1}(\loc_1)} \reqf_i = (8,2) \leq \capl_1=(10,3) \text{ and} \\
(0,2)=\caplmin_2 & \leq \sum_{\fa_i \in \asg^{-1}(\loc_1)} \reqf_i = (7,3) \leq \capl_2=(8,3), 
\end{align*}
\end{linenomath*}
implying that $\asg$ is a feasible assignment.

It is also not hard to see that $\asg$ is Pareto-optimal as well. To see this, consider first family~$\fa_4$ who is assigned to its favorite place. Notice that once we decide to assign~$\fa_4$ to~$\loc_2$, neither~$\fa_2$ nor~$\fa_3$ may be additionally assigned to~$\loc_2$: the former would violate the upper quota of~$\loc_2$ to house at most~$8$ people, while the latter would violate the lower quota of~$\loc_2$ to enroll at least two children in its school (unless some further family is added, violating the upper quota).
Notice also that once $\fa_2$ and~$\fa_3$ are both assigned to~$\loc_1$, it is not possible to accommodate~$\fa_1$ alongside them  at~$\loc_1$, as that would violate the upper quota of~$\loc_1$ to house at most~$10$ people. 
Therefore, without disimproving the situation of~$\fa_4$, we cannot improve the situation of any other family in any feasible assignment, showing the Pareto-optimality of~$\asg$.

It is also clear that in the setting with utilities, the total utility of~$\asg$ is $\util(\asg)=3\cdot 1 + 2=5$.

Consider now the assignment $\asg'$ 
that assigns families $f_1$ and $f_4$ to~$\loc_1$, and assigns families~$f_2$ and~$f_3$ to~$\loc_2$. It is clear that 
\begin{linenomath*}
\begin{align*}
(0,2)=\caplmin_1 & \leq \sum_{\fa_i \in \asg'^{-1}(\loc_1)} \reqf_i = (7,3) \leq \capl_1=(10,3) \text{ and } \\
(0,2)=\caplmin_2 & \leq \sum_{\fa_i \in \asg'^{-1}(\loc_1)} \reqf_i = (8,2) \leq \capl_2=(8,3), 
\end{align*}
\end{linenomath*}
implying that $\asg'$ is a feasible assignment.

We claim that $\asg'$ is a maximum-utility feasible assignment. Clearly, its utility is $\util(\asg')=3 \cdot 2 + 1=7$. Thus, any assignment whose utility exceeds $\util(\asg')$ must assign all families to the place for which they have utility~$2$. However, such an assignment is not feasible, as assigning each  of $\fa_2, \fa_3$, and $\fa_4$ to~$\loc_2$ would violate its upper quota to house at most 8 people.

}

\mypara{Central problems.} We are now ready to define our problems.

\decprob{\feasible-\RR}{An instance $I$ of \RefRes.}{Is there a feasible assignment~$\asg$ for~$I$?}

\decprob{\maxutil-\RR}{An instance $I$ of refugee resettlement, a utility vector $\profit_i \in \mathbb{Z}^m$ for each family $\fa_i \in \FA$, and an integer bound~$\finalutil$.}{Is there a feasible assignment~$\asg$ for~$I$ such that $\util(\asg) \geq \finalutil$?}

\dectask{\pareto-\RR}{An instance $I$ of refugee resettlement and a preference order~$\succeq_i$ for every $\fa_i \in \FA$.}{\emph{Find} a feasible and \acceptable\ Pareto-optimal assignment~$\asg$ for $I$ or report that none exists.}
\smallskip
We remark that there is a straightforward way to reduce \pareto-\RR\ to the optimization variant of \maxutil-\RR\ in the following sense.
Suppose that an algorithm~$\mathcal{A}$ finds a maximum-utility feasible assignment for each instance of \maxutil-\RR\ that admits a feasible assignment.
Such an algorithm 
can be used to solve an instance~$I$ the \pareto-\RR\ problem as follows.
\begin{restatable}[\appsymb]{obs}{obsparmax}
\label{rem:maxutil-vs-pareto}
Given an instance~$I$ of \pareto-\RR, construct an instance~$I'$ of \maxutil-\RR\ as follows.
For each family~$\fa_i \in \FA$: 
\begin{compactitem}
\item for every place~$\loc_j$ that $\fa_i$ finds acceptable, set $\profit_i[j] = |\{\loc_{j'} \in \LOC \mid \loc_j \succeq_i \loc_{j'}\}|$;
 \item  for each place~$\loc_j$ that $\fa_i$ finds unacceptable, set
$\profit_i[j]=-m\cdot n$. 
\end{compactitem}
Let $\asg$ be a maximum-utility feasible assignment for~$I'$. If ${\util(\asg)>0}$, then $\asg$ is a feasible, \acceptable, and Pareto-optimal assignment for~$I$; otherwise 
there is no feasible and \acceptable\ assignment for~$I$.
\end{restatable}

\appendixproofwithstatement{rem:maxutil-vs-pareto}{\obsparmax*}{
Clearly, $\asg$ has positive utility if $\asg$ is an \acceptable\ assignment for~$I$, and it has negative utility if it is unacceptable for~$I$ (since already one utility value of~$-m \cdot n$ guarantees $\util(\asg)<0$).
Hence, if $\util(\asg)<0$, then 
 no feasible assignment exists for~$I$. Otherwise, we know that $\asg$ is feasible for~$I$. Since $\asg$ has maximum utility among all assignments, it is necessarily Pareto-optimal, because a Pareto-improvement would imply an assignment with greater utility value. 
}

\mypara{Parameterization.}
We study the following parameters:
\begin{compactitem}
\item number of places ($m=|\LOC|$),
\item number of refugee families ($n=|\FA|$), 
\item number of services ($\noser = |\SER|$),
\item maximum number of units required for all services and by all families~($\maxreq = \max \{\reqf_i[k] : \fa_i \in \FA, \ser_k \in \SER\}$).
\end{compactitem}

We also study the following parameters for \maxutil-RR: the total utility bound~$\finalutil$ and the maximum utility brought by a family
$\maxutilf = \max \{\profit_i[j] :\fa_i \in \FA, \loc_j \in \LOC\}$. 


\todoHinline{Turn the following into a remark?}

Additionally, we consider the maximum length of the ties in preference lists. However, this parameter is upper-bounded by $m$, and most problems are already hard w.r.t.\ $m$.
We also consider the highest upper quota any place has for a service $\maxcap = \max_{\loc_j \in \LOC, \ser_k \in \SER}\capl_j[k]$, but discover that this parameter behaves very similarly to the smaller and better-motivated parameter~$\maxreq$. Note that we may assume that for each family there is at least one place that can accommodate it, otherwise we can remove the family from our instance; this implies that we can assume $\maxcap \geq \maxreq$.

We also obtain \FPT\ results w.r.t.\ the sum of the capacities of the places, that is, $\sumcaps = \sum_{\loc_j \in \LOC, \ser_k \in \SER}\capl_j[k]$, and the sum of the requirements of the families $\sumreqs = \sum_{\fa_i \in \FA, \ser_k \in \SER}\reqf_i[k]$.
If there are \zerobound, we also have an \FPT\ result w.r.t.\ the sum of utilities $\sumutil = \sum_{\fa_i \in \FA, \loc_j \in \SER}\profit_i[j]$. If the instance has non-zero lower quotas, then the problem is hard even when all utilities are zero, and this parameter is not helpful.
We also study the complexity w.r.t.\ the number of agents who have ties in their preference lists $\tieno$.
\ifshort
\else Discussion on parameters $\sumcaps$, $\sumreqs$, $\sumutil$, and $\tieno$ are deferred to Appendix~\ref{sec:moreparameters}.
\fi


\newcommand{\OQ}{\textbf{\color{orange!70!black}?}}

\newcommand{\unwritten}[1]{#1} 
\newcommand{\our}[1]{\textbf{#1}}
\newcommand{\missing}[1]{\textit{ #1?}}

\newcommand{\tNPh}{{\color{red!60!black}NPh}}

\newcommand{\tid}{$^{\circ}$}
\newcommand{\eqpref}{$^{=}$}
\newcommand{\tidalg}{$^{\smallblackdiamond}$}
\newcommand{\eqprefalg}{$^{\star}$}
\newcommand{\dicalg}{$^{\varheartsuit}$}

\newcommand{\lbres}{$^\dag$}
\newcommand{\lbalg}{$^\ddag$}

\newcommand{\tFPT}{{\color{green!40!black} \FPT}}
\newcommand{\tPP}{{\color{green!60!black}\PP}}
\newcommand{\tXP}{{\color{blue!60!black}\XP}}
\newcommand{\tWoneh}{{\color{red!50!blue}W1h}}
\newcommand{\tWtwoh}{{\color{red!70!blue}\Wtwo-h}}

\newcommand{\refpropfeashardm}{\unwritten{[P\ref{prop_feas_hard_m}]}}
\newcommand{\refthmparetoPm}{[P\ref{thm:paretoPm1}]}
\newcommand{\refthmfeasmthard}{\unwritten{[T\ref{thm:feas:mt1hard}]}}
\newcommand{\refxpmt}{\unwritten{[P\ref{xp:mt}]}}
\newcommand{\reffptnmodel}{\unwritten{[P\ref{prop:fpt-n}]}}
\newcommand{\refthmsumcapsemikernel}{\unwritten{[T\ref{thm:sumcapsemikernel}]}}
\newcommand{\reffptrequirementsmodel}{\unwritten{[P\ref{fpt:requirements:model1}]}}
\newcommand{\refxpfinalutilmodel}{\unwritten{[P\ref{xp:finalutil}]}}

\newcommand{\reffptfinalutiltone}{\unwritten{[T\ref{prop:fpt-finalutilt1}]}}

\newcommand{\reffptfinalutiltonenph}{\unwritten{[T\ref{prop:fpt-finalutilt1},P\ref{prop:wtbinpacking}]}}

\newcommand{\refpropfeashardt}{\unwritten{[T\ref{prop_feas_hard_t}]}}
\newcommand{\refpropwtbinpacking}{\unwritten{[P\ref{prop:wtbinpacking}]}}
\newcommand{\refthmsatwithties}{\unwritten{[T\ref{thm:satwithties}]}}
\newcommand{\refthmonetboundedrmaxilp}{\unwritten{[T\ref{thm:onetboundedrmaxilp}]}}
\newcommand{\reffptsumutilsmodel}{\unwritten{[P\ref{fpt:sumutils:model}]}}
\newcommand{\refpropfpttieno}{\unwritten{[P\ref{prop:fpt_tieno}]}}

\newcommand{\refpropparetotwolocs}{\unwritten{[P\ref{prop:pareto_two_locs}]}}
\newcommand{\refpropparetotwocap}{\unwritten{[P\ref{prop:pareto_two_cap}]}}

\newcommand{\refthmparetolocmaxcap}{\unwritten{[T\ref{thm:pareto_loc_maxcap}]}}
\newcommand{\refthmfptservmaxreq}{\unwritten{[T\ref{thm:fptservmaxreq}]}}
\newcommand{\refthmfptservmaxcaploc}{\unwritten{[T\ref{thm:fptservmaxcaploc}]}}

\newcommand{\refthmonetboundedrmaxm}{\unwritten{[T\ref{thm:onetboundedrmax+m}]}}

\newcommand{\nolower}{LQ$=$0}
\newcommand{\withlower}{LQ$\neq$0}
\newcommand{\doublefpt}{\our{\tFPT}\,/\,\our{\tFPT}}
\newcommand{\doublexp}{\our{\tXP}\,/\,\our{\tXP}}
\newcommand{\doublenph}{\our{\tNPh}\,/\,\our{\tNPh}}
\newcommand{\equtilpref}{\;\;\scriptsize eq.\ util./pref.}
\newcommand{\dichopref}{\;\;\scriptsize dicho.\ pref.}
\newcommand{\binaryutil}{\;\;\scriptsize binary util.}
\begin{table}[t!]
  \centering
  \extrarowheight=.5\aboverulesep
  \aboverulesep=1pt
  \belowrulesep=1pt
  \addtolength{\extrarowheight}{\belowrulesep}
 \resizebox{\columnwidth}{!}
  {
    \begin{tabular}{@{}l@{\,} | @{\,} c@{\,}l @{\,}c@{\,} c@{\,}l @{}c@{\,} c@{}l @{}} 
\toprule
      {\small Parameter} &  \multicolumn{2}{@{\;}c@{\;}}{\feasible} & & \multicolumn{2}{@{\;}c@{\;}}{\maxutil} &  & \multicolumn{2}{@{\;}c@{\;}}{\pareto} \\
      \cline{5-6} \cline{8-9}  
& & & & \multicolumn{1}{c}{\nolower~/~\withlower}& & & \multicolumn{1}{c}{\nolower~/~\withlower} & \\
      \midrule
      $m$ &  \our{\tWoneh} & \refpropwtbinpacking &
             & \our{\tWoneh\tid}/\our{\tWoneh\tid} & \refpropwtbinpacking &
             & \our{\tWoneh\eqpref}/\our{\tWoneh\eqpref} & \refpropwtbinpacking\\[-.8ex]
             & \our{\tXP} & \refxpmt & & \doublexp &  \refxpmt &
      &  \doublexp &  \refxpmt\\\hline
      $\maxreq $ & \our{\tFPT} & \refthmonetboundedrmaxilp & & \doublenph
              &\refthmsatwithties & & \doublenph  & \refthmsatwithties\\[-.8ex]
    \equtilpref &  -- & -- & & \doublefpt 
          & \refthmonetboundedrmaxilp & & \doublefpt& \refthmonetboundedrmaxilp \\\hline
      $m + \maxreq$ & \our{\tFPT} & \refthmonetboundedrmaxm && \doublefpt
        & \refthmonetboundedrmaxm & & \doublefpt & \refthmonetboundedrmaxm \\\hline
      $\maxutilf$ & -- & -- &
                                                                                             & \tNPh\tid/\tNPh\tid & ~\cite{gurski2019knapsack} 
                                                          & & -- & -- \\  
      \finalutil  & -- & -- &&  \our{\tFPT}/\our{\tNPh}\tid & \reffptfinalutiltone/\refpropwtbinpacking & & -- & --  \\ 
      \bottomrule\multicolumn{9}{c}{}\\[-2ex]  \toprule
      $m+\maxreq$ &  \our{\tNPh} & \refpropfeashardm && \tNPh\tid/\tNPh\tid & 
                                                                              \cite{gurski2019knapsack} && \our{\tNPh\eqpref}/\our{\tNPh\eqpref} & \refthmparetolocmaxcap/\refpropfeashardm \\ \hline
      $\noser$ &  \our{\tNPh} & \refpropwtbinpacking & & \tNPh\tid/\tNPh\tid & \cite{gurski2019knapsack} &&  \tNPh\eqpref/\tNPh\eqpref  & \cite{aziz2018stability} \\ \hline
         $n$ & \our{\tFPT} & \reffptnmodel & & \doublefpt & \reffptnmodel & & \doublefpt & \reffptnmodel \\ \hline

      $m + \noser$ &  \our{\tWoneh} & \refpropwtbinpacking &
       & \our{\tWoneh\tid}/\our{\tWoneh\tid} & \refpropwtbinpacking
      && \our{\tWoneh\eqpref}/\our{\tWoneh\eqpref} & \refpropwtbinpacking  \\[-.8ex]  
       & \our{\tXP}& \refxpmt && \doublexp & \refxpmt  
       && \doublexp & \refxpmt  \\  \hline
      $\noser + \maxreq$ & \our{\tFPT} & \refthmfptservmaxreq &
      & \doublenph & \refthmsatwithties && \doublenph & \refthmsatwithties  \\
     \equtilpref & -- & -- &
     & \our{\tFPT}/\our{\tFPT} & \refthmfptservmaxreq && \doublefpt & \refthmfptservmaxreq  \\ \hline
      {$m + \noser+ \maxreq$} &  \our{\tFPT} & \refthmfptservmaxreq
      && \doublexp,\OQ & \refxpmt  && \doublefpt & \refthmfptservmaxcaploc\\
     \binaryutil  &  -- & -- &&
      \our{\tFPT} & \refthmfptservmaxcaploc & & -- & --\\  \hline
      \finalutil& -- & -- && 
      \tWoneh\tid/\our{\tNPh}\tid & \cite{gurski2019knapsack}/\refpropfeashardm  & & -- & -- \\[-.8ex]
                         & -- & -- && \our{\tXP}/\our{\tNPh}\tid & \refxpfinalutilmodel/\refpropfeashardm && -- & --  \\  \bottomrule
     \multicolumn{8}{c}{} \\[-2ex]
    \end{tabular}}
  \caption{All three problems are NP-hard in general; see \cite{AzizGaspersSunWalsh2019aamas}, \cite[T32]{gurski2019knapsack},\cite[P7.1]{aziz2018stability}.
Above: Results for the single-service case ($t = 1$).
We skip the parameterization by $n$ since for this case since it is FPT for the more general case. 
Below: Results for the general case. Here, we skip the parameterization by~$\maxutilf$ since it is already \NP-hard for the single-service case.
Bold faced results are obtained in this paper. \nolower\ (resp.\ \withlower) refers to the case when lower quotas are zero (resp.\ may be positive).
\tNPh\ means that the problem remains NP-hard even if the corresponding parameter 
is constant.
All hardness results hold for \dichotomous\ preferences or \ \binary\ utilities.
Additionally, \tid\ (resp.\ \eqpref) means hardness results hold even for equal utilities (resp.\ preferences).
\ifshort
\else The results for the remaining parameter combinations are deferred to Appendix~\ref{sec:moreparameters}, \cref{table:utilparams}.
\fi \label{table:results}}
\end{table}

\section{Single service}
\label{sec:singleservice}
\appendixsection{sec:singleservice}

Let us assume that there is only a single service in our input instance. Thus, we will simply refer to~$\reqf_i[1]$ as the \myemph{requirement} of a family~$\fa_i \in \FA$, and we will write $\reqf_i=\reqf_i[1]$ accordingly. 
Observe that we may assume w.l.o.g.\ that each family has a positive requirement.
Similarly, we will refer to~$\capl_j[1]$ and~$\caplmin_j[1]$ as the  \myemph{upper} and the \myemph{lower quota} of a place~$\loc_j \in \LOC$, writing also $\capl_j=\capl_j[1]$ and~$\caplmin_j=\caplmin_j[1]$.

\medskip 
The reader may observe that when our sole concern is feasibility, then the problem can be seen as a multidimensional variant of the classic \textsc{Bin Packing} or \textsc{Knapsack} problems. 
On one hand, it is not hard to show that the parameterized hardness of~\textsc{Bin Packing} w.r.t.\ the number of bins as parameter translates to parameterized hardness of \feasible-\RR\ w.r.t.\ the number of places; see \cref{prop:wtbinpacking}. 
On the other hand, the textbook dynamic programming technique for \textsc{Knapsack} 
was used by Gurski et al.~\cite[Proposition~34]{gurski2019knapsack} to solve the so-called \textsc{Max Multiple Knapsack} problem which in our model coincides with the \maxutil-\RR\ problem without lower quotas.
This approach
can be adapted in a straightforward way to solve the \maxutil-\RR\ problem even for the case when there are multiple services and lower quotas;
in \cref{xp:mt} we present an algorithm running in $O((\maxcap)^{mt}nm)$ time.

\begin{restatable}[\appsymb]{proposition}{propwtbinpacking}
\label{prop:wtbinpacking}
The following problems are \Woneh\ w.r.t.~$m$  for $\noser=1$: 
\begin{compactitem}
\item \feasible-\RR;
\item \pareto-\RR\ with no lower quotas and \indifferent\ preferences;
\item \pareto-\RR\ when all families have strict preferences;
\item \maxutil-\RR\ with no lower quotas and \uniform\ utilities;
\item \maxutil-\RR\ with $\finalutil=0$.
\end{compactitem}
\end{restatable}

\appendixproofwithstatement{prop:wtbinpacking}{\propwtbinpacking*}{
We give a reduction from the following variant of \textsc{Bin Packing}. The input instance~$I_{\textup{BP}}$ contains item sizes $a_1,\dots,a_n$ and an integer~$k$, and the question is whether these items can be allocated into $k$ bins such that the total size of items allocated to each bin is exactly $B=(\sum_{i=1}^n a_i) /k$, the bin size. This problem is known to be \Woneh\ w.r.t.~$k$~\cite{JKMS-binpacking-2013}. 

Let us construct an instance~$I_{\textup{FRR}}$ of \feasible-\RR\ with $\noser=1$. We define families~$\fa_1,\dots,\fa_n$ and set the requirement of family~$\fa_i$ as~$a_i$ for each $i \in [n]$. We also define places $\loc_1,\dots,\loc_k$, each with its lower and upper quota both set to~$B$. Then $I_{\textup{FRR}}$ admits a feasible assignment if and only if our \textsc{Bin Packing} instance~$I_{\textup{BP}}$ is solvable; this proves the result for \feasible-\RR.

To obtain a \pareto-\RR\ instance~$I_{\textup{PRR}}$, we can reset the lower quotas to zero, 
and set \indifferent\ preferences for each family. Then $I_{\textup{PRR}}$ admits a feasible and complete assignment if and only if $I_{\textup{BP}}$ admits a solution, 
because the total requirement of the families is $kB$ which equals the total upper quota of the places. It remains to recall that by  \cref{obs:Pareto-vs-completeness}, we can decide whether a feasible and complete assignment exists for~$I_{\textup{PRR}}$ by solving \pareto-\RR\ on~$I_{\textup{PRR}}$. 
This proves the first result for \pareto-\RR.

Next we show that \pareto-\RR\ is \Wone-hard even when all families have strict preferences. Assume, towards a contradiction, that there is an algorithm that can find a Pareto-optimal feasible and acceptable assignment in \FPT\ time w.r.t.\ $m$. 
We show that such an algorithm could be used to solve \feasible-\RR\ in \FPT\ time w.r.t.\ $m$, contradicting the first result unless \FPT\ $=$ \Wone.
Let $I$ be an instance of \feasible-\RR. We reduce it to an instance $I^{\succ}$ of \pareto-\RR\ by giving every family arbitrary complete strict preferences over the places. The instances are otherwise identical. 
Observe that any assignment of $I^{\succ}$ is acceptable, and an assignment $\asg$ is feasible on $I^{\succ}$ if and only if it is feasible on $I$.
If an algorithm for \pareto-\RR\ finds a Pareto-optimal feasible and acceptable assignment $\asg$ for $I^{\succ}$, then $\asg$ is also a feasible assignment of $I$, and we can report $I$ is a yes-instance.
Correspondingly, if an algorithm for \pareto-\RR\ reports there is no feasible and acceptable Pareto-optimal assignment~$\asg$ of~$I^{\succ}$, then there is no feasible and acceptable assignment of $I^{\succ}$, and thus no feasible assignment of $I$, and we can report $I$ is a no-instance.

It is straightforward to check  that by setting unit utilities instead of \indifferent\ preferences we obtain an instance~$I_{\textup{MURR}}$ of \maxutil-\RR\ that is equivalent with $I_{\textup{PRR}}$ in the sense that $I_{\textup{MURR}}$ admits a feasible assignment with utility at least~$kB$ if and only if $I_{\textup{PRR}}$ admits a complete and feasible assignment, proving the fourth statement of the theorem.
Finally, the hardness results for \feasible-\RR\ also hold for \maxutil-\RR\ even if we set all utilities as zero and $\finalutil=0$, proving the last statement. }

In spite of the strong connection between \feasible-\RR\ and \textsc{Bin Packing} (or between \maxutil-\RR\ and \textsc{Knapsack}), the
context of \RefRes\ motivates parameterizations that have not been studied for these two classical problems.
One such parameter is~$\maxreq$, the maximum units of a service that any refugee family may require.
\cref{thm:onetboundedrmaxilp} presents an efficient algorithm for \feasible-\RR\ for the case when $\maxreq$ is small; the proposed algorithm can be used to solve \pareto- and \maxutil-\RR\ as well, assuming \uniform\  preferences or utilities, 
when the task is to assign as many refugee families as possible. 
%
%

Let us introduce an important notion used in our algorithms. 
Let $\blocksize$ denote the least common multiple of all integers in the set $\{1,\dots,\maxreq\}$; then $\blocksize \leq (\maxreq)!$ is clear. 
We say that a set~$\FA ' \subseteq \FA$ of families is a \myemph{\superblock}, if all families in~$\FA '$ have the same requirement, and their total requirement is exactly~$\blocksize$.

\begin{restatable}[\appsymb]{obs}{obssuperblock}
\label{obs:superblock}
Suppose that the number of services is $\noser=1$. 
If $\FA ' \subseteq \FA$ is a set of families such that $\sum_{\fa_i \in \FA'}\reqf_i > \maxreq (\blocksize-1)$, then $\FA'$ contains a \superblock.
\end{restatable} 

\appendixproofwithstatement{obs:superblock}{\obssuperblock*}{
Group families in~$\FA'$ according to their requirements so that families with the same requirement are contained in one group. Clearly, if there is a group  whose total requirement is at least~$\blocksize$, then this group contains a \superblock, because $\blocksize$ is divisible by the (common) requirement of the families in the group. If each group has total requirement at most~$\blocksize-1$, then the total requirement of all families in~$\FA'$ is at most~$\maxreq (\blocksize-1)$, as required.
}


In the case where $\maxreq$ is a constant and the families are indifferent between the places, we can use \cref{obs:superblock} to bound the number of different possible places.
The idea is to observe that while the location capacities are unbounded, we can bound the ``relevant part'' of them by a function of $\maxreq$. If the total requirement of families assigned to a place is more than $\blocksize (\maxreq - 1)$, there must be a \superblock\ among them. We can treat these \superblock s separately from the places they originate from, and thus bound the upper quotas by a function of $\maxreq$. 
The family requirements are also trivially bounded by $\maxreq$.

Since we have now bounded both the maximum requirements of the families and the capacities of the locations, we can enumerate all the different ways families may be matched to places.
We can create an ILP that has a variable for each such way and additionally variables for the \superblock s. As the number of variables and constraints is bounded above by a function of $\maxreq$, we can solve this ILP in \FPT\ time w.r.t.\ $\maxreq$~\cite{lenstra1983integer}.

\begin{theorem}\label{thm:onetboundedrmaxilp}
\pareto- and \maxutil-\RR\ are \FPT\ w.r.t.\ $\maxreq$ when $t = 1$ and families have \indifferent\ preferences or utilities, respectively. 
\end{theorem}
\newcommand{\bigblocksize}{\hat{\blocksize}}
\begin{proof}

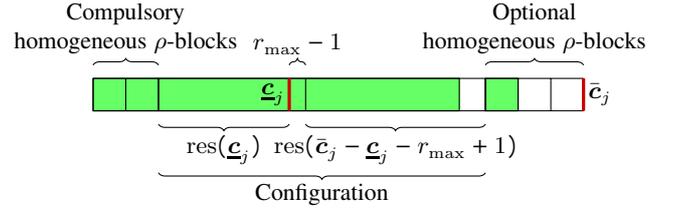
\begin{figure}[t!]
    \centering
    \begin{tikzpicture}[scale=0.43]
    \draw[fill = green!60!white, opacity=0.5] (6,0) -- (17.2,0) -- (17.2,1) -- (6,1)  -- cycle;
    \draw[fill = green!60!white, opacity=0.5] (18,0) -- (19,0) -- (19,1) -- (18,1)  -- cycle;
   		\draw[] (6,1) grid (8,0);
      \draw[] (8,0) -- (18,0) -- (18,1) -- (8,1)  -- cycle;
      \draw[color=red!80!black, very thick] (12, 0) -- (12,1);
      \draw[] (12.5,0) -- (12.5,1);
      \draw[] (18,1) grid (21,0);
     \draw[color=red!80!black, very thick] (21,0) -- (21,1);

     \node[] at (11.5, 0.5) {$\caplmin_j$};
     \node[] at (21.5, 0.5) {$\capl_j$};
     
     \draw[decoration={brace,mirror,raise=5pt},decorate]
  (12.5,0) -- node[below=6pt] {$\res(\capl_j - \caplmin_j - \maxreq + 1)$} (18,0);
  
       \draw[decoration={brace, mirror, raise=5pt},decorate]
  (8,0) -- node[below=6pt] {$\res(\caplmin_j)$} (12,0);
  
       \draw[decoration={brace,raise=5pt},decorate]
  (12.0,1) -- node[above=6pt] {$\maxreq - 1$} (12.5,1);
  
         \draw[decoration={brace,raise=5pt},decorate]
  (6.0,1) -- node[above=6pt, align=center] {Compulsory \\ \superblock s} (8,1);
  
           \draw[decoration={brace,raise=5pt},decorate]
  (18.0,1) -- node[above=6pt, align=center] {Optional \\ \superblock s} (21,1);
  
         \draw[decoration={brace, mirror, raise=5pt},decorate]
  (8,-1.5) -- node[below=6pt] {Configuration} (18,-1.5);

    \end{tikzpicture}
  \caption{Illustration for the proof of \cref{thm:onetboundedrmaxilp}. The bar shows the upper and lower quotas of a place, and the green area represents the requirements of the families matched to it.}     \label{fig:superblock_intuition}
\end{figure}
We start by showing that we can find an assignment that matches the maximum number of families in \FPT\ time w.r.t.\ $\maxreq$. Under \indifferent\ preferences (resp.\ utilities) this assignment must maximize utility (resp.\ be Pareto-optimal).

We start by defining two functions which will be useful for typing places by their service quotas:
the function $\argres \colon \NN \to \NN$
and 
the function $\res \colon \NN \to \{0, \dots, \maxreq ( \blocksize - 1)\}$: 
\begin{linenomath*}
\begin{align}
\notag
\res(x) &= \max_{\alpha \in \mathbb{N}} \{x - \alpha \cdot \blocksize : x - \alpha \cdot \blocksize \leq \maxreq ( \blocksize - 1) \} 
\\
\notag
\argres(x) &= \argmax_{\alpha \in \mathbb{N}} \{x - \alpha \cdot \blocksize : x - \alpha \cdot \blocksize \leq \maxreq ( \blocksize - 1) \}.
\\
\textrm{Observe that }\phantom{m} & 
\forall x \in \mathbb{N}, x = \blocksize \cdot \argres(x) + \res(x).
\label{eq:residue}
\end{align}
\end{linenomath*}
%

Using the $\res$ function we can associate each place with a type~\loctype. Let $\loctype(\loc_j) = (\res(\caplmin_j), x)$, where 
\begin{linenomath*}
\begin{equation*}
x \! = \!
\left\{ \!\!
\begin{array}{l@{\hspace{3pt}}l}
\res( \capl_j - \caplmin_j - \maxreq + \! 1) + \maxreq \! - 1, & \text{if } \capl_j - \caplmin_j \geq \maxreq \! - 1,\\
\capl_j - \caplmin_j, & \text{otherwise.}\end{array}
\right.
\end{equation*}
\end{linenomath*}

The first element of the type tells us the lower quota of the place after we have discounted all the \superblock s that are used to satisfy the lower quota.
The second element tells us the size of the ``optional'' quota $\capl_j - \caplmin_j$ when we have again discounted all the \superblock s that this part may contain.
We however only compute the residue on the part of $\capl_j - \caplmin_j$ that is larger than $\maxreq - 1$. This is because the total requirements of the families that are used to satisfy $\caplmin_j$ may be slightly greater than $\caplmin_j$, and thus there may be a family whose requirement is partially counted for $\capl_j - \caplmin_j$, however it may have been used for a \superblock.
See \cref{fig:superblock_intuition} for intuition of how the lower and upper quotas of a place are divided into the configuration and \superblock s.

Let $\loctypes = \{0, \dots, \maxreq (\blocksize - 1)\} \times \{0, \dots, \maxreq(\blocksize - 1) + \maxreq - 1\}$ be the set of possible place types.
It is clear that their number is bounded above by $\maxreq^2 \blocksize^2$, which is a function of $\maxreq$.
We additionally know that every place must have $\argres(\caplmin_j)$ many \superblock s assigned to it. We call these \emph{compulsory} \superblock s.
Similarly, we may assign at most $\argres(\capl_j - \caplmin_j - \maxreq + 1)$ additional \superblock s to~$\loc_j$, which we call \emph{optional} \superblock s.
Because the families are indifferent between the places, we do not need to keep track of the place the compulsory and optional \superblock s belong to, and we only enforce that the families assigned to compulsory \superblock s create exactly the number of compulsory \superblock s needed, and the families assigned to optional \superblock s create at most the number of optional \superblock s.

Now with a bound on the upper and lower quotas of the places discounting the \superblock s, we can enumerate all possible ways to satisfy these quotas.
Let $\bigblocksize \coloneqq  2\maxreq ( \blocksize  - 1) + \maxreq - 1$. This is the maximum sum of requirements that may be assigned to any place and that are not part of a \superblock. 
We create configurations $\config \in \{0, \dots, \bigblocksize\}^{\maxreq}$ that tell us 
the number of
families of each requirement type. 
Let us denote the set of possible configurations $\configs \coloneqq \{0, \dots, \bigblocksize\}^{\maxreq}$. It is clear that the number of possible configurations is bounded above by $(\bigblocksize + 1)^{\maxreq}$, which is a function of $\maxreq$.

We say that a place type $\loctype$ is \myemph{suitable} for a configuration $\config$ if $\loctype[1] \leq  \sum_{r \in [\maxreq]} \config[r] \cdot r \leq \loctype[1] + \loctype[2]$. 
This means that if the families are assigned to a place according to the configuration, its upper and lower quotas are satisfied.

We create an ILP with the following variables and constants:
\begin{compactitem}
\item non-negative integer variables $\underline{b}_r$ and $\bar{b}_r$ for each $r \in [\maxreq]$, representing the number of 
compulsory or optional, respectively, \superblock s filled with families of requirement~$r$. 
\item non-negative integer variable $x(\loctype, \config)$ for every $\loctype \in \loctypes$, $\config \in \configs$ such that $\loctype$ is suitable for $\config$,
counting the number of places of type $\loctype$ that are assigned families according to configuration~$\config$.
%
\item $m_{\loctype}$ is the number of places of type $\loctype$ 
for each $\loctype \in \loctypes$;
\item $\underline{B}=\sum_{\loc_j \in \LOC} \argres(\caplmin_j)$ (resp.\ $\overline{B}=\sum_{\loc_j \in \LOC} \argres(\max(\capl_j - \caplmin_j - \maxreq + 1), 0)$) is the number of compulsory (resp.\ optional) \superblock s; 
\item $n_r$ is the number of families $\fa_{i} \in \FA$ such that $\reqf_{i} = r$, for each $r \in [\maxreq]$.
\end{compactitem}

We create the following ILP:
\begin{linenomath*}
\begin{equation}
\leqnomode
\tag{ILP$_1$} 
\label{LP1}
\max 
\sum_{r \in [\maxreq]}
\left(\frac{\blocksize}{r}(\underline{b}_r + \bar{b}_r) + 
\sum_{\substack{\loctype \in \loctypes \\ \config \in  \configs}}  
\config[r] x(\loctype, \config)\right) 
 \quad \textrm{s.t.} 
\end{equation}
\begin{align}
\notag \\[-22pt]
& \forall \loctype \in \loctypes \colon
&& \!\!\! \sum_{\substack{\config \in \configs \\  \config \text{ is suitable for } \loctype}}  x(\loctype, \config) = m_{\loctype} \label{const:sbilp1}\\
& && \!\!\! \sum_{r \in [\maxreq]}\underline{b}_r = \underline{B} \quad \text{ and } \sum_{r \in [\maxreq]}\bar{b}_r \leq \overline{B} \label{const:sbilp3}\\
& \forall r \in [\maxreq] \colon && \!\!\!
\frac{\blocksize}{r}(\underline{b}_r + \bar{b}_r) + \sum_{\substack{\loctype \in \loctypes \\ \config \in  \configs}} \config[r] x(\loctype, \config) \leq n_r   \label{const:sbilp4}
\end{align}
\end{linenomath*}

\conref{const:sbilp1} enforces that every place has families matched to it according to some suitable configuration.
\conref{const:sbilp3} enforces that every compulsory \superblock\ is filled with refugee families and that no non-existing optional \superblock s are filled with refugee families.
\conref{const:sbilp4} enforces that for each service-requirement, only available number of refugees are used. 
The objective function formulates the total number of families assigned.

It is clear that the number of variables is bounded above by $2 \maxreq + (\bigblocksize + 1)^{\maxreq} \maxreq^2 \blocksize^2$, and the number of constraints by $3 \maxreq^2 \blocksize^2 + \maxreq$, which are functions of $\maxreq$. Thus the problem can be solved in \FPT\ time w.r.t.\ $\maxreq$~\cite{lenstra1983integer}. 
The correctness of this approach follows from \cref{clm:ILP-rmax-iff}.


\begin{restatable}[\appsymb]{claim}{clmILPrmaxiff}
\label{clm:ILP-rmax-iff}
\ref{LP1} admits a solution
 with value~$\finalutil$ if and only if 
there is an assignment of families with utility~\finalutil. \qedhere
\end{restatable}
\end{proof}

\appendixproofwithstatement{clm:ILP-rmax-iff}{\clmILPrmaxiff*}{
We prove the two directions of the claim separately.

\begin{claim}
\label{clm:ILP-rmax-sufficient}
If \ref{LP1}  admits a solution with value~$\finalutil$, then there is an assignment of families with utility~$\finalutil$.
\end{claim}

\begin{claimproof}
Assume there is a solution to the constructed ILP-instance with value~$\finalutil$. 
We construct an assignment $\asg$. Start with $\asg(\fa_i) = \bot$ for every $\fa_i \in \FA$.

Let us create a counter $\beta_r$ for each $r \in [\maxreq]$ and initialize it by setting $\beta_r \coloneqq \underline{b}_r$. Similarly, let us create a counter $\psi_j$ for each place $\loc_j \in \LOC$ and initialize it by setting $\psi_j \coloneqq \argres(\caplmin_j)$.

\paragraph{Step 1.}
We build $\asg$ as follows: While there exists some $r \in [\maxreq]$ such that $\beta_r > 0$,
\begin{compactenum}[(i)]
\item choose $\loc_j \in \LOC$ such that $\psi_j > 0$,
\item assign a \superblock\ with requirement~$r$ to~$\loc_j$, and
\item update $\beta_r \coloneqq \beta_r - 1$ and $\psi_j \coloneqq \psi_j - 1$.
\end{compactenum}

Since the solution is valid, by \conref{const:sbilp3} we know  that $\sum_{r \in [\maxreq]} \underline{b}_r = \underline{B}= \sum_{\loc_j \in \LOC} \argres(\caplmin_j)$. Therefore at the end of the process, $\sum_{r \in [\maxreq]} \beta_r = \sum_{\loc_j \in \LOC} \psi_j = 0$.

\paragraph{Step 2.}
Next we assign the families that are matched in optional \superblock s.
Let us create a counter $\bar{\beta}_r$ for each $r \in [\maxreq]$ and initialize $\bar{\beta}_r \coloneqq \bar{b}_r$. Similarly, let us create a counter $\bar{\psi}_j$ for each $\loc_j \in \LOC$ and initialize $\bar{\psi}_j \coloneqq \argres(\max(\caplmin - \caplmin_j - \maxreq + 1, 0))$.

We build $\asg$ as follows: While there exists some $r \in [\maxreq]$ such that $\bar{\beta}_r > 0$,
\begin{compactenum}[(i)]
\item choose $\loc_j \in \LOC$ such that $\bar{\psi}_j > 0$,\label{ilp:sol:opt1}
\item assign a \superblock\ with requirement $r$ to $\loc_j$, and
\item update $\bar{\beta}_r \coloneqq \bar{\beta}_r - 1$ and $\bar{\psi}_j \coloneqq \bar{\psi}_j - 1$.
\end{compactenum}

Since the solution is valid, by \conref{const:sbilp3} we know  that  $\sum_{r \in [\maxreq]} \bar{b}_r \leq \overline{B}= \sum_{\loc_j \in \LOC} \argres(\max(\caplmin - \caplmin_j - \maxreq + 1, 0))$. Therefore while $\sum_{r \in [\maxreq]} \bar{\beta}_r \geq 0$, we can always find $\loc_j \in \LOC$ such that $\psi_j > 0$.

\paragraph{Step 3.}
Finally, we look at every place type $\loctype \in \loctypes$. We have  
\begin{linenomath*}
\begin{equation}
\label{eq:step3}
\sum_{\substack{\config \in \configs \\  \config \text{ is suitable for } \loctype}}  x(\loctype, \config) = m_{\loctype}
\end{equation}
\end{linenomath*}
 due to \conref{const:sbilp1}, where $m_{\loctype}$ is the number of places of type~$\loctype$. 
 For each place type~$\loctype$ do the following: for each $\config \in \configs$, assign $x(\loctype, \config)$ many families to places of type~$\loctype$ according to configuration~$\config$:
for each place $\loc_j$ of type~$\loctype$, we assign a set of families having configuration~$\config$ to~$\loc_j$, i.e., for each $r \in [\maxreq]$, we assign $\config[r]$ many families with requirement $r$ to the place~$\loc_j$. By the equality in~(\ref{eq:step3}), every place of type~$\loctype$ will be treated exactly once in this step.

\paragraph{Correctness.}
This concludes the description of the assignment. It remains to show that this assignment is feasible and matches $\finalutil$ many agents.
We start by showing that the quotas of the places are respected. Let $\loc_j \in \LOC$ be an arbitrary place and let $\loctype_j$ be the type of~$\loc_j$. 
In Step~1, we assign $\argres(\caplmin_j)$ many \superblock s to~$\loc_j$. In Step~2, we assign at most $\argres(\max(\caplmin - \caplmin_j - \maxreq + 1, 0))$ many \superblock s to~$\loc_j$.
In Step~3, we choose some $\config \in \configs$ that is suitable for $\loctype_j$ and assign families accordingly. Thus, $\loctype_j[1] \leq R \leq \loctype_j[1] + \loctype_j[2]$ holds for the total requirement~$R$ of the families assigned to $\loc_j$ in Step~3. Thus $R \geq \res(\caplmin_j)$. 

Therefore, the total requirement of the families assigned to $\loc_j$ is at least $\blocksize \cdot \argres(\caplmin_j) + \res(\caplmin_j) \stackrel{\eqref{eq:residue}}{=} \caplmin_j$, as required.

If $\capl - \caplmin_j \geq \maxreq - 1$, then we know 
\begin{linenomath*}
\begin{equation*}
R \leq \res(\caplmin_j) + \res(\capl_j - \caplmin_j - \maxreq + 1) + \maxreq - 1.
\end{equation*}
\end{linenomath*}
 Otherwise, $R \leq \res(\caplmin_j) + \capl_j - \caplmin_j$.

If $\capl - \caplmin_j \geq \maxreq - 1$, then the total requirement of the families assigned to $\loc_j$ is at most 
\begin{linenomath*}
\begin{align*}
\blocksize \cdot \argres(\caplmin_j) &+ \blocksize \cdot \argres(\capl_j - \caplmin_j - \maxreq + 1) +\res(\caplmin_j) \\
& + \res(\capl_j - \caplmin_j - \maxreq + 1) 
+ \maxreq - 1 \\
&\stackrel{\eqref{eq:residue}}{=} \caplmin_j + \capl_j - \caplmin_j - \maxreq + 1 + \maxreq - 1 = \capl_j,
\end{align*}
\end{linenomath*}
 as required.

Otherwise, the total requirement of the families assigned to $\loc_j$ is at most $\blocksize \cdot \argres(\caplmin_j) + 0 + \res(\caplmin_j) + \capl_j - \caplmin_j \stackrel{\eqref{eq:residue}}{=} \caplmin_j + \capl_j - \caplmin_j = \capl_j$, as required.

Next, we show that all the families that are assigned in $\asg$ exist. Let $r \in [\maxreq]$ be an arbitrary service requirement. Whenever we match a \superblock\ of requirement~$r$, we match $\frac{\blocksize}{r}$ many families of this requirement. Therefore in Step~1 we match $\frac{\blocksize}{r}\cdot \underline{b}_r$ families of requirement $i$. Similarly, in Step~2 we match $\frac{\blocksize}{r}\cdot \bar{b}_r$ families of requirement~$r$.
In Step~3 we assign $\config[r] \cdot x(\loctype, \config)$ many families for each $\loctype \in \loctypes, \config \in \configs$.
Thus the total number of families of requirement~$r$ matched is 
\begin{linenomath*}
\begin{align*}
\frac{\blocksize}{r}(\underline{b}_r + \bar{b}_r) + \sum_{\substack{\loctype \in \loctypes \\ \config \in  \configs}} \config[r] x(\loctype, \config).
\end{align*}
\end{linenomath*}
 By \conref{const:sbilp4} this is at most $n_r$, which is the number of families of requirement~$r$.

With the same reasoning, the total number of families matched is 
\begin{linenomath*}
\begin{align*}
\sum_{r \in [\maxreq]}\frac{\blocksize}{r}(\underline{b}_r + \bar{b}_r) + \sum_{\substack{\loctype \in \loctypes \\ \config \in  \configs}} \config[r] x(\loctype, \config),
\end{align*}
\end{linenomath*} 
which is exactly the value of our solution to~\ref{LP1}, and thus equals~$\finalutil$.
\end{claimproof}

\begin{claim}
\label{clm:ILP-rmax-necessary}
If there is a feasible assignment that matches $\finalutil$ many agents, then \ref{LP1}  admits a solution of value~$\finalutil$.
\end{claim}

\begin{claimproof}
Let $\asg$ be a feasible assignment that matches $\finalutil$ many agents.

We start by initializing 
$\underline{b}_r \coloneqq 0$ and $\bar{b}_r \coloneqq 0$ for every $r \in [\maxreq]$. 
For each $r \in [\maxreq]$, we also construct sets $\Ablock_r$, $\Ablockopt_r$ and $\Aconfig{\config}$ for each $\config \in \configs$; each of these sets will contain families with requirement~$r$.  The set $\Ablock_r$ will contain families that are matched to compulsory \superblock s, $\Ablockopt_r$ the families matched to optional \superblock s, and $\Aconfig{\config}$ will contain
families 
that are part of the configurations.

For each place $\loc_j \in \LOC$, choose an arbitrary subset of families $\FA^1_j \subseteq \asg^{-1}(\loc_j)$ such that $\sum_{\fa_{i} \in \FA^1_j}\reqf_{i} \geq \caplmin_j$ and moreover, for each family $\fa_{i^*} \in \FA^1_j $ we have $\sum_{\fa_{i} \in \FA^1_j \setminus \{\fa_{i^*}\} } \reqf_{i} < \caplmin_j$; that is, the families in $\FA^1_j$  satisfy the lower quota of $\loc_j$, and removing any family from~$\FA^1_j$ means the lower quota is no longer satisfied. Since $\asg$ is feasible, such a set must exist.
Let $d \coloneqq \sum_{\fa_{i} \in \FA^1_j}\reqf_{i} - \caplmin_j$; this is the amount by which the total requirement of the families in $\FA_j^1$ exceeds~$\caplmin_j$. Observe that by construction, $0 \leq d \leq \maxreq - 1$.

If $\sum_{\fa_i \in \FA^1_j} \reqf_i - d > \maxreq (\blocksize - 1)$, then we start an iteration as follows. Clearly, $\sum_{\fa_i \in \FA^1_j} \reqf_i > \maxreq (\blocksize - 1)$  holds, so by \cref{obs:superblock}, there must be a \superblock\ $\FA'_j$ contained in~$\FA^1_j$. Let the requirement of the families in $\FA'_j$ be $r \in [\maxreq]$. We create a new subset $\FA^2_j \coloneqq \FA^1_j \setminus \FA'_j$ and increment the value of the variable~$\underline{b}_{r}$ by one.
We also add the families  
in $\FA'_j$ to $\Ablock_{r}$. Observe that there are precisely 
$\frac{\blocksize}{r}$
 families in~$\FA'_j$. We repeat this process until we reach $k$ such that $\sum_{\fa_i \in \FA^k_j} \reqf_i - d \leq \maxreq (\blocksize - 1)$.

Since $\sum_{\fa_i \in \FA^1_j} \reqf_i - d = \caplmin_j$ by construction, we stop when $\sum_{\fa_i \in \FA^k_j} \reqf_i - d \leq \maxreq (\blocksize - 1)$, and in each iteration we remove~$\blocksize$ from the total requirement, we observe that the number~$k$ of iterations  is precisely $\argres(\caplmin_j)$, and $\sum_{\fa_i \in \FA^k_j} \reqf_i - d = \res(\caplmin_j)$. Therefore we increment the variables $\{\underline{b}_r :r \in [\maxreq]\}$ precisely $\sum_{\loc_j \in \LOC}\argres(\caplmin_j)$ times. This shows the first part of \conref{const:sbilp3}.

Similarly, consider $\bar{\FA}^1_j \coloneqq \asg^{-1}(\loc_j) \setminus \FA^1_j$. Intuitively, these families are assigned to $\loc_j$ but they are not necessary for satisfying the lower quota of $\loc_j$.
 
If $\sum_{\fa_i \in \bar{\FA}^1_j} \reqf_i - \maxreq + d + 1> \maxreq (\blocksize - 1)$, then we start a second iteration as follows. Clearly, $\sum_{\fa_i \in \bar{\FA}^1_j} \reqf_i > \maxreq (\blocksize - 1)$ holds, so  by \cref{obs:superblock}, there must be a \superblock\ $\FA'_j$ contained in~$\bar{\FA}^1_j$. Let the requirement of the families in $\FA'_j$ be $r \in [\maxreq]$. We create a new subset $\bar{\FA}^2_j \coloneqq \bar{\FA}_j^1 \setminus \FA'_j$ and increment the value of the variable~$\bar{b}_{r}$ by one. We repeat this process until we reach $\altk$ such that $\sum_{\fa_i \in \bar{\FA}^{\altk}_j} \reqf_i - \maxreq + d + 1\leq \maxreq (\blocksize - 1)$.


Observe that by the definition of~$d$,
\begin{linenomath*}
\begin{align*}
\sum_{\fa_i \in \bar{\FA}^1_j} \reqf_i - \maxreq + d + 1 &= \sum_{\fa_i \in \bar{\FA}^1_j} \reqf_i + \sum_{\fa_i \in \FA^1_j} \reqf_i - \caplmin_j - \maxreq + 1 \\
&\leq \capl_j - \caplmin_j - \maxreq + 1.
\end{align*}
\end{linenomath*}
If $\capl_j - \caplmin_j \geq \maxreq - 1$, then 
\begin{linenomath*}
\begin{align*}
\altk - 1\leq \argres(\capl_j - \caplmin_j - \maxreq + 1), 
\end{align*}
\end{linenomath*}
and thus 
\begin{linenomath*}
\begin{align*}
\sum_{\fa_i \in \bar{\FA}^k_j}\reqf_i - \maxreq + d + 1 \leq \res(\capl_j - \caplmin_j - \maxreq + 1).
\end{align*}
\end{linenomath*}
By contrast, if $\capl_j - \caplmin_j < \maxreq - 1$, then we never remove \superblock s from~$\loc_j$, and $\bar{\FA}^{\altk}_j = \bar{\FA}^1_j$.
Thus, we increment the variables $\{\bar{b}_r :r \in [\maxreq]\}$ throughout the whole process performed for each place $\loc_j \in \LOC$ at most 
\begin{linenomath*}
\begin{align*}
\sum_{\loc_j \in \LOC}\argres(\max(\caplmin_j - \capl_j - \maxreq + 1, 0 ))
\end{align*}
\end{linenomath*}
 times, 
which shows the second part of \conref{const:sbilp3}.

After finishing the two iterations for finding $\FA^k_j$ and $\bar{\FA}^{\altk}_j$, we identify 
the configuration~$\config_j \in \configs$  whose $r$-th coordinate satisfies
\begin{linenomath*}
\begin{align*}
\config_j[r] =|\{\fa_h \in \FA^k_j \cup \bar{\FA}^{\altk}_j : \reqf_h = r \}|
\end{align*}
\end{linenomath*}
 for each $r \in [\maxreq]$.
We then add to $\Aconfig{\config_j}$ the $\config_j[r]$ families in the set $\{\fa_h \in \FA^k_j \cup \bar{\FA}^{\altk}_j : \reqf_h = r \}$.
By construction, $\sum_{\fa_i \in \FA^k_j \cup \bar{\FA}^{\altk}_j} \reqf_i = \sum_{r \in [\maxreq]} \config_j[r] \cdot r$.
Let $\loctype_j$ be the type of $\loc_j$.
We increment the variable $x(\loctype_j, \config_j)$ by one.
We now show that $\loctype_j$ is suitable for $\config_j$.

Recall that $\sum_{\fa_i \in \FA^k_j} \reqf_i - d = \res(\caplmin_j) = \loctype_j[1]$ for each $\loc_j \in \LOC$.
Therefore, 
\begin{linenomath*}
\begin{align*}
\loctype_j[1] = \sum_{\fa_i \in \FA^k_j \cup \bar{\FA}^{\altk}_j} \reqf_i - d \leq \sum_{\fa_i \in \FA^k_j \cup \bar{\FA}^{\altk}_j} \reqf_i \leq  \sum_{r \in [\maxreq]} \config[r] \cdot r,
\end{align*}
\end{linenomath*}
 as required. 
It remains to show that 
\begin{linenomath*}
\begin{equation*}
\sum_{r \in [\maxreq]} \config_j[r] \cdot r \leq \loctype[1] + \loctype[2].
\end{equation*}
\end{linenomath*}

We distinguish between two cases.

\paragraph*{Case 1: $\capl_j - \caplmin_j \geq \maxreq - 1$.}  Then 
\begin{linenomath*}
\begin{equation*}
\sum_{\fa_i \in \bar{\FA}^{\altk}_j}\reqf_i + d  \leq \res(\capl_j - \caplmin_j - \maxreq + 1) + \maxreq - 1 = \loctype[2].
\end{equation*}
\end{linenomath*}
Therefore, 
\begin{linenomath*}
\begin{align*}
\sum_{r \in [\maxreq]} \config[r] \cdot r &= \sum_{\fa_i \in \FA^k_j \cup \bar{\FA}^{\altk}_j} \reqf_i =  \sum_{\fa_i \in \FA^k_j} \reqf_i - d + \sum_{\fa_i \in \bar{\FA}^{\altk}_j} \reqf_i + d  \\
&\leq \res(\caplmin_j) + \res(\capl - \caplmin_j - \maxreq + 1) + \maxreq - 1 
\\
&= \loctype_j[1] + \loctype_j[2],
\end{align*}
\end{linenomath*}
 as required.

\paragraph*{Case 2:  $\capl_j - \caplmin_j < \maxreq - 1$.}
In this case, $\loctype_j[2] = \capl_j - \caplmin_j$. 
Recall that $\asg$ is a valid assignment,  so $\sum_{\fa_i \in \FA^1_j} \reqf_i + \sum_{\fa_i \in \bar{\FA}^1_j} \reqf_i \leq \capl_j$. Therefore we obtain that 
\begin{linenomath*}
\begin{align*}
\sum_{\fa_i \in \bar{\FA}^{\altk}_j} \reqf_i  &= \sum_{\fa_i \in \bar{\FA}^1_j} \reqf_i  \leq \capl_j - \sum_{\fa_i \in \FA^1_j} \reqf_i = \capl_j - \caplmin_j - d  = \loctype_j[2] - d
\end{align*}
\end{linenomath*}
which is equivalent to
\begin{linenomath*}
\begin{equation*}
 \sum_{\fa_i \in \bar{\FA}^{\altk}_j} \reqf_i + d \leq \loctype_j[2].
\end{equation*}
\end{linenomath*} 
 Thus 
\begin{linenomath*}
\begin{align*}
\sum_{r \in [\maxreq]} \config[r] \cdot r & = \sum_{\fa_i \in \FA^k_j \cup \bar{\FA}^{\altk}_j} \reqf_i = \sum_{\fa_i \in \FA^k_j} \reqf_i - d + \sum_{\fa_i \in \bar{\FA}^{\altk}_j} \reqf_i + d 
\\
&\leq \loctype_j[1] + \loctype_j[2],
\end{align*}
\end{linenomath*}
 as required.

For each place type $\loctype$, every place of type~$\loctype$ increments a single variable in $\{x(\loctype, \config) : \config \in \configs, \config \text{ is suitable for } \loctype\}$. Thus $\sum  \{ x(\loctype, \config): \config \in \configs,  \config \text{ is suitable for } \loctype\} = m_{\loctype}$, which proves \conref{const:sbilp1}.

Next, let us prove \conref{const:sbilp4}. Let $r \in [\maxreq]$ be a service requirement. Observe that the sets $\Ablock_r, \Ablockopt_r$, and $ \Aconfig{\config}$ for $\config \in \configs$ are all subsets of $\{\fa_h \in \FA :\reqf_h = r\}$ as only families of requirement~$r$ are added to these sets. Moreover, they are pairwise disjoint, as any family is added to at most one set. Finally, the set of families matched to some place under~$\asg$ is precisely $\Ablock_r \cup \Ablockopt_r \cup \bigcup_{\config \in \configs} \Aconfig{\config}$, as every family that is matched to a place is added to one of these sets.

Each time we incremented variable $\underline{b}_r$ by one, we added $\frac{\blocksize}{r}$ families to $\Ablock_r$. Thus $|\Ablock_r| = \frac{\blocksize}{r} \underline{b}_r$.
Similarly, each time we incremented variable $\bar{b}_r$ by one, we added $\frac{\blocksize}{r}$ families to $\Ablockopt_r$. Thus $|\Ablockopt_r| = \frac{\blocksize}{r} \bar{b}_r$.
Finally, for every $\config \in \configs$, anytime we incremented a variable in~$\{x(\loctype, \config) : \loctype \in \loctypes\}$ by one, we added $\config[r]$ families to~$\Aconfig{\config}$.
Thus $|\Aconfig{\config}| = \sum_{\loctype \in \loctypes}\config[r]x(\loctype, \config)$.

Therefore we can conclude that the number of families of requirement $r$ matched under $\asg$ is precisely 
\begin{linenomath*}
\begin{align*}
|\Ablock_r| + |\Ablockopt_r| & + \sum_{\config \in \configs}|\Aconfig{\config}| \\
&=\frac{\blocksize}{r} \underline{b}_r + \frac{\blocksize}{r} \bar{b}_r + \sum_{\config \in \configs}\sum_{\loctype \in \loctypes}\config[r] x(\loctype, \config) \\
&= \frac{\blocksize}{r}(\underline{b}_r + \bar{b}_r) + \sum_{\substack{\loctype \in \loctypes \\ \config \in  \configs}} \config[r] x(\loctype, \config).
\end{align*}
\end{linenomath*}
Since there are $n_r$ families of requirement~$r$, we get 
\begin{linenomath*}
\begin{equation*}
\frac{\blocksize}{r}(\underline{b}_r + \bar{b}_r) + \sum_{\substack{\loctype \in \loctypes \\ \config \in  \configs}} \config[r] x(\loctype, \config) \leq n_r.
\end{equation*}
\end{linenomath*} 
This proves \conref{const:sbilp4}.

We deduce that the number of families matched under $\asg$ is 
\begin{linenomath*}
\begin{equation*}
\sum_{r \in [\maxreq]}\frac{\blocksize}{r}(\underline{b}_r + \bar{b}_r) + \sum_{\substack{\loctype \in \loctypes \\ \config \in  \configs}} \config[r] x(\loctype, \config).
\end{equation*}
\end{linenomath*} 
Because $\asg$ matches $\finalutil$ many families, 
the value of the constructed solution for~\ref{LP1} is exactly~$\finalutil$.
\end{claimproof}

\cref{clm:ILP-rmax-sufficient,clm:ILP-rmax-necessary} together prove the correctness of our ILP.
}

When preferences or utilities are not \identical,  parameterization by~$\maxreq$ alone does not yield fixed-parameter tractability: as established by \cref{thm:satwithties}, the case $\maxreq=\maxcap=2$ is \NPh\ even in a very restricted case, when there are no lower quotas, families have dichotomous preferences (or binary utilities), and each family finds at most two places acceptable (or of positive utility).

Let us remark that a slightly weaker result (the statement without the condition that each family finds exactly two places acceptable, or has positive utility for exactly two paces) follows via a fairly straightforward reduction from the \textsc{Matching with Couples} problem~\cite{biro-mcdermid-matching-sizes,glass-kellerer-scheduling}.

\begin{restatable}[\appsymb]{theorem}{thmsatwithties}
\label{thm:satwithties}
\pareto-\RR\ and \maxutil-\RR\ for $ \noser = 1$ are \NPh\ even when $\maxreq = \maxcap = 2$ and there are no lower quotas.
The result holds for \pareto-\RR\ even if all families have \dichotomous\ preferences and find exactly two places acceptable,
and for \maxutil-\RR\ even if utilities are binary and each family has positive utilities for exactly two places.
\end{restatable}

\appendixproofwithstatement{thm:satwithties}{\thmsatwithties*}{

Here we present a reduction from the variant of 3-SAT where each literal appears exactly twice. Let our input be a formula $\varphi=\bigwedge_{C \in \mathcal{C}} C$ over a set~$X$ of variables where both literals $x$ and $\overline{x}$ appear exactly twice in~$\varphi$ for each $x \in X$, 
and each clause~$C \in \mathcal{C}$ contains exactly three distinct literals.
This problem is \NPh\
\cite{BKS-2bal3sat-2003}.
We construct an instance~$I$ of \pareto-\RR\ with $\noser=1$ and without lower quotas as follows.

We set $\{\fa_x,x^1,x^2,\overline{x}^1,\overline{x}^2:x \in X\}$ as the set of families; we set the requirement of all families in $\FA_X:=\{\fa_x:x \in X\}$ as two, and the requirement of all remaining families as one. 
We introduce two places, $\loc_{x}$ and~$\loc_{\overline{x}}$, for each variable $x \in X$ 
and a place~$\loc_C$ for each clause~$C \in \mathcal{C}$. 
We set the upper quota of every place as two.
Note that $\maxreq=\maxcap=2$, as promised.

We next define the set of acceptable places for each family. 
First, for each literal~$\ell$ (that is, for each $\ell \in \{x,\overline{x}:x \in X\}$), the family $\ell^h$ for some $h \in [2]$ corresponds to the $h$-th occurrence of the literal~$\ell$, and hence finds two places acceptable:	$\loc_\ell$ and $\loc_{C(\ell,h)}$ where $C(\ell,h)$ is the clause containing the $h$-th occurrence of literal~$\ell$. 
Second, for each variable~$x \in X$, the family~$\fa_x$ finds $\loc_x$ and $\loc_{\overline{x}}$ acceptable. 
We set \dichotomous\ preferences for all families, so each family is indifferent between the two places they find acceptable. 

\begin{claim}
\label{clm:3sat-equivalence}
$I$ admits a feasible, acceptable and complete assignment 
if and only if $\varphi$ is satisfiable. 
\end{claim}
\begin{claimproof}
Suppose first that $\varphi$ admits a satisfying truth assignment. 
We define an assignment~$\asg$ for~$I$ as follows.
For each variable~$x \in X$ set to \true, 
the family~$\fa_x$ is assigned to the place~$\loc_{\overline{x}}$, 
the two families $\overline{x}^1$ and~$\overline{x}^2$ that correspond to false literals are assigned to the places~$\loc_{C(\overline{x},1)}$ and~$\loc_{C(\overline{x},2)}$, respectively, 
while the two families $x^1$ and~$x^2$ that correspond to true literals are assigned to~$\loc_x$.
Similarly, 
For each variable~$y \in X$ set to \false, 
the family~$\fa_y$ is assigned to the place~$\loc_y$, 
the two families $y^1$ and~$y^2$ that correspond to false literals are assigned to the places~$\loc_{C(y,1)}$ and~$\loc_{C(y,2)}$, respectively, 
while the two families $\overline{y}^1$ and~$\overline{y}^2$ that correspond to true literals are assigned to~$\loc_{\overline{y}}$.
Observe that $\asg$ never assigns families with a total requirement more than two to a place,
because each clause contains at most two literals that evaluate to \false, and thus each place~$\loc_C$, $C \in \mathcal{C}$ accommodates at most two families under~$\asg$, both with requirement one.
Therefore, $\asg$ is feasible, and it is straightforward to verify that it is complete and acceptable as well. 

Suppose now that $\asg$ is an assignment for~$I$ that is feasible, acceptable, and complete. 
We set each variable~$x \in X$ to \true\ if and only if $\asg(\fa_x)=\loc_{\overline{x}}$.
In other words, we set a literal~$\ell$ to \false\ exactly if the place~$\loc_\ell$ is occupied by a family in~$\FA_X$. To show that each clause is satisfied, let us assume for the sake of contradiction that we set all three literals in some clause $C$ to \false. This means that for each literal~$\ell$ appearing in~$C$, 
the place~$\loc_\ell$ accommodates a family in~$\FA_X$. Since each such family has requirement two, which equals the upper quota of~$\loc_\ell$, we get that $\asg$ must assign the families 
$\ell^1$ and $\ell^2$ elsewhere due to its feasibility. Since $\asg$ respects acceptability and is complete, it must assign $\ell^1$ and~$\ell^2$ to the clauses $C(\ell,1)$ and $C(\ell,2)$, respectively. 
In particular, the three families that correspond to the literals appearing in~$C$ (that is, the families $\ell_1^{h_1},\ell_2^{h_2},$ and $\ell_3^{h_3}$ for which $C(\ell_j,h_j)=C$ for $j \in [3]$) 
can only be assigned to~$\loc_C$, contradicting the feasibility of~$\asg$.
This shows that the constructed truth assignment satisfies~$\varphi$, and hence, the claim holds.
\end{claimproof}
The result for \pareto-\RR\ follows from \cref{clm:3sat-equivalence} due to \cref{obs:Pareto-vs-completeness}, and from that, the \NPh{}ness for \maxutil-\RR\ follows from \cref{rem:maxutil-vs-pareto}.
}

To tackle the computational intractability of \cref{thm:satwithties}, we focus on the parameter~$m+\maxreq$ and propose an FPT algorithm with this parameterization in \cref{thm:onetboundedrmax+m}.

To prove \cref{thm:onetboundedrmax+m},
we are going to present an algorithm for \maxutil-\RR\ that constructs a feasible assignment with maximum utility, or concludes that no feasible assignment exists.
By \cref{rem:maxutil-vs-pareto}, such an algorithm can be used to solve the \pareto-\RR\ problem as well.
Let $I$ denote our input instance of \maxutil-\RR.

We use a two-phase dynamic programming approach based on the following key idea:
once we have obtained an optimal assignment $\asg$ for a partial instance $\mathcal{J}$, then a small modifications to this partial instance results in an instance~$\mathcal{J}'$ that admits an optimal assignment~$\asg'$ that is ``close'' to~$\asg$. By guessing  how~$\asg'$ differs from~$\asg$, we can compute $\asg'$ efficiently. Let us give a high-level view of our algorithm.

 In the first phase, we disregard lower quotas, and starting from an instance with only a single family, we add families one by one.  
For each $i \in [n]$, let $\FA_i=\{\fa_1,\dots,\fa_i\}$ denote the set of the first~$i$ families, and $I_i$ the instance obtained by restricting~$I$ to~$F_i$ and setting all lower quotas as zero. Starting from a maximum-utility feasible assignment~$\asg_1$ for~$I_1$, we construct a maximum-utility feasible assignment~$\asg_i$ for $i=2,\dots,n$ by slightly modifying the assignment~$\asg_{i-1}$.

In the second phase, starting from the instance~$\hat I_0=I_n$ without lower quotas, we define a sequence $\hat I_1,\dots,\hat I_{\sumcapmins}$ of instances where each instance is obtained from the previous one by raising the lower quota of a single place by one in an arbitrary way so long as 
the lower quotas for~$I$ are not exceeded; notice that this implies $I=\hat I_{\sumcapmins}$ where $\sumcapmins=\sum_{\loc_j \in \LOC} \caplmin_j$. 
Then starting from~${\hat\asg_0:=\asg_n}$ we compute a maximum-utility feasible assignment~$\hat\asg_q$ for $q=1,\dots,\sumcapmins$ from the assignment~$\hat\asg_{q-1}$ by applying small modifications.

\begin{theorem}
\label{thm:onetboundedrmax+m}
\maxutil- and \pareto-\RR\ for $ \noser = 1$ are \FPT\ w.r.t.~$m+\maxreq$. 
\end{theorem}

\begin{proof}
We may assume w.l.o.g. that there exists a feasible assignment for~$I_i$ with maximum utility that is complete. Indeed, to ensure this for each $i \in [n]$, we can simply create a dummy place whose upper quota is $\sum_{\fa_h \in \FA} \reqf_h$ and towards which all families have zero utility; this also shows that we can assume $m \geq 2$.
For brevity's sake, we say that an assignment is \myemph{optimal} if it is feasible and complete, and has maximum utility among all feasible assignments.

Let $\blockbig=m^m \cdot \blocksize \cdot \maxreq$.
Notice that since~$\blocksize$ is a function of~$\maxreq$, we know that $\blockbig$ is a function of~$m$ and~$\maxreq$ only. 

The first phase of our algorithm relies on \cref{clm:DPbyfamilies-bounded-diff}, which proves that given an optimal assignment~$\asg_i$
for~$I_i$ for some ${i \in [n-1]}$, we can obtain an optimal assignment for~$I_{i+1}$ whose distance from~$\asg_i$ is bounded by a function of~$m$ and $\maxreq$.
We measure the distance of two assignments~$\asg$ and~$\asg'$ as the total requirement of all families that are assigned to different places by~$\asg$ and~$\asg'$, that is,
\begin{linenomath*}
\begin{equation*}
\Delta(\asg,\asg')=\sum \left\{ \reqf_h\colon \fa_h \in \FA_{\cap}, \asg(\fa_h)\neq \asg'(\fa_h)\right\},
\end{equation*}
\end{linenomath*} 
where $\FA_\cap$ is the intersection of the domains of~$\asg$ and~$\asg'$.
\begin{restatable}[\appsymb]{claim}{clmDPbyfamiliesboundeddiff}
\label{clm:DPbyfamilies-bounded-diff}
Suppose that $i \in [n-1]$, and let $\asg_i:\FA_i \rightarrow \LOC$ denote an optimal allocation for~$I_i$. Then there exists an optimal allocation~$\asg_{i+1}:\FA_{i+1} \rightarrow \LOC$ for~$I_{i+1}$ such that
$
\Delta(\asg_i,\asg_{i+1})
\leq  m \cdot \blockbig.
$
\end{restatable}

\appendixproofwithstatement{clm:DPbyfamilies-bounded-diff}{\clmDPbyfamiliesboundeddiff*}{
It will be useful for us to define an integer~$\blockbnd_j$ recursively for each $j \in [m]$ by setting 
\begin{linenomath*}
\begin{equation}
\label{eqn:def-blockbnd}
\blockbnd_j=\left\{\begin{array}{ll}
\maxreq (\blocksize-1) & \textrm{ for } j=1 \\
\blockbnd_{j-1} \cdot (m-1) + \blocksize+\maxreq 
& \textrm{ for } j=2,\dots,m.
\end{array}
\right. 
\end{equation}
\end{linenomath*}
It is straightforward to verify by simple calculus that $\blockbnd_m \leq \blockbig$ (relying also on our assumption that $m \geq 2$).

Let $\asg_{i+1}$ be an optimal allocation for $I_{i+1}$ that minimizes $\Delta(\asg_i,\asg_{i+1})$. 
Consider the 
the movement of families when
the allocation changes from~$\asg_i$ to~$\asg_{i+1}$; we will refer to this as \emph{the relocation}.
For two distinct places~$\loc_j$ and~$\loc_{j'}$, let $d(\loc_j,\loc_{j'})$ be the total requirement of all families moving from~$\loc_j$ into~$\loc_{j'}$ under the relocation.
Define also $d^{\mvin}(\loc_j):=
\sum_{\loc_{j'} \in \LOC \setminus \{\loc_j\}} d(\loc_{j'},\loc_j)$ as  the total requirement of families in~$\FA_i$ moving into~$\loc_j$ under the relocation.
As both~$\asg_i$ and~$\asg_{i+1}$ are complete,
we know $\Delta(\asg_i,\asg_{i+1}) 
= 
 \sum_{\loc_j \in \LOC} d^{\mvin}(\loc_j) 
$. 

We will prove that $d^{\mvin}(\loc_j) \leq \blockbnd_m$ for each $\loc_j \in \LOC$. From this
$\Delta(\asg_i,\asg_{i+1}) =  \sum_{\loc_j \in \LOC} d^{\mvin}(\loc_j) \leq m \blockbnd_m \leq m \blockbig$ follows, proving the claim. 
Assume for the sake of contradiction that there exists some~$\loc_{j_0} \in \LOC$ with $d^{\mvin}(\loc_{j_0})> \blockbnd_m$.

\mypara{An acyclic auxiliary graph $G_\blocksize$.}
We proceed with defining an auxiliary digraph~$G_\blocksize$ defined over~$\LOC$ in which $(\loc_j,\loc_{j'})$ is an arc if and only if $d(\loc_j,\loc_{j'}) > \blockbnd_1$; recall that $\blockbnd_1=\maxreq (\blocksize-1)$. 
Notice that by \cref{obs:superblock}, for each arc~$e$ in~$G_\blocksize$ we know that the set of families moving from the tail of~$e$ to the head of~$e$ under the relocation contains a \superblock; let $B_e$ be such a \superblock\ corresponding to~$e$.

We claim that $G_\blocksize$ is acyclic. Suppose for the sake of contradiction that a set~$C$ of arcs forms a directed cycle  in~$G_\blocksize$. 
Consider the assignment $\asg'_{i+1}$ for~$I_{i+1}$ obtained from~$\asg_{i+1}$ by moving all families contained in 
$\FA_C:=\cup_{e \in C} B_e$ to the place they were assigned under~$\asg_i$, i.e., ``backward'' along the cycle~$C$.
Since all \superblock{}s have the same total requirement~$\blocksize$, the assignment $\asg'_{i+1}$ is feasible for~$I_{i+1}$. 
Moreover, since $\Delta(\asg_i,\asg'_{i+1}) < \Delta(\asg_i,\asg_{i+1})$, we know that $\util(\asg'_{i+1})<\util(\asg_{i+1})$, due to our choice of~$\asg_{i+1}$. This means
\begin{linenomath*} 
\begin{equation}
\label{eq:moving-in-cycle}
\sum_{\fa_h \in \FA_C,\loc_j=\asg_i(\fa_h)} \profit_h[j] < 
\sum_{\fa_h \in \FA_C,\loc_j = \asg_{i+1}(\fa_h)} \profit_h[j].
\end{equation}
\end{linenomath*} 

Let us now construct an assignment~$\asg'_i$ from $\asg_i$ by moving all families contained in~$\FA_C$ to the place they are assigned under~$\asg_{i+1}$, i.e., ``forward'' along the cycle~$C$. 
Again, $\asg'_i$ is feasible for~$I_i$. 
Moreover, by (\ref{eq:moving-in-cycle}) we know that $\util(\asg'_i)>\util(\asg_i)$, which contradicts the optimality of~$\asg_i$. We obtain that $G_\blocksize$ is indeed acyclic.

\mypara{Finding a path~$\chainpath$ in $G_\blocksize$.}
Let $\LOC_{\dashv}$ denote the set of places with \emph{free capacity} at least~$\blocksize$ under~$\asg_i$, i.e. $\LOC_{\dashv}=\{\loc_j: \capl_j \geq \blocksize + \satur(\loc_j,\asg_i)\}$.
We define $\LOC_{\vdash}$ analogously, to contain places with free capacity at least $\blocksize$ under~$\asg_{i+1}$, so $\LOC_{\vdash}=\{\loc_j: \capl_j \geq \blocksize + \satur(\loc_j,\asg_{i+1})\}$. We are going to construct a path  from~$\LOC_{\vdash}$ to~$\LOC_{\dashv}$ in~$G_\blocksize$.

Define a \emph{chain} from~$\loc_{j_0}$ to~$\LOC_{\dashv}$ as a sequence $\loc_{j_0},\loc_{j_1},\dots,\loc_{j_{\altsubm}}$ such that (i)
$\loc_{j_{\altsubm}} \in \LOC_{\dashv}$, and 
(ii) $d(\loc_{j_{h-1}},\loc_{j_h})> \blockbnd_{m-h}$ for each integer~$h$ with $0<h \leq \,\, \altm$. 
We build such a chain by induction. We start from the sequence containing only~$\loc_{j_0}$, and maintain condition~(ii) as an invariant; note that (ii) holds trivially for the sequence $\loc_{j_0}$. 
So suppose that we have already built a sequence 
$\loc_{j_0},\loc_{j_1},\dots,\loc_{j_h}$ for which condition~(ii) holds.
If $\loc_{j_h} \in \LOC_{\dashv}$, then we are done, as the sequence fulfills condition~(i) as well, and thus is a chain. 
If $\loc_{j_h} \notin \LOC_{\dashv}$ then,
by the definition of~$\LOC_{\dashv}$, the free capacity of~$\loc_{j_h}$ under~$\asg_i$ is less than~$\blocksize$. 
This implies that the families moving out from~$\loc_{j_h}$ under the relocation must have total requirement greater than $d^{\mvin}(\loc_{j_h})-\blocksize$ (as otherwise the families moving into~$\loc_{j_h}$ during the relocation would not fit). 
Since these families must move into some place in~$\LOC \setminus \{\loc_{j_h}\}$, there must exist some~$\loc_{j_{h+1}} \in \LOC \setminus \{\loc_{j_h}\}$ for which 
\begin{linenomath*} 
\begin{equation*}
d(\loc_{j_h},\loc_{j_{h+1}}) \geq \frac{d^{\mvin}(\loc_{j_h})-\blocksize}{m-1}
>\frac{\blockbnd_{m-h}-\blocksize}{m-1}
>
\blockbnd_{m-h-1}
\end{equation*}
\end{linenomath*} 
where the second inequality follows from condition~(ii) if $h>0$, and from the assumption that $d^{\mvin}(\loc_{j_0}) > \blockbnd_m$ in the case $h=0$;
the third equality follows from the definition of~$\blockbnd_{m-h}$ which satisfies
\begin{linenomath*} 
\begin{equation*}
\frac{\blockbnd_{m-h}-\blocksize-\maxreq}{m-1} =
\blockbnd_{m-h-1}.
\end{equation*}
\end{linenomath*} 
Hence, we can pick $\loc_{j_{h+1}}$ as the next place in the chain, since the sequence $\loc_{j_0},\loc_{j_1},\dots,\loc_{j_h},\loc_{j_{h+1}}$ fulfills condition~(ii). Since $G_\blocksize$ is acyclic and $|\LOC|$ is finite, the existence of a chain from~$\loc_{j_0}$ to~$\LOC_{\dashv}$ follows.

Next, we similarly build a \emph{back-chain} from~$\LOC_{\vdash}$ to~$\loc_{j'_0}=\loc_{j_0}$, which is defined as a sequence $\loc_{j'_{\altsubmprime}},\dots,\loc_{j'_1},\loc_{j'_0}$ such that (i')
$\loc_{j'_{\altsubmprime}} \in \LOC_{\vdash}$, and 
(ii') $d(\loc_{j'_h},\loc_{j'_{h-1}})>\blockbnd_{m-h}$ for each integer~$h$ with $0<h \leq {\altm}'$. 
Because $d^{\mvin}(\loc_{j_0}) > \blockbnd_m$, we know that there exists some place~$\loc_{j'_1}$ for which $d(\loc_{j'_1},\loc_{j_0}) \geq d^{\mvin}(\loc_{j_0})/(m-1)>\blockbnd_m/(m-1)> \blockbnd_{m-1}$. We build our back-chain starting from the sequence $\loc_{j'_1},\loc_{j_0}$ by induction, using the same technique we applied to build our chain in the previous paragraph. 

Namely, suppose that we already have a sequence 
$\loc_{j'_h},\dots,\loc_{j'_1},\loc_{j_0}$ for which condition~(ii') holds.
If $\loc_{j'_h} \in \LOC_{\vdash}$, then we are done, as the sequence fulfills condition~(i') as well, and thus is a back-chain. 
If $\loc_{j'_h} \notin \LOC_{\vdash}$ then, by the definition of~$\LOC_{\vdash}$, the free capacity of~$\loc_{j'_h}$ under~$\asg_{i+1}$ is less than~$\blocksize$. 
However, $d(\loc_{j'_h},\loc_{j'_{h-1}})>\blockbnd_{m-h}$ by condition~(ii'); 
therefore, taking into account that $\fa_{i+1}$ might be assigned to~$\loc_{j'_h}$ by~$\asg_{i+1}$, we obtain that 
the total requirement of families moving into~$\loc_{j'_h}$ under the relocation (recall that this excludes~$\fa_{i+1}$) is 
$d^{\mvin}(\loc_{j'_h}) > d(\loc_{j'_h},\loc_{j'_{h-1}})-\blocksize -\reqf_{i+1}> \blockbnd_{m-h}-\blocksize-\maxreq$.
Since these families must have come from some place in~$\LOC \setminus \{\loc_{j'_h}\}$, we know that there exists some place~$\loc_{j'_{h+1}} \in \LOC$ for which 
\begin{linenomath*} 
\begin{equation*}
d(\loc_{j'_{h+1}},\loc_{j'_h}) \geq \frac{d^{\mvin}(\loc_{j'_h})}{m-1}>\frac{\blockbnd_{m-h}-\blocksize-\maxreq}{m-1}=\blockbnd_{m-h-1}.
\end{equation*}
\end{linenomath*} 
Hence, we can pick $\loc_{j'_{h+1}}$ as the next place in the back-chain, because the sequence $\loc_{j'_{h+1}},\loc_{j'_h},\dots,\loc_{j'_1},\loc_{j_0}$ fulfills condition~(ii'). 
Since $G_\blocksize$ is acyclic and $|\LOC|$ is finite, the existence of a back-chain from~$\LOC_{\vdash}$ to~$\loc_{j_0}$ follows.

Consider the sequence~$\chainpath$ of places obtained by concatenating our back-chain from~$\LOC_{\vdash}$ to~$\loc_{j_0}$ with the chain from~$\loc_{j_0}$ to~$\LOC_{\dashv}$. Observe that by conditions~(ii) and~(ii'), there is an arc in~$G_\blocksize$ from each place in~$\chainpath$ to the next place in~$\chainpath$. Since $G_\blocksize$ is acyclic, this means that~$\chainpath$ forms a path in~$G_\blocksize$.

\mypara{The contradiction implied by our path~$\chainpath$.}
It remains to show that the existence of our path~$\chainpath$ from~$\LOC_{\vdash}$ to~$\LOC_{\dashv}$ in~$G_\blocksize$ leads to a contradiction. Let $E(\chainpath)$ denote the set of arcs on this path, and define $\FA_{\chainpath}=\bigcup_{e \in E(\chainpath)} B_e$, that is, $\FA_{\chainpath}$ is the union of \superblock{}s corresponding to the arcs on~$\chainpath$. 
Define the assignment $\asg'_{i+1}$ for~$I_{i+1}$ obtained from~$\asg_{i+1}$ by moving all families contained in 
$\FA_{\chainpath}$ to the place they were assigned under~$\asg_i$, i.e., ``backward'' along the path~$\chainpath$.
Since all \superblock{}s have the same total requirement~$\blocksize$ and, in addition, the first place on~$\chainpath$ belongs to~$\LOC_{\vdash}$ and thus has free capacity at least~$\blocksize$ under~$\asg_{i+1}$ that can be used to accommodate the superblock corresponding to the first arc of~$\chainpath$, we get that $\asg'_{i+1}$ is feasible for~$I_{i+1}$. 
Moreover, since $\Delta(\asg_i,\asg'_{i+1}) < \Delta(\asg_i,\asg_{i+1})$, we know that $\util(\asg'_{i+1})<\util(\asg_{i+1})$ due to our choice of~$\asg_{i+1}$, which yields
\begin{linenomath*}
\begin{equation}
\label{eq:moving-along-path}
\sum_{\fa_h \in \FA_{\chainpath},\loc_j = \asg_i(\fa_h)} \profit_h[j] < 
\sum_{\fa_h \in \FA_{\chainpath},\loc_j = \asg_{i+1}(\fa_h)} \profit_h[j].
\end{equation}
\end{linenomath*}

Let us now construct an assignment~$\asg'_i$ from $\asg_i$ by moving all families contained in~$\FA_{\chainpath}$ to the place they are assigned under~$\asg_{i+1}$, i.e., ``forward'' along the path~$\chainpath$. 
Again, $\asg'_i$ is feasible for~$I_i$, because the last place on~$\chainpath$ belongs to~$\LOC_{\dashv}$. 
Due to~(\ref{eq:moving-along-path}),
$\util(\asg'_i)>\util(\asg_i)$, which contradicts to the optimality of~$\asg$. This contradiction proves the claim.
}

The second phase of our algorithm relies \cref{clm:DPbyfamilies-bounded-diff-2} which is an analog of \cref{clm:DPbyfamilies-bounded-diff} with a quite similar proof.
\ifshort
\else The proofs of \cref{clm:DPbyfamilies-bounded-diff,clm:DPbyfamilies-bounded-diff-2} can be found in Appendices~\ref{proof:clm:DPbyfamilies-bounded-diff} and~\ref{proof:clm:DPbyfamilies-bounded-diff-2}.
Recall that $I_i$ and $\hat I_q$ are defined above.
\fi

\begin{restatable}[\appsymb]{claim}{clmDPbyfamiliesboundeddifftwo}
\label{clm:DPbyfamilies-bounded-diff-2}
Suppose that $q \in [\sumcapmins]$, and let $\hat\asg_q:\FA \rightarrow \LOC$ denote an optimal allocation for~$\hat I_q$. Then there exists an optimal allocation~$\hat\asg_{q+1}:\FA \rightarrow \LOC$    for~$\hat I_{q+1}$ such that 
$
\Delta(\hat \asg_q,\hat\asg_{q+1})
 \leq m \cdot \blockbig.
$
\end{restatable}

\appendixproofwithstatement{clm:DPbyfamilies-bounded-diff-2}{\clmDPbyfamiliesboundeddifftwo*}{
The proof is similar to the proof of \cref{clm:DPbyfamilies-bounded-diff}.
Let $\caplmin_j$ and $\caplmin'_j$ denote the lower quotas given for some location~$\loc_j \in \LOC$ in~$\hat{I}_q$ and in~$\hat{I}_{q+1}$, respectively; then $\caplmin_j=\caplmin'_j$ holds for each place~$\loc_j$ but one, and 
$\sum_{j \in [m]} \caplmin'_j=
1+\sum_{j \in [m]} \caplmin_j$.
Let $\hat \asg_{q+1}$ be an optimal allocation for~$\hat I_{q+1}$ that minimizes $\Delta(\hat\asg_q,\hat\asg_{q+1})$. 
Consider the situation where the allocation changes from~$\hat\asg_q$ to~$\hat\asg_{q+1}$; we will refer to this as \emph{the relocation}.

For distinct places~$\loc_j$ and~$\loc_{j'}$, let $\hat d(\loc_j,\loc_{j'})$ be the total requirement of all families moving from~$\loc_j$ into~$\loc_{j'}$ under the relocation.
Define also the values $\hat d^{\mvin}(\loc_j):=\sum_{\loc_{j'} \in \LOC \setminus \{\loc_j\}} \hat d(\loc_{j'},\loc_j)$
and 
$\hat d^{\mvout}(\loc_j):=\sum_{\loc_{j'} \in \LOC \setminus \{\loc_j\}} \hat d(\loc_j,\loc_{j'})$
as the total requirement of families moving into~$\loc_j$ and out of~$\loc_j$, respectively, under the relocation.
As both~$\hat\asg_q$ and~$\hat\asg_{q+1}$ are complete assignments, we know that $\Delta(\hat\asg_q,\hat\asg_{q+1}) 
= \sum_{\loc_j \in \LOC} \hat d^{\mvin}(\loc_j) 
= \sum_{\loc_j \in \LOC} \hat d^{\mvout}(\loc_j) 
$. 

Recall the definition of values~$\blockbnd_h$ for $h \in [m]$ as defined by~\ref{eqn:def-blockbnd}.
We will prove that $\hat d^{\mvin}(\loc_j) \leq \blockbnd_m$ for each $\loc_j \in \LOC$. From this
$\Delta(\hat\asg_q,\hat\asg_{q+1}) = \sum_{\loc_j \in \LOC} \hat d^{\mvin}(\loc_j) \leq m \cdot \blockbnd_m \leq m \cdot \blockbig$ follows, implying the claim. 
Assume for the sake of contradiction that there exists some~$\loc_{j_0} \in \LOC$ with $\hat d^{\mvin}(\loc_{j_0})> \blockbnd_m$.

\mypara{An acyclic auxiliary graph $\hat G_\blocksize$.}
We proceed with defining an auxiliary digraph~$\hat G_\blocksize$ defined over~$\LOC$ in which $(\loc_j,\loc_{j'})$ is an arc exactly if $\hat d(\loc_j,\loc_{j'}) > \blockbnd_1$. 
Notice that by \cref{obs:superblock}, for each arc~$e$ in~$\hat G_\blocksize$ we know that the set of families moving from the tail of~$e$ to the head of~$e$ under the relocation contains a \superblock; let $B_e$ be such a \superblock\ corresponding to~$e$. 

It is straightforward to see that the same argument used in the proof of \cref{clm:DPbyfamilies-bounded-diff} for showing that $G_\blocksize$ is acyclic can be applied to prove that~$\hat G_\blocksize$ is acyclic. 
However, to find a path in~$\hat G_\blocksize$ that leads to a contradiction, we need a more careful, somewhat different argument.

\mypara{Finding a path~$\hat\chainpath$ in $G_\blocksize$.}
Let $\hat \LOC_\dashv$ contain those places~$\loc_j \in \LOC$ where 
\begin{itemize}[leftmargin=30pt]
\item[(a$_\dashv$)] $ 
\satur(\loc_j,\hat\asg_q) \leq \capl_j - \blocksize$, and
\item[(b$_\dashv$)] $
\satur(\loc_j,\hat\asg_{q+1}) \geq  \caplmin'_j+\blocksize$.
\end{itemize}
Analogously, let $\hat\LOC_{\vdash}$ contain those places~$\loc_j \in \LOC$ where
\begin{itemize}[leftmargin=30pt]
\item[(a$_\vdash$)] $
\satur(\loc_j,\hat\asg_{q+1}) \leq \capl_j - \blocksize$, and
\item[(b$_\vdash$)] $
\satur(\loc_j,\hat\asg_q) \geq  \caplmin_j+\blocksize$.
\end{itemize}

Define a \emph{chain} from~$\loc_{j_0}$ to~$\hat\LOC_{\dashv}$ as a sequence $\loc_{j_0},\loc_{j_1},\dots,\loc_{j_{\altsubm}}$ such that (i)
$\loc_{j_{\altsubm}} \in \hat\LOC_{\dashv}$, and 
(ii) $\hat d(\loc_{j_{h-1}},\loc_{j_h})> \blockbnd_{m-h}$ for each integer~$h$ with $0<h \leq \, \altm$. 
We build such a chain by induction. We start from the sequence containing only~$\loc_{j_0}$, and maintain condition~(ii) as an invariant; note that (ii) holds trivially for the sequence~$\loc_{j_0}$. 
So suppose that we have already built a sequence 
$\loc_{j_0},\loc_{j_1},\dots,\loc_{j_h}$ for which condition~(ii) holds.
If $\loc_{j_h} \in \hat\LOC_{\dashv}$, then we are done, as the sequence fulfills condition~(i) as well, and thus is a chain. 
If $\loc_{j_h} \notin \hat\LOC_{\dashv}$, then either condition~(a$_\dashv$) or condition~(b$_\dashv$) does not hold for~$\loc_{j_h}$.

First let us assume that $\loc_{j_h}$ does not satisfy condition~(a$_\dashv$). 
Then 
\begin{linenomath*}
\begin{align*}
\capl_{j_h} - \blocksize &<
\satur(\loc_j,\hat\asg_{q}) \\
&= 
\satur(\loc_j,\hat\asg_{q+1}) 
- \hat{d}^\mvin(\loc_{j_h})
+ \hat{d}^\mvout(\loc_{j_h}) \\
&\leq \capl_{j_h} - \hat{d}^\mvin(\loc_{j_h})
+ \hat{d}^\mvout(\loc_{j_h}).
\end{align*}
\end{linenomath*}
From this, we get that
\begin{linenomath*}
\begin{equation}
\hat{d}^{\mvout}(\loc_{j_h}) > \hat{d}^{\mvin}(\loc_{j_h})-\blocksize.
\end{equation} 
\end{linenomath*}
Assume now that $\loc_{j_h}$ violates condition~(b$_\dashv$). 
Then 
{\allowdisplaybreaks
\begin{linenomath*}
\begin{align*}
\caplmin_{j_h}+1 &\geq  \caplmin'_{j_h} > 
\satur(\loc_{j_h},\hat\asg_{q+1})
- \blocksize \\
&= 
\satur(\loc_{j_h},\hat\asg_{q})
+ \hat{d}^{\mvin}(\loc_{j_h}) 
- \hat{d}^{\mvout}(\loc_{j_h}) - \blocksize \\
&\geq \caplmin_{j_h} + \hat{d}^{\mvin}(\loc_{j_h}) 
- \hat{d}^{\mvout}(\loc_{j_h}) - \blocksize, 
 \end{align*}
 \end{linenomath*}
}
which implies 
\begin{linenomath*}
\begin{equation}
\label{eq:not-at-pathend}
\hat{d}^{\mvout}(\loc_{j_h}) > \hat{d}^{\mvin}(\loc_{j_h})-\blocksize-1.
\end{equation} 
\end{linenomath*}Hence, (\ref{eq:not-at-pathend}) holds in both cases. 

Since the families moving out from~$\loc_{j_h}$ must move into some place in~$\LOC \setminus \{\loc_{j_h}\}$, there must exist some~$\loc_{j_{h+1}} \in \LOC \setminus \{\loc_{j_h}\}$ for which 
\begin{linenomath*}
\begin{equation}
\hat d(\loc_{j_h},\loc_{j_{h+1}}) \geq \frac{\hat d^{\mvin}(\loc_{j_h})-\blocksize -1}{m-1}
>
\frac{\blockbnd_{m-h}-\blocksize-\maxreq}{m-1} =
\blockbnd_{m-h-1}
\end{equation} 
\end{linenomath*}
where the second inequality follows from condition~(ii) if $h>0$, and from the assumption that $d^{\mvin}(\loc_{j_0}) > \blockbnd_m$ in the case $h=0$. Hence, we can pick $\loc_{j_{h+1}}$ as the next place in the chain, as the sequence $\loc_{j_0},\loc_{j_1},\dots,\loc_{j_h},\loc_{j_{h+1}}$ fulfills condition~(ii). Since $G_\blocksize$ is acyclic and $|\LOC|$ is finite, the existence of a chain from~$\loc_{j_0}$ to~$\hat\LOC_{\dashv}$ follows.

	Next, we similarly build a \emph{back-chain} from~$\hat\LOC_{\vdash}$ to~$\loc_{j'_0}=\loc_{j_0}$, which is a sequence $\loc_{j'_{\altsubmprime}},\dots,\loc_{j'_1},\loc_{j'_0}$ such that (i')
$\loc_{j'_{\altsubmprime}} \in \hat\LOC_{\vdash}$, and 
(ii') $\hat d(\loc_{j'_h},\loc_{j'_{h-1}})>\blockbnd_{m-h}$ for each integer~$h$ with $0<h \leq {\altm}'$. 
Because $\hat d^{\mvin}(\loc_{j_0}) > \blockbnd_m$,  there must exist some place~$\loc_{j'_1} \in \LOC \setminus \{\loc_{j_0}\}$ for which $\hat d(\loc_{j'_1},\loc_{j_0}) \geq \hat d^{\mvin}(\loc_{j_0})/(m-1)>\blockbnd_m/(m-1)> \blockbnd_{m-1}$. We build our back-chain starting from the sequence $\loc_{j'_1},\loc_{j_0}$ by induction. 

Suppose that we already have a sequence 
$\loc_{j'_h},\dots,\loc_{j'_1},\loc_{j_0}$ for which condition~(ii') holds.
If $\loc_{j'_h} \in \hat\LOC_{\vdash}$, then we are done, as the sequence fulfills condition~(i') as well, and thus is a back-chain. 
If $\loc_{j'_h} \notin \hat\LOC_{\vdash}$, then 
either condition~(a$_\vdash$) or condition~(b$_\vdash$) fails for~$\loc_{j'_h}$.

First, assume that condition~(a$_\vdash$) fails for~$\loc_{j'_h}$. Then 
\begin{linenomath*}
\begin{align*}
\capl_{j'_h} - \blocksize &<
\satur(\loc_{j'_h},\hat\asg_{q+1}) \\
&= 
\satur(\loc_{j'_h},\hat\asg_{q}) 
+ \hat{d}^\mvin(\loc_{j'_h})
- \hat{d}^\mvout(\loc_{j'_h}) \\
&\leq \capl_{j'_h} + \hat{d}^\mvin(\loc_{j'_h})
- \hat{d}^\mvout(\loc_{j'_h}),
\end{align*}
\end{linenomath*}
which implies 
\begin{linenomath*}
\begin{equation}
\label{eq:not-at-pathstart}
\hat{d}^{\mvin}(\loc_{j'_h}) > \hat{d}^{\mvout}(\loc_{j'_h})-\blocksize.
\end{equation} 
\end{linenomath*}
Second, assume that condition~(b$_\vdash$) fails for~$\loc_{j'_h}$. Then
\begin{linenomath*}
\begin{align*}
\caplmin_{j'_h} &\leq \caplmin'_{j'_h} \leq \satur(\loc_{j'_h},\hat\asg_{q+1}) \\
&=\satur(\loc_{j'_h},\hat\asg_{q}) 
+ \hat{d}^\mvin(\loc_{j'_h})
- \hat{d}^\mvout(\loc_{j'_h}) \\
&< \caplmin_{j'_h} + \blocksize 
+ \hat{d}^\mvin(\loc_{j'_h})
- \hat{d}^\mvout(\loc_{j'_h})
\end{align*}
\end{linenomath*}
which again implies (\ref{eq:not-at-pathstart}).

Since the families moving into~$\loc_{j'_h}$
must have come from some place in~$\LOC \setminus \{\loc_{j'_h}\}$, there must exist  some place~$\loc_{j'_{h+1}} \in \LOC \setminus \{\loc_{j'_h}\}$ for which 
\begin{linenomath*}
\begin{equation*}
d(\loc_{j'_{h+1}},\loc_{j'_h}) \geq \frac{d^{\mvin}(\loc_{j'_h})}{m-1}>\frac{\blockbnd_{m-j}-\blocksize}{m-1}>\blockbnd_{m-j-1}.
\end{equation*} 
\end{linenomath*} 
Hence, we can pick $\loc_{j'_{h+1}}$ as the next place in the back-chain, because the sequence $\loc_{j'_{h+1}},\loc_{j'_h},\dots,\loc_{j'_1},\loc_{j_0}$ fulfills condition~(ii'). 
Since $\hat G_\blocksize$ is acyclic and $|\LOC|$ is finite, the existence of a back-chain from~$\hat\LOC_{\vdash}$ to~$\loc_{j_0}$ follows.

Consider the sequence~$\hat{\chainpath}$ of places obtained by concatenating our back-chain from~$\hat\LOC_{\vdash}$ to~$\loc_{j_0}$ with the chain from~$\loc_{j_0}$ to~$\hat\LOC_{\dashv}$. Observe that by conditions~(ii) and~(ii'), there is an arc in~$\hat G_\blocksize$ from each place in~$\hat{\chainpath}$ to the next place in~$\hat{\chainpath}$. Since $G_\blocksize$ is acyclic, this means that~$\hat{\chainpath}$ forms a path in~$\hat G_\blocksize$.

\mypara{The contradiction implied by our path~$\hat{\chainpath}$.}
It remains to show that the existence of our path~$\hat{\chainpath}$ from~$\hat\LOC_{\vdash}$ to~$\LOC_{\dashv}$ in~$\hat G_\blocksize$ leads to a contradiction. Let $E(\hat{\chainpath})$ denote the set of arcs on this path, and define $\FA_{\hat{\chainpath}}=\bigcup_{e \in E(\hat{\chainpath})} B_e$, that is, $\FA_{\hat{\chainpath}}$ is the union of \superblock{}s corresponding to the arcs on~$\hat{\chainpath}$. 
Define the assignment $\hat\asg'_{q+1}$ for~$I_{i+1}$ obtained from~$\hat\asg_{q+1}$ by moving all families contained in 
$\FA_{\hat{\chainpath}}$ to the place they were assigned under~$\hat\asg_q$, i.e., ``backward'' along the path~$\hat{\chainpath}$.
Let us show that $\hat\asg'_{q+1}$ is feasible for~$\hat I_{q+1}$.

Since the first place~$\loc_{\loc_{j'_{\altsubm'}}}$ on~$\hat{\chainpath}$ belongs to~$\LOC_{\vdash}$, by condition~(a$_{\vdash}$) it has free capacity at least~$\blocksize$ under~$\hat\asg_{q+1}$ that can be used to accommodate the superblock corresponding to the first arc of~$\hat{\chainpath}$, so the upper quota of~$\loc_{\loc_{j'_{\altsubm'}}}$ is not exceeded under~$\hat\asg'_{q+1}$. 
Since the last place~$\loc_{\loc_{j_{\altsubm}}}$ on~$\hat{\chainpath}$ belongs to~$\LOC_{\dashv}$, removing a \superblock\ from the families assigned by~$\hat\asg_{q+1}$ to~$\loc_{\loc_{j_{\altm}}}$  does not violate its lower quota due to condition~(b$_\dashv$). 
Since all \superblock{}s have the same total requirement~$\blocksize$, it also follows that the lower and upper quotas are respected by~$\hat\asg'_{q+1}$ for all remaining places as well. This proves that $\hat\asg'_{q+1}$ is feasible for~$\hat I_{q+1}$.

Moreover, since $\Delta(\hat\asg_q,\hat\asg'_{q+1}) < \Delta(\hat\asg_q,\hat\asg_{q+1})$, we know that the total utility of~$\hat\asg'_{q+1}$ is less than that of~$\hat\asg_{q+1}$, due to our choice of~$\hat\asg_{q+1}$, which yields
\begin{linenomath*}
\begin{equation}
\label{eq:moving-along-path-2}
\sum_{\fa_h \in \FA_{\hat{\chainpath}},\loc_j = \hat\asg_q(\fa_h)} \profit_h[j] < 
\sum_{\fa_h \in \FA_{\hat{\chainpath}},\loc_j = \hat\asg_{q+1}(\fa_h)} \profit_h[j].
\end{equation}
\end{linenomath*}

Let us now construct an assignment~$\hat\asg'_q$ from $\hat\asg_q$ by moving all families contained in~$\FA_{\hat{\chainpath}}$ to the place they are assigned under~$\hat\asg_{q+1}$, i.e., ``forward'' along the path~$\hat{\chainpath}$. 
Again, we can show that $\hat\asg'_q$ is feasible for~$\hat I_q$.

The last place on~$\hat{\chainpath}$ belongs to~$\LOC_{\dashv}$, and therefore by condition~(a$_\dashv$) can accommodate a \superblock\ besides the families assigned to it by~$\hat\asg_q$. The first place on~$\hat\chainpath$ belongs to~$\LOC_{\vdash}$,
 and thus by condition~(b$_\vdash$) we can remove a \superblock\ from among the families assigned to it by~$\hat\asg_q$ without violating its lower quota. Feasibility is therefore maintained at all places, as promised.
Due to~(\ref{eq:moving-along-path-2}),
the utility of~$\hat\asg'_q$ exceeds the utility of~$\hat\asg_q$, which contradicts the optimality of~$\hat\asg_q$. This contradiction proves the claim.
}

We are now ready to present 
 our algorithm for \maxutil-\RR\ based on Claims~\ref{clm:DPbyfamilies-bounded-diff} and~\ref{clm:DPbyfamilies-bounded-diff-2}. 
We use a combination of dynamic programming and color-coding.

Initially, we compute a maximum-utility feasible allocation~$\asg_1$ for~$I_1$ by assigning family~$\fa_1$ to a place that can accommodate it, and among all such places, yields the highest utility for~$\fa_1$. Then, in the first phase of the algorithm, for each $i \in [n-1]$ we compute an optimal assignment for~$I_{i+1}$ by slightly modifying~$\asg_i$. 
In the second phase, starting from the assignment~$\hat\asg_0:=\asg_n$ for $\hat I_0:=I_n$, we compute an optimal assignment for~$\hat I_q$ by slightly modifying~$\hat \asg_{q-1}$ for each $q \in [\sumcapmins]$.  
In each step of the first and second phases, we apply a procedure based on color-coding; 
the remainder of the proof contains the description of this procedure and its proof of correctness. 

Let $I_\curr$ be the instance of phase~1 or~2 for which we have already computed 
an optimal assignment~$\asg_\curr$, and suppose that $I_\nxt$ is the next instance for which we aim to compute an optimal assignment $\asg_\nxt$. 
Thus, $I_\nxt$ is either obtained from~$I_\curr$ by adding some family~$\fa_i \in  \FA$, or by raising the lower quota for one of the places in~$\LOC$ by one. 
Let $\FA_\curr$ and~$\FA_\nxt$ denote the set of families in~$I_\curr$ and in~$I_{\nxt}$, respectively.
Due to~\cref{clm:DPbyfamilies-bounded-diff,clm:DPbyfamilies-bounded-diff-2}, we can choose $\asg_{\nxt}$ so that  
$\Delta(\asg_{\nxt},\asg_\curr) \leq m \cdot  \blockbig$.

\mypara{Guessing step.}
Let $X(\loc_j,\loc_{j'},r)$ denote the set of 
all families with requirement~$r$ that are assigned to~$\loc_j$ by~$\asg_{\curr}$ but are moved to~$\loc_{j'}$ by~$\asg_{\nxt}$.
We guess the number $x(\loc_j,\loc_{j'},r)=|X(\loc_j,\loc_{j'},r)|$ for each $\loc_j,\loc_{j'} \in \LOC$ and $r \in [\maxreq]$.
By our choice of~$\asg_{\nxt}$,
we have 
\begin{linenomath*} 
\begin{equation*}
\sum_{j \in [m]} \sum_{j' \in [m] \setminus \{j\}} \sum_{r \in [\maxreq]} x(\loc_{j'},\loc_j,r) 
= \Delta(\asg_{\nxt},\asg_\curr) \leq m \cdot  \blockbig.
\end{equation*}
\end{linenomath*} 
Since we need to guess $m\cdot (m-1)\cdot \maxreq$ values  that add up to at most~$m \cdot \blockbig$, there are no more than $(\blockbig)^{m\cdot (m-1)\cdot \maxreq}$ possibilities to choose all values~$x(\loc_{j'},\loc_{j},r)$.
Thus, the number of possibilities for all our guesses 
is bounded by a function of~$m$ and~$\maxreq$ only. 

\mypara{Color-coding step.} We proceed by randomly coloring all families in~$I_\nxt$ with $m$ colors in a uniform and independent way. We say that a coloring is \emph{suitable} for $\asg_{\nxt}$, if for each $\loc_j \in \LOC$, all families in~$\asg_{\nxt}^{-1}(\loc_j) \setminus \asg_\curr^{-1}(\loc_j)$ have color~$j$. Thus, in a suitable coloring, each family whose assignment changes between~$\asg_\nxt$ and~$\asg_{\curr}$ must be assigned by~$\asg_\nxt$ to the place corresponding to its color.
Considering that  $I_\nxt$ may contain one more family than~$I_\curr$, we get
\begin{linenomath*} 
\begin{equation*}
 \sum_{\loc_j \in \LOC} \left| \asg_\nxt^{-1}(\loc_j) \setminus \asg_\curr^{-1}(\loc_j)\right| 
\leq 1+ \Delta(\asg_{\nxt},\asg_\curr) \leq m \cdot \blockbig+1.
\end{equation*}
\end{linenomath*} 
Therefore, the probability that the algorithm produces a suitable coloring is at least $m^{-m \blockbig+1}$. 

\mypara{Modification step.}
Assume that our coloring~$\chi$ is suitable. 
In the first phase, this implies  that the unique family~$\fa_i\in \FA_\nxt \setminus \FA_\curr$ must be assigned by~$\asg_{\nxt}$ to $\loc_{\chi(\fa_i)}$. 
Thus, we fix the assignment on~$\fa_i$ as~$\loc_{\chi(\fa_i)}$.
We proceed with the remaining families of~$\FA_\nxt$ 
as follows. 

For each $\loc_j,\loc_{j'} \in \LOC$ and $r \in [\maxreq]$, we compute the set $D(\loc_j,\loc_{j'},r):=\{\fa_h \in \FA_\curr: \asg_\curr(\fa_h)=\loc_j, \chi(\fa_h)=j',\reqf_h=r\}$; the suitability of $\chi$ means that $X(\loc_j,\loc_{j'},r) \subseteq D(\loc_j,\loc_{j'},r)$.
With each family~$\fa_h \in D(\loc_j,\loc_{j'},r)$, we associate the value~$\profit_h^{j'}-\profit_h^{j}$ which describes the increase in utility caused by moving $a_h$ from~$\loc_j$ to~$\loc_{j'}$. 
We order the families in~$D(\loc_j,\loc_{j'},r)$
in a non-increasing order of these values, and we pick the first $x(\loc_j,\loc_{j'},r)$ families according to this ordering; denote the obtained set $\wtD(\loc_j,\loc_{j'},r)$.
We can now define $\asg'_{i+1}$ as follows for each $\fa_h \in \FA_\curr$: 
\begin{linenomath*} 
\begin{equation*}
 \asg'_\nxt(\fa_h)=\left\{ 
\begin{array}{ll}
\loc_{\chi(\fa_h)} & \textrm{ if } \fa_h \in \FA_\nxt \setminus \FA_\curr; \\
\loc_{j'} & \textrm{ if } \exists j,r: \fa_h \in \wtD(\loc_j,\loc_{j'},r) ; \\
\asg_\curr(\fa_h) & \textrm{ otherwise.}
\end{array}
\right.
\end{equation*}
\end{linenomath*} 
Observe that the total requirement of all families assigned to some place~$\loc_j \in \LOC$ is the same in~$\asg'_{\nxt}$ as in~$\asg_{\nxt}$, due to the definition~$\asg'_{\nxt}$ and the correctness of our guesses. Therefore, $\asg'_{\nxt}$ is feasible.
Furthermore,
\begin{linenomath*} 
\begin{align*}
\notag 
&\sum_{\substack{\fa_h \in \FA_\curr,\\ \loc_j=\asg'_{\nxt}(\fa_h)}} \!\!\!\!   \profit_h[j] 
= 
\util(\loc_j,\asg_\curr)
+ 
\sum_{\substack{\exists j,j',r:\\
 \fa_h \in \wtD(\loc_j,\loc_{j'},r) }} \!\!\!\! (\profit_h[j'] - \profit_h[j])  
\\
& \geq  
\util(\loc_j,\asg_\curr)
+ 
\!\!\!\! 
\label{eqn:change-in-utility}
\sum_{\substack{\exists j,j',r: \\
\fa_h \in X(\loc_j,\loc_{j'},r) }} 
\!\!\!\! (\profit_h[j'] - \profit_h[j])  
= \sum_{\substack{\fa_h \in \FA_\curr,\\ \loc_j=\asg_\nxt(\fa_h)}} 
\!\!\!\! \profit_h[j]
\end{align*}
\end{linenomath*} 
where the inequality follows from our choice of the sets~$\wtD(\loc_j,\loc_{j'},r)$ and the facts
 $|\wtD(\loc_j,\loc_{j'},r)| = |X(\loc_j,\loc_{j'},r)|$
and $X(\loc_j,\loc_{j'},r) \subseteq D(\loc_j,\loc_{j'},r)$, 
which in turn follow from our assumptions that our guesses are correct and that the coloring~$\chi$ is suitable.
Since $\asg'_\nxt$ coincides with~$\asg_\nxt$ on $\FA_\nxt \setminus\FA_\curr$,
the above inequality
implies 
that $\asg'_{\nxt}$ is a maximum-utility feasible assignment for~$I_{\nxt}$, proving the correctness 
of our algorithm.

The presented algorithm can be derandomized using standard techniques, based on $(n,m \cdot \blockbig+1)$-perfect families of perfect hash functions~\cite{AYZ95}. 
Since both the number of possible guesses and the number of families that we have to color correctly are bounded by a function of $m+\maxreq$, the modification procedure applied in the first or second phases of the algorithm runs in FPT time when parameterized by~$m+\maxreq$. As we have to carry out this procedure $n+\sumcapmins$ times
and we can assume w.l.o.g.\ that $\sumcapmins \leq n \cdot \maxreq$, the total running time is FPT w.r.t.~$m+\maxreq$.
\end{proof}

We close this section by showing that if the desired total utility~$\finalutil$ is small and there are no lower quotas, then \maxutil-\RR\ for $\noser=1$ can be solved efficiently. 
Recall that with lower quotas, even the case $\finalutil=0$ is \NPh\ by \cref{prop:wtbinpacking}.
The algorithm of \cref{prop:fpt-finalutilt1},
\ifshort
\else presented in Appendix~\ref{proof:prop:fpt-finalutilt1},
\fi starts with a greedily computed assignment, and then deletes irrelevant families to obtain an equivalent instance with at most~$(\finalutil)^3$ families that can be solved efficiently.

\begin{restatable}[\appsymb]{theorem}{propfptfinalutilt}
\label{prop:fpt-finalutilt1}
\maxutil- and \pareto-\RR\ for $\noser=1$ are \FPT\ w.r.t.~$\finalutil$, the desired utility, if there are no lower quotas.
\end{restatable}

\appendixproofwithstatement{prop:fpt-finalutilt1}{\propfptfinalutilt*}{
Let $I$ denote our input instance.
We present an algorithm which, in polynomial time, produces an equivalent instance~$I'$ with at most $(\finalutil)^3$ families. Applying \cref{prop:fpt-n} to~$I'$ then yields fixed-parameter tractability for parameter~$\finalutil$. 

We start by checking whether there exists a family~$\fa_i \in \FA$ and a place~$\loc_j \in \LOC$ such that $\loc_j$ can accommodate~$\fa_i$ and $\profit_i[j]\geq \finalutil$. If so, we output the assignment that assigns~$\fa_i$ to~$\loc_j$ and leaves every other family unassigned.

Otherwise, we proceed by greedily assigning families to places so that (i) the assignment remains feasible, and (ii) each family~$\fa_i \in \FA$ assigned to some place~$\loc_j \in \LOC$ has positive utility for that place, i.e., $\profit_i[j]\geq 1$. 
Let $\asg_0$ be the feasible assignment obtained at the end of this greedy process.

If $\util(\asg_0)\geq \finalutil$, then we output~$\asg_0$. 
Otherwise, consider those places and families that are ``relevant'' according to~$\asg_0$, namely the sets $\LOC_0=\{\loc_j \in \LOC: \asg_0^{-1}(\loc_j) \neq \emptyset\}$ and~$\FA_0=\{\fa_i:\asg_0(\fa_i) \neq \perp\}$. 
First, we \emph{mark} each family in~$\FA_0$.
Next, for each $\loc_j \in \LOC_0$ and each $\gamma \in [\finalutil]$, 
we order all families in $\FA_{j,\gamma}:=\{\fa_i \in \FA: \profit_i[j]=\gamma\}$  according to their requirement in a non-decreasing manner, and mark the first $\finalutil$ families in this ordering (or all of them, if $|\FA_{j,\gamma}|<\finalutil$). 
We define an instance~$I'$ of \maxutil-\RR\ by deleting all unmarked families.

Let $\FA'$ be the set of families in~$I'$, that is, the families we have marked;  
we claim $|\FA'|\leq (\finalutil)^3$. First, due to condition~(ii), we know that $|\FA_0| <\finalutil$ and $|\LOC_0|< \finalutil$ follows from $\util(\asg_0) <\finalutil$.
Additionally, we marked at most $\finalutil$  families from~$\FA_{j,\gamma}$ for each $\loc_j \in \LOC_0$ and $\gamma \in [\finalutil]$. Summing this for all values of~$\loc_j$ and~$\gamma$, we get at most $(\finalutil-1)(\finalutil)^2$ families. Taking into account the at most~$\finalutil$ families in~$\FA_0$, we get that there at most~$(\finalutil)^3$ marked families, as promised.

\begin{claim}
\label{clm:semikernel}
$I'$ is equivalent with~$I$.
\end{claim} 
\begin{claimproof}
Clearly, a feasible assignment for~$I'$ is also feasible for~$I$, and has the same utility in both instances. 
Suppose now that $\asg$ is a feasible assignment for~$I$ with $\util(\asg) \geq \finalutil$; we construct a feasible 
assignment for~$I'$
with utility at least~$\finalutil$ as follows. 

First, for each family in~$\FA'$, we let~$\asg'$ and~$\asg$ coincide.
Let $\LOC^\star$ contain all places that have at least one family assigned to them by~$\asg$. 
Now, for  each $\loc_j \in \LOC^\star$ and each $\gamma\in [\finalutil]$, we take the ordering of~$\FA_{j,\gamma}$ used in the marking process (recall that the families in~$\FA_{j,\gamma}$ were ordered in a non-decreasing way according to their requirements), and we pick $n_{j,\gamma}$ families one by one from $\FA_{j,\gamma}$ among 
those that are not yet assigned to some place by~$\asg'$,  always picking the first family available, where 
\begin{linenomath*}
\begin{equation*}
n_{j,\gamma}=\left| \{\fa_i \in \FA \setminus \FA': \fa_i \in \asg^{-1}(\loc_j),  \profit_i[j]=\gamma\}\right|.
\end{equation*}
\end{linenomath*}
We let $\asg'$ assign the picked families to~$\loc_j$. 
This process stops once the number of assigned families reaches~$\finalutil$ or we have iterated through all places in~$\LOC^\star$ and all utility values~$\gamma \in [\finalutil]$. We set~$\asg'(\fa_i)=\perp$ for all families of~$\fa_i \in \FA'$ left unassigned thus far.

Let us show that for each $\loc_j \in \LOC^\star$ and $\gamma \in [\finalutil]$, we are always able to pick a marked family from~$\FA_{j,\gamma}$ during the above process.
Assume first that
we have marked the first $\finalutil$ families from~$\FA_{j,\gamma}$. Then due to our stopping condition, at each step there are less than~$\finalutil$ families in~$\FA_{j,\gamma}$ that are already assigned by~$\asg'$ to some place,  so 
at each step when we pick the first available family from~$\FA_{j,\gamma}$, we pick a marked family.
Second, assume that $|\FA_{j,\gamma}|<\finalutil$, and thus all families in $\FA_{j,\gamma}$ are marked. In this case, there is no family in $\FA \setminus \FA'$ that has utility~$\gamma$ for~$\loc_j$, due to the definition of~$\FA_{j,\gamma}$. However, this implies $n_{j,\gamma}=0$. This proves that all picked families are in~$F'$, and 
hence, $\asg'$ is an assignment for~$I'$.

We next show that $\util(\asg') \geq \finalutil$. On one hand, this is clear if the algorithm stops because it has assigned at least~$\finalutil$ families, since all families contribute at least~$1$ to the total utility of~$\asg'$. 
On the other hand, if the algorithm stops because it has iterated over all possible places and utility values considered, then we know that at each place~$\loc_j \in \LOC^\star$ and for each $\gamma \in [\finalutil]$, there are exactly the same number of families assigned to~$\loc_j$ by~$\asg$ and~$\asg'$, due to our definition of~$n_{j,\gamma}$. This implies 
{ \allowdisplaybreaks
\begin{linenomath*}
\begin{align*}
\util(\asg')&=
\sum_{\loc_j \in \LOC^\star} 
\sum_{\fa_i \in \asg'^{\, -1}(\loc_j)} \profit_i[j] \\
&= 
\sum_{\loc_j \in \LOC^\star} 
\left(
\sum_{\fa_i \in \asg'^{\, -1}(\loc_j) \cap \asg^{-1}(\loc_j)} 
\profit_i[j] +
\sum_{\gamma \in [\finalutil]}\gamma \cdot n_{j,\gamma} 
\right)
\\
&=
\sum_{\loc_j \in \LOC^\star} 
\sum_{\fa_i \in \asg^{-1}(\loc_j)} \profit_i[j] = \util(\asg)\geq \finalutil. 
\end{align*}
\end{linenomath*}
}

It remains to show that $\asg'$ is feasible. Consider some~$\loc_j \in \LOC^\star$. Clearly the families that both~$\asg$ and~$\asg'$ assign to~$\loc_j$, i.e., those in~$\asg^{-1}(\loc_j) \cap \FA'$, contribute the same amount to the load of~$\asg$ and~$\asg'$. Let $r_{j,\gamma}$ denote the maximum requirement of any marked family in~$\FA_{j,\gamma}$.
Consider the families in $\asg^{-1}(\loc_j) \setminus \FA'$, and partition this set according to the utility values these families have for~$\loc_j$. Consider some~$\gamma \in [\finalutil]$.
Since each family $\fa_i$ in $\asg^{-1}(\loc_j) \setminus \FA'$ with $\profit_i[j]=\gamma$ is unmarked, it has requirement at least~$r_{j,\gamma}$. By contrast, the $n_{j,\gamma}$ families assigned to~$\loc_j$ from~$\FA_{j,\gamma}$ during the picking process are all marked, and thus have requirement at most~$r_{j,\gamma}$. Hence, we get 
\begin{linenomath*}
\begin{align*}
\satur(\loc_j,\asg')&=
\sum_{\fa_i \in \asg'^{\, -1}(\loc_j)} \reqf_i \\
&\leq  
\sum_{\fa_i \in \asg'^{\, -1}(\loc_j) \cap \asg^{-1}(\loc_j)} 
\reqf_i +
\sum_{\gamma \in [\finalutil]} r_{j,\gamma} \cdot n_{j,\gamma} 
\\
&\leq
\sum_{\fa_i \in \asg^{-1}(\loc_j)} \reqf_i
 = \satur(\loc_j,\asg). 
\end{align*}
\end{linenomath*}
Thus, $\asg$ is a feasible assignment for~$I'$ with utility at least~$\finalutil$, as required.
\end{claimproof}
The result for \maxutil\-\RR\ now follows from \cref{clm:semikernel} and \cref{prop:fpt-n}, and the result for \pareto-\RR\ follows by \cref{rem:maxutil-vs-pareto}.

We remark that it is possible to reduce the number of places as well: it suffices to keep (at most) $\finalutil$ places for each family $\fa_i \in \FA'$ among those which can accommodate it: we need to pick them in non-decreasing order of $\fa_i$'s utility for them. 
This yields a ``pseudo-kernelization'' algorithm for parameter~$\finalutil$ in the sense that it produces an equivalent instance where both the number of families and the number of places is bounded by a function of~$\finalutil$; however, the requirement values and the upper quotas may be unbounded.
}

\section{Multiple services}
\label{sec:moreservices}
\appendixsection{sec:moreservices}
Let us now consider the model when there are several services, i.e., $\noser>1$.
\todoHinline{Move the subsections to paragraphs. So need to rewrite the following because there are no section numbers anymore.}
We start with a strong intractability result for \feasible-\RR.
Then we focus in Pareto-optimality, and propose several algorithms that solve \pareto-\RR\ but not \maxutil-\RR, contrasted by tight hardness results. 
We close  by investigating \maxutil-\RR.

\paragraph{Feasibility.}
\label{sec:feasible}

When the number of services can be unbounded, then a simple reduction from \textsc{Independent Set} by Gurski et al.~\cite[Theorem~23]{gurski2019knapsack} shows that \maxutil-\RR\ is \NPh\ even if $m=1$, there are no lower quotas and the utilities are \uniform. With a slight modification of their reduction, we obtain \cref{prop_feas_hard_m} which shows the \NPh{}ness of \feasible-\RR\ in a very restricted setting.

\begin{restatable}[\appsymb]{proposition}{propfeashardm}
\label{prop_feas_hard_m}
The following problems are \NPh\ even if $\maxcap = \maxreq = 1$ and $m=1 \colon$
\begin{compactitem}
\item \feasible-\RR;
\item \pareto-\RR\ with equal preferences;
\item \maxutil-\RR\ with equal utilities.
\end{compactitem}
\end{restatable}

\appendixproofwithstatement{prop_feas_hard_m}{\propfeashardm*}{
We present a reduction from \textsc{Multicolored Independent Set}  to \feasible-\RR.
In this problem, we are given a graph~$G=(V,E)$ and an integer~$k$, with the vertex set of~$G$ partitioned into $k$ sets~$V_1,\dots,V_k$; 
the task is to decide whether there exists an independent set of size~$k$ that contains one vertex from each set~$V_i$, $i \in [k]$. 
This problem is \NPh\ and, in fact, \Woneh\ when parameterized by~$k$~\cite{pietrzak-multicolored-2003}.\footnote{Pietrzak dealt with the \textsc{Multicolored} (or \textsc{Partitioned) Clique} problem, but the claimed hardness results follow immediately from his results.}

We construct an instance of \feasible-\RR\ with a single location~$\loc$ as follows. 
We set $V$ as the set of families, and $E \cup \{\ser_1,\dots,\ser_k\}$ as the set of  services, with $\loc$ offering exactly one unit from every service. 
Each family~$v \in V_i$ for some~$i \in [k]$ requires one unit of each service associated with its incident edges, and one unit of~$\ser_i$. We set the lower quota for $\loc$ as~$1$ for each service~$\ser_i$, $i \in [k]$, and as~$0$ for each service in~$E$.

Notice that a feasible assignment assigns at least one vertex from each set~$V_i$, $i \in [k]$, to~$\loc$. Moreover, no two vertices assigned to~$p$ can be adjacent, as the edge connecting  them corresponds to a service from which both of these two vertices (i.e., families) need one unit. Thus, a feasible assignment yields an independent set in~$G$ containing a vertex from each set~$V_i$, $i \in [k]$. Vice versa, assigning such an independent set to~$\loc$ satisfies all lower and upper quotas, and thus yields a feasible assignment.

Finally, observe that adding arbitrary preferences or utilities to the constructed instance of \feasible-\RR\, we obtain an instance of \pareto- or \maxutil-\RR\, respectively, that is equivalent with the original input instance.
}

\paragraph{Pareto-optimality.}
\label{sec:pareto}
The reduction from \textsc{Independent Set} used by Gurski et al.~\cite{gurski2019knapsack} and also in \cref{prop_feas_hard_m} can be adapted to show the \NPh{}ness of \pareto-\RR\ 
in the case when there are no lower quotas, $m=2$, and we allow $\maxcap$ to be unbounded; see \cref{prop:pareto_two_locs}. 
Notice that in instances without lower quotas, a feasible, acceptable and Pareto-optimal assignment always exists. Hence, our hardness results for \pareto-\RR\ rely on the following fact.
\begin{restatable}[\appsymb]{obs}{obsparcomp}
\label{obs:Pareto-vs-completeness}
Given an instance~$I$ of \pareto-\RR\ with \dichotomous\ preferences, we can decide the existence of a feasible, acceptable and complete assignment for~$I$ by solving \pareto-\RR\ on~$I$. 
\end{restatable}

\appendixproofwithstatement{obs:Pareto-vs-completeness}{\obsparcomp*}{
It suffices to observe that there exists a feasible, acceptable and complete (\emph{fac}, for short) assignment if and only if every feasible, acceptable and Pareto-optimal (\emph{faP}, for short) assignment is complete. 
Clearly, a fac assignment is necessarily Pareto-optimal, since families are indifferent between places that they find acceptable. On the other hand, if there exists a fac assignment~$\asg$, then each faP assignment must be complete, as otherwise $\asg$ would be a Pareto-improvement for it.
}

\vspace{-6pt}
\begin{restatable}[\appsymb]{proposition}{propparetotwolocs}
\label{prop:pareto_two_locs}
\pareto- and \maxutil-\RR\ are \NPh\ even if $m  = 2, \maxreq = 1$, there are no lower quotas, and families have \indifferent\ preferences or utilities.
\end{restatable}

\appendixproofwithstatement{prop:pareto_two_locs}{\propparetotwolocs*}{
We present a reduction from \textsc{Independent Set} to \pareto-\RR\ that is similar to the proof of  
\cref{prop_feas_hard_m}.
Let $(G,k)$ be our input instance with~$G=(V,E)$. 

We set $V$ as the set of families, $E \cup \{\ser^\star\}$ as the set of services, and $\LOC=\{\loc^\star,\loc\}$ as the set of places, with both places acceptable for each family. Place~$\loc^\star$ offers exactly one unit of each service in~$E$, and offers~$k$ units of service~$\ser^\star$. 
Place~$\loc$ offers offers $|V|-k$ units of every service. 
Moreover, each family~$v \in V$ requires one unit of each service associated with its incident edges, and one unit of service~$\ser^\star$. 

Note that $\loc$ can accommodate an arbitrary set of~$|V|-k$ families (but not more), while $\loc^\star$ can accommodate at most $k$ families corresponding to an independent set. Thus,
$G$ admits an independent set of size~$k$ if and only if~$I$ admits a feasible and complete assignment. From this, the result for \pareto-\RR\ follows by \cref{obs:Pareto-vs-completeness}.

By adding \uniform\ utilities to the constructed instance and setting the desired utility as $\finalutil=|V|$, we obtain an instance where an allocation reaches the desired utility if and only if it is complete. From this, the result for \maxutil-\RR\ follows.
}

In the reduction proving \cref{prop:pareto_two_locs}, the value~$\maxcap$ is unbounded. Next, we show a reduction from \textsc{3-Coloring} proving that even the case when $\maxcap=1$ is \NPh\ if there are at least~3 places.

\begin{restatable}[\appsymb]{theorem}{thmparetolocmaxcap}\label{thm:pareto_loc_maxcap}
\pareto- and \maxutil-\RR\ are \NPh\ even when $m = 3, \maxcap = 1$, there are no lower quotas, and families have \indifferent\ preferences or utilities, respectively.
\end{restatable}

\appendixproofwithstatement{thm:pareto_loc_maxcap}{\thmparetolocmaxcap*}{
We present a reduction from \textsc{3-Coloring}; let $G=(V,E)$ be the input graph. We create an instance~$I$ of \pareto-\RR\ with \indifferent\ preferences and without lower quotas as follows. 

First, we set~$V$ as the set of families, $E$ as the set of services, and $\LOC=\{\loc_1, \loc_2, \loc_3\}$ as the set of places, with all places acceptable for each family. Each place offers exactly one unit of each service in~$E$, 
and each family~$v \in V$ requires one unit of each service associated with its incident edges. 

Notice that a feasible assignment cannot assign two families corresponding to adjacent vertices to the same place, as the edge connecting them represents a service required by both of them. Thus, the families assigned to a given place must form an independent set in~$G$. 

We claim that a feasible and complete assignment for~$I$ exists if and only if $G$ admits a proper 3-coloring. 
First, if $\asg$ is a complete and feasible assignment, then it yields a proper 3-coloring of~$G$ by the above reasoning. 
Second, if $G$ is 3-colorable, then a 3-coloring of~$G$ yields a feasible and complete assignment, since there is no service from which a color class requires more than one unit. Since the families have \indifferent, and hence, \dichotomous\ preferences, by \cref{obs:Pareto-vs-completeness}
we can decide the 3-colorability of~$G$ by solving \pareto-\RR\ on~$I$.

To obtain \NPh{}ness for \maxutil-\RR, we set \uniform\ utilities and $\finalutil=|V|$ as the desired utility, so that an assignment has utility at least~$\finalutil$ exactly if it is complete. The result then follows from the above reasoning.
}

Contrasting the intractability 
result of \cref{prop:pareto_two_locs} for $m=2$,
we show that a simple, greedy algorithm solves \pareto-\RR\ for $m=1$ in polynomial time assuming that there are no lower quotas.

\begin{restatable}[\appsymb]{proposition}{thmparetoPmone}\label{thm:paretoPm1}
\pareto-\RR\ for $m=1$ is polynomial-time solvable if  there are no lower quotas.
\end{restatable}

\appendixproofwithstatement{thm:paretoPm1}{\thmparetoPmone*}{
Let $\LOC = \{\loc_1\}$. 
The key observation is that a feasible and acceptable assignment~$\asg$ is Pareto-optimal if and only if it \emph{non-wasteful}, meaning that
there exists no unassigned family~$\fa_i \in \FA$ that finds $\loc_1$ acceptable and for which $\asg^{-1}(\loc_1) \cup \{\fa_i\}$ can be accommodated at~$\loc_1$.
First, if $\asg$ is non-wasteful, then 
we cannot accommodate any more families in~$\loc_1$ without disimproving some family, so $\asg$ is Pareto-optimal. 
Second, if $\asg$ is Pareto-optimal, then it must be non-wasteful as well: indeed, the existence of an unassigned family $\fa_i$ that finds $\loc_1$ acceptable and for which $\asg^{-1}(\loc_1) \cup \{\fa_i\}$ can be accommodated  at~$\loc_1$ 
implies that assigning~$\fa_i$ to~$\loc_1$ alongside the families in~$\asg^{-1}(\loc_1)$ yields a Pareto-improvement over~$\asg$. 

It remains to show that we can find a non-wasteful and feasible assignment in polynomial time. Start with an empty assignment and iterate over the families that find $\loc_1$ acceptable. If $\loc_1$ can accommodate such a family alongside everyone already assigned to $\loc_1$, assign the family to $\loc_1$. The resulting assignment is feasible and acceptable, because we never assign a family to a place it finds unacceptable. It is also non-wasteful, because if the place cannot accommodate a family~$\fa_i$ during the iteration when $\fa_i$ is considered, it also cannot accommodate~$\fa_i$ at any later point during the algorithm.
}

Our next results shows that \pareto-\RR\ can be solved efficiently if there are only a few families whose preferences contain ties, assuming that there are no lower quotas.
Recall that in the presence of lower quotas, \pareto-\RR\ is \NPh\ even if $\noser=1$ and $\tieno=0$, i.e., all preferences are strict, as shown in \cref{prop:wtbinpacking}.
The algorithm of \cref{prop:fpt_tieno} first applies serial dictatorship among families whose preferences do not contain ties, and then tries all possible assignments for the remaining families.

\begin{restatable}[\appsymb]{proposition}{propfpttieno}\label{prop:fpt_tieno}
\pareto-\RR\ is \FPT\ w.r.t.\ the number of families with ties $\tieno$, if there are no lower quotas.
\end{restatable}

\appendixproofwithstatement{prop:fpt_tieno}{\propfpttieno*}{
Our approach is to first apply serial dictatorship for each family whose preferences do not contain ties, and then try all possible assignments for the remaining families. 

Let $\FA_\succ$ denote the set of those families in our input instance~$I$ whose preferences are a linear order over the subset of~$\LOC$ they find acceptable (i.e., whose preference list does not contain ties).
By re-indexing the set~$\FA$, we may assume that $\FA_{\succ}=\{\fa_1,\dots,\fa_{n-\tieno} \}$. 

In the first phase of the algorithm,  starting from an empty assignment, we iterate over $i=1,\dots, n-\tieno$ and, at the $i$-th iteration, assign the family $\fa_i$ to the place most preferred by~$\fa_i$ among all places in~$\LOC$ that are acceptable for~$\fa_i$ and can accommodate~$\fa_i$ alongside the families already assigned to it. Let $\asg_\succ$ denote the assignment obtained at the end of this iteration.

In the second phase of the algorithm, we delete the families~$\FA_{\succ}$ from the instance, and reduce the upper quota of each place~$\loc_j \in \LOC$ with the load vector of~$\satur(\loc_j,\asg_\succ)$. We then solve \pareto-\RR\- on the obtained instance~$I'$ by applying \cref{prop:fpt-n}; let $\asg_\sim$ denote the returned assignment (note that there always exists a feasible and acceptable assignment, since there are no lower quotas for~$I$). We return the assignment $\asg:=\asg_\sim \cup \asg_\succ$.\footnote{Since the domains of~$\asg_\sim$ and $\asg_\succ$ partition~$\FA$, taking the union of~$\asg_\sim$ and $\asg_\succ$ yields an assignment for~$I$.}

The feasibility and acceptability follows from the feasibility and acceptability of~$\asg_\succ$ in~$I$, and of~$\asg_\sim$ in~$I'$, as well as the fact that the upper quota for~$\loc_j$ in~$I'$ is set to~$\capl_j-\satur(\loc_j,\asg_\succ)$. To see that $\asg$ is Pareto-optimal as well, assume for the sake of contradiction that some assignment~$\asg'$ is a Pareto-improvement over~$\asg$. 

First, observe that $\asg'(\fa_i)=\asg_\succ(\fa_i)$ for each $\fa_i \in \FA_\succ$: indeed, assuming otherwise, we obtain that there exists some family~$\fa_i$ such that $\asg'(\fa_{i'})=\asg_\succ(\fa_{i'})$ for each $i' <i$, but $\asg'(\fa_i) \succ_i \asg_\succ(\fa_i)$; however, this either contradicts the definition of~$\asg_\succ$ or the feasibility or acceptability of~$\asg'$. Therefore, we know that the restriction of~$\asg'$ to~$I'$ must be a Pareto-improvement over~$\asg_\sim$ which contradicts the correctness of our algorithm for \cref{prop:fpt-n}.
Hence, our algorithm always produces a correct output.
}

\paragraph{Maximizing utility.}
\label{sec:maxutil}

Let us start with a simple fixed-parameter tractable algorithm for \maxutil-\RR\ w.r.t~$n$, the number of families. \cref{prop:fpt-n} presents an FPT algorithm for parameter~$n$ based on the following approach: We first guess the partitioning~$\mathcal{F}$  of families arising from a maximum-weight feasible assignment, and then we map the partitions of~$\mathcal{F}$ to the places by computing a maximum-weight matching in an auxiliary bipartite graph.

\begin{restatable}[\appsymb]{proposition}{propfptn}\label{prop:fpt-n}
\feasible-, \pareto-, and \maxutil-\RR\ are \FPT\ w.r.t.~$n$.
\end{restatable}

\appendixproofwithstatement{prop:fpt-n}{\propfptn*}{
We present an algorithm for \maxutil-\RR, the result then follows from \cref{rem:maxutil-vs-pareto}.

Assume that our input instance~$I$ admits a feasible assignment, and let~$\asg$ be such an assignment with maximum utility. 
Our algorithm first guesses the set family~$\mathcal{F}=\{\asg^{-1}(\loc_j):\loc_j \in \LOC,\asg(\loc_j) \neq \emptyset\}$; since $\mathcal{F}$ is a partition of a subset of~$\FA$, there are $O(n^n)$ possible guesses to try.

The algorithm next creates the following edge-weighted bipartite graph $G=(\mathcal{F} \cup \LOC,E)$. A set of families $\Gamma \in \mathcal{F}$ is adjacent in~$G$ to some place~$\loc_j$ 
if and only if $\caplmin_j \leq \sum_{\fa_i \in \Gamma} \reqf_i \leq \capl_j$, and we set the weight of the edge connecting~$\Gamma$ and~$\loc_j$ (if it exists) as~$\sum_{\fa_i \in \Gamma} \profit_i[j]$. 
We compute a maximum-weight matching~$M^\star$ among those that cover the set~$\FA^\star=\{\fa_i \in \FA: \caplmin_i \neq 0\}$;
this can be done in $O((n+m)^3)$ time using e.g., the Hungarian algorithm. The algorithm outputs the assignment that to each place~$\loc_j$ covered by~$M^\star$ assigns the families contained in  $\Gamma \in \mathcal{F}$ where $\{\Gamma,\loc_j\} \in M^\star$.

The correctness of this algorithm follows from the observation that every matching~$M$ in~$G$ that covers~$\FA^\star$ corresponds to a feasible assignment whose utility is the weight of~$M$, and vice versa; note that feasibility is guaranteed by the definition of the edge set of~$G$ and the condition that the matching covers $\FA^\star$.
The total running time of the algorithm is $O(n^{n+3})$.
}

Let us now present a generalization of \cref{thm:onetboundedrmaxilp} for \maxutil-\RR\ restricted to \indifferent\ preferences. The algorithm for \cref{thm:fptservmaxreq} is based upon an $N$-fold IP formulation for this problem.
By \cref{rem:maxutil-vs-pareto}, the obtained algorithm also implies tractability for \pareto-\RR\ for \uniform\ utilities.

\begin{restatable}[\appsymb]{theorem}{fptservmaxreq}
\label{thm:fptservmaxreq}
\feasible-\RR, 
\pareto-\RR\ for \indifferent\ preferences, 
and 
\maxutil-\RR\ for \identical\ utilities, 
are \FPT\ w.r.t.~\ $t + \maxreq$.
\end{restatable}

\appendixproofwithstatement{thm:fptservmaxreq}{\fptservmaxreq*}{

We present an $N$-fold IP for \maxutil-\RR\ with \identical\ utilities as follows.
First, let $\R=\{\reqf_i: \fa_i \in \FA\}$ contain all requirement vectors associated with some family in our input instance.
By possibly re-indexing the families in~$\FA$, we can ensure $\R=\{\reqf_1,\dots,\reqf_{|\R|}\}$. 
Note that $|\R| \leq (\maxreq+1)^t$ due to the definition of~$\maxreq$. Let also $n_{\reqf}$ denote the number of families with requirement vector~$\reqf \in \R$.

We introduce a variable $x(\loc_j,\reqf)$ for each $\loc_j \in \LOC$ and $\reqf \in \R$ which is interpreted as the number of families with requirement vector~$\reqf$ assigned to~$\loc_j$. Additionally, we introduce ``slack'' variables $y(\loc_j,\ser_k)$ for each place $\loc_j \in \LOC$ and service $\ser_k \in \SER$ interpreted as the available free capacity for service~$\ser_k$ at place~$\loc_j$.
Consider the following integer program \ref{LP:nfold}:
\begin{linenomath*}
\begin{equation}
\leqnomode
\tag{ILP$^{N}$} 
\label{LP:nfold}
\max \sum_{\loc_j \in \LOC,\reqf \in \R} x(\loc_j,\reqf) \quad \textrm{such that} \qquad \qquad \phantom{ssssssss}
\end{equation}
\begin{align}
\notag \\[-22pt]
\label{con:cap-notexceeded}
& \forall \loc_j\in \LOC,\ser_k \in \SER: \!\!\!
&& y(\loc_j,\ser_k) + \sum_{\reqf \in \R} \reqf[k] x(\loc_j,\reqf) = \capl_j[k] \\
\label{con:fam-notoverused}
& \forall \reqf \in \R: 
&& \sum_{\loc_j \in \LOC} x(\loc_j,\reqf) \leq n_{\reqf} \\
\label{con:x-nonneg}
& \forall \loc_j \in \LOC, \reqf \in \R:
\!\!\!
&& 0 \leq x(\loc_j,\reqf) \\
\label{con:slack-y-bounds}
& \forall \loc_j\in \LOC,\ser_k \in \SER:
\!\!\!
&& 0 \leq y(\loc_j,\ser_k) \leq \capl_j[k]-\caplmin_j[k] 
\end{align}
\end{linenomath*}

\begin{claim}
There exists an integer solution to~\ref{LP:nfold} with value~$\finalutil$ if and only if there exists a feasible assignments for~$I$ with utility~$\finalutil$.
\end{claim}
\begin{claimproof}
Interpret $x(\loc_j,\reqf)$ as the number of families with requirement vector~$\reqf$ assigned to some place~$\loc_j$, and $y(\loc_j,\ser_k)$ as the unused capacity of~$\loc_j$ for service~$\ser_k$. Then 
 constraints~(\ref{con:cap-notexceeded}) and~(\ref{con:slack-y-bounds}) for each $\loc_j\in \LOC$ and $\ser_k \in \SER$ express the condition that the total requirement for service~$\ser_k$ of  all families assigned to~$\loc_j$ should be at least~$\caplmin_j[k]$ and at most~$\capl_j[k]$. 
Constraints~(\ref{con:fam-notoverused}) and (\ref{con:x-nonneg}) express the condition that for each $\reqf \in \R$, the total number of families with requirement~$\reqf$ assigned to some place in~$\LOC$ should not exceed~$n_{\reqf}$. This shows that constraints (\ref{con:cap-notexceeded})--(\ref{con:slack-y-bounds}) together characterize feasibility assignments. 

It remains to note that since utilities are \identical, the total utility of an assignment is proportional to the number of families assigned, which is expressed by the objective function. 
\end{claimproof}

\begin{restatable}[\appsymb]{claim}{clmnfold}
\label{clm:nfold}
\ref{LP:nfold} is an $N$-fold IP for $N=m$ with constraint~$A^{(m)}$ where  $A=\left( 
\begin{array}{@{}c@{}}
I_{|\R|} \,\, 0\\[2pt]
A_\R \,\,\, I_t
\end{array}
\right)$ for some $A_\R$ with $||A_{\R}\||_\infty \leq \maxreq$, and 
$I_h$ denotes the identity matrix of size~$h \times h$ for each $h \in \NN$.
\end{restatable}
\begin{claimproof}
Define the following matrices and vectors: 
\begin{linenomath*}
 \begin{align*}
A_{\R} &=(\reqf_1 \,\, \reqf_2 \,\, \dots \,\reqf_{|\R|}) \in \mathbb{N}^{\noser \times |\R|} \\
\vecx_j&=(x(\loc_j,\reqf_1) \,\,\, x(\loc_j,\reqf_2) \,\,\, \dots \,\,\, x(\loc_j,\reqf_{|\R|}))^\top
\in \mathbb{Z}^{|\R|} \\
\vecy_j&=(y(\loc_j,\ser_1) \,\,\, y(\loc_j,\ser_2) \,\,\, \dots \,\,\, y(\loc_j,\ser_\noser))^{\top} \in \mathbb{Z}^t
\end{align*}
\end{linenomath*}
Then constraints~(\ref{con:cap-notexceeded}) for all $\reqf \in \R$ but fixed $\loc_j \in \LOC$  can be written as
\begin{linenomath*} 
\begin{equation}
\label{eqn:AI-matrix-fixedj}
(A_{\R} \,\, I_\noser) \cdot
\left(
\begin{array}{@{}c@{}}
\vecx_j\\[2pt]
\vecy_j
\end{array}
\right) = \capl_j.
\end{equation}
\end{linenomath*}
Gathering (\ref{eqn:AI-matrix-fixedj}) for every $j \in [m]$, we obtain 
\begin{linenomath*}
\begin{equation}
\label{eqn:AI-matrix}
\left(
\begin{array}{cccc}
A_{\R} \,\, I_\noser & 0 & \dots & 0 
\\[2pt]
0 & A_{\R} \,\, I_\noser & \dots & 0 
\\[2pt] 
\vdots & \vdots & \ddots & \vdots 
\\[2pt]
0 & 0 & \dots & A_{\R} \,\, I_\noser
\end{array}
\right) \cdot
\left(
\begin{array}{@{}c@{}}
\vecx_1\\[2pt]
\vecy_1\\[2pt]
\vecx_2\\[2pt]
\vecy_2\\[2pt]
\vdots \\[2pt]
\vecx_m\\[2pt]
\vecy_m
\end{array}
\right) =
\left(
\begin{array}{@{}c@{}}
\capl_1 \\[2pt]
\capl_2 \\[2pt]
\vdots \\[2pt]
\capl_m 
\end{array}
\right).
\end{equation}
\end{linenomath*}
Constraints~(\ref{con:fam-notoverused}) for all $\reqf \in \R$ can be written as
\begin{linenomath*}
\begin{equation}
\label{eqn:IR-matrix}
\left(I_{|\R|} \,\, I_{|\R|} \,\, \dots, \,\, I_{|\R|}\right) \cdot
\left(
\begin{array}{@{}c@{}}
\vecx_1\\[2pt]
\vecx_2\\[2pt]
\vdots \\[2pt]
\vecx_m
\end{array}
\right) \leq  
\left(
\begin{array}{@{}c@{}}
n_{\reqf_1}\\[2pt]
n_{\reqf_2}\\[2pt]
\vdots \\[2pt]
n_{\reqf_{|\R|}}
\end{array}
\right).
\end{equation}
\end{linenomath*}
Recall also that $N$-fold IPs can handle lower and upper bounds on variables, as required by constraints~(\ref{con:x-nonneg}) and~(\ref{con:slack-y-bounds}).
Therefore, we can observe that, summing up equations~(\ref{eqn:AI-matrix}) and inequalities~(\ref{eqn:IR-matrix}), constraints (\ref{con:cap-notexceeded})--(\ref{con:slack-y-bounds}) can be formed as an $N$-fold IP\footnote{The fact that (\ref{eqn:IR-matrix}) is not an equality but an inequality does not cause problems, as shown by Knop et al. in the full version of their paper~\cite{KKM-nfoldIP-full}.
} whose coefficient matrix is
\begin{linenomath*}
\begin{equation*}
A^{(m)}= \left(
\begin{array}{cccc}
I_{|\R|} \,\, 0& I_{|\R|} \,\, 0 & \dots, & I_{|\R|} \,\, 0
\\[2pt]
A_{\R} \,\,\, I_\noser & 0 & \dots & 0 
\\[2pt]
0 & A_{\R} \,\,\, I_\noser & \dots & 0 
\\[2pt] 
\vdots & \vdots & \ddots & \vdots 
\\[2pt]
0 & 0 & \dots & A_{\R} \,\,\, I_\noser
\end{array}
\right). 
\end{equation*}
\end{linenomath*}
for the matrix $A=\left( 
\begin{array}{@{}c@{}}
I_{|\R|} \,\, 0\\[2pt]
A_\R \,\,\, I_t
\end{array}
\right)$.
\end{claimproof}

The algorithm by Hemmecke et al~\cite[Theorem 6.2]{HOR-2013} solves such an $N$-fold IP for $N=m$ in time $||A||_\infty^{O(\noser \cdot |\R|^2+|\R|\cdot \noser^2)} \cdot m^3 \cdot L$ where~$L$ denotes the binary encoding of the constants on the right-hand side of the IP, the upper and lower quotas, and the objective function\footnote{See
Eisenbrand et al.~\cite{EHK-2018} for a more recent, slightly faster algorithm. 
}; in our case $L=O(\log(n+m+\maxcap))$. 
Recall that $|\R| \leq (\maxreq +1)^\noser$, and that each entry in~$A$ is an integer at most~$\maxreq$. We can  observe that the running time is fixed-parameter tractable w.r.t.\ parameter~$\noser+\maxreq$.}

Our next algorithm is applicable in a more general case than \cref{thm:fptservmaxreq} (which works only when preferences or utilities are equal) at the cost of setting $m+t+\maxreq$ as the parameter (recall that the parameter considered in \cref{thm:fptservmaxreq} is $t+\maxreq$). 
\cref{thm:fptservmaxreq} is based on an ILP formulation that solves \pareto-\RR\ for arbitrary preferences as well as \maxutil-\RR\ for a broad range of utilities.

\begin{restatable}[\appsymb]{theorem}{thmfptservmaxcaploc}\label{thm:fptservmaxcaploc}
The following problems are \FPT\ w.r.t.\ parameter ${m + t + \maxreq}\colon$ 
\begin{compactitem}
\item \pareto-\RR,
\item \maxutil-\RR\ on instances where the number of different utility values is at most $g(m+t+\maxreq)$ for some computable function~$g$.
\end{compactitem}
\end{restatable}

\begin{proof}[Proof sketch]
The main idea is that there is no need to distinguish between families that have the same utilities and requirements. Since  the number of possible requirement vectors and the number of possible utility vectors are both bounded by a function of the parameter, the number of family types will also be bounded. This allows us to define a variable for each place and family type describing the number of families of a given type assigned to a given place. The resulting ILP contains a bounded number of variables and constraints, and is therefore solvable by standard techniques in FPT time~\cite{lenstra1983integer}.
\end{proof}
\appendixproofwithstatement{thm:fptservmaxcaploc}{\thmfptservmaxcaploc*}{
Recall that the reduction from~\pareto-\RR\ to  \maxutil-\RR\ described in \cref{rem:maxutil-vs-pareto} constructs an instance where utility values fall into the range $[m]$ with the sole exception of $-m\cdot n$; hence, the number of different utility values is at most~$m+1$. Hence, it suffices to solve \maxutil-\RR, as the first result follows from the second one.

We are going to present an ILP  for \maxutil-\RR.

First, let $\R=\{\reqf_i: \fa_i \in \FA\}$ contain all requirement vectors associated with some family in our input instance~$I$ of \maxutil-\RR.
Then $|\R| \leq (\maxreq+1)^t$ due to the definition of~$\maxreq$.  
Second, let $\UU=\{\profit_i: \fa_i \in \FA\}$ 
contain all utility vectors associated with some family. Since the number of different utility values is at most $g(m+t+\maxreq)$, we know that- $|\UU| \leq (g(m+t+\maxreq))^m$.

We define the \emph{type} of a family~$\fa_i \in \FA$ as $(\reqf_i,\profit_i)$, so two families have the same type, if they have the same requirement and utility vectors. Let $\TT$ denote the set of all family types appearing in the instance; then 
$|\TT| \leq |\R| \cdot |\UU| \leq (\maxreq+1)^t \cdot (g(m+t+\maxreq))^m$,
so the number of family types is bounded by a function of the parameter.
We define the two type sets 
$\TT_{(\reqf,\cdot)}=\{(\reqf,\profit) \in \TT: \profit \in \UU\}$
and~$\TT_{(\cdot, \profit)}=\{(\reqf,\profit) \in \TT: \reqf \in \R\}$.
Also, for each type~$\tau \in \TT$ we let $n_{\tau}$ denote the number of families of type~$\tau$ in the instance.

\smallskip
We introduce a variable $x(\loc_j,\tau)$ for each $\loc_j \in \LOC$ and $\tau \in \TT$ which is interpreted as the number of families of type~$\tau$ assigned to~$\loc_j$. The number of variables is therefore $|\TT|\cdot m$.
Consider the following integer program \ref{LP2}:
\begin{linenomath*}
\begin{equation}
\leqnomode
\tag{ILP$_2$} 
\label{LP2}
\max \sum_{\loc_j \in \LOC} \sum_{\profit \in \UU} \sum_{\tau \in \TT_{(\cdot,\profit)}} \profit[j] \cdot x(\loc_j,\tau) \quad \textrm{such that} \phantom{ssssssss}
\end{equation}
\begin{align}
\notag \\[-22pt]
\label{con:LP2:cap-satisfied}
& \forall \loc_j\in \LOC: \!\!\!
&& \caplmin_j \leq f
\sum_{\reqf \in \R} \sum_{\tau \in \TT_{(\reqf,\cdot)}}  x(\loc_j,\tau) \cdot \reqf 
\leq \capl_j \\
\label{con:LP2:fam-notoverused}
& \forall \tau \in \TT: 
&& 0 \leq \sum_{\loc_j \in \LOC} x(\loc_j,\tau) \leq n_{\tau} 
\end{align}
\end{linenomath*}

\begin{restatable}[\appsymb]{claim}{clmlptwocorrect}
\label{clm:lp2correct}
There exists an integer solution to~\ref{LP2} with value~$\finalutil$ if and only if  there exists a feasible assignments for~$I$ with utility~$\finalutil$.
\end{restatable}

\begin{claimproof}
Interpret $x(\loc_j,\tau)$ as the number of families of type~$\tau$ that are assigned to place~$\loc_j$. 
Then for each $\loc_j \in \LOC$, 
\begin{linenomath*}
\begin{equation*}
\sum_{\reqf \in \R} \sum_{\tau \in \TT_{(\reqf,\cdot)}}  x(\loc_j,\tau) \cdot \reqf,
\end{equation*}
\end{linenomath*}
formulates exactly the load of~$\loc_j$, which implies that 
constraint~(\ref{con:LP2:cap-satisfied}) expresses the condition of feasibility. 
Moreover, for each family type~$\tau \in \TT$, the expression $\sum_{\loc_j \in \LOC} x(\loc_j,\tau)$ formulates the total number families of type~$\tau$ assigned to some place. Hence, constraint~(\ref{con:LP2:fam-notoverused}) expresses the condition that the assignment can only assign at most~$n_\tau$ families in total.
This means that integer solutions to~\ref{LP2} correspond to feasible assignments for~$I$, and vice versa.
Finally, notice that the expression
\begin{linenomath*}
\begin{equation*}
\sum_{\loc_j \in \LOC} \sum_{\profit \in \UU} \sum_{\tau \in \TT_{(\cdot,\profit)}} \profit[j] \cdot x(\loc_j,\tau) 
\end{equation*}
\end{linenomath*} 
formulates exactly the utility of the assignment, and thus a solution for~\ref{LP2} with value~$\finalutil$ implies the existence of a feasible assignment~$\asg$ with $\util(\asg)=\finalutil$, and vice versa. 
\end{claimproof}

Since the number of variables in \ref{LP2} is bounded by a function of the parameter~$m+t+\maxreq$ and the number of constraints is \FPT\ w.r.t.\ to the parameter, 
the problem can be solved in \FPT\ time~\cite{lenstra1983integer}.
}

Taking an even stronger parameterization than \cref{thm:fptservmaxcaploc}, namely $m+t+\maxcap$, 
yields fixed-parameter tractability:
in \cref{xp:mt} we present an algorithm running in $O((\maxcap)^{mt}nm)$ time.
This algorithm is a straightforward adaptation of the textbook dynamic programming method for \textsc{Knapsack}. 
The same approach was also
used by Gurski et al.~\cite[Proposition~34]{gurski2019knapsack} to solve a simpler variant of \maxutil-\RR\ without lower quotas and with a single service.

\begin{restatable}[\appsymb]{proposition}{xpmt}
\label{xp:mt}
\feasible-, \pareto-\RR, and \maxutil-\RR\ are in \XP\ w.r.t.\ $m + \noser$ and are \FPT\ w.r.t.\ $m + \noser + \maxcap$.
\end{restatable}

\appendixproofwithstatement{xp:mt}{\xpmt*}{
We present an algorithm for \maxutil-\RR; due to \cref{rem:maxutil-vs-pareto}, this can be used to solve \pareto-\RR\ as well. Our approach is a straightforward adaptation of the textbook dynamic programming for \textsc{Knapsack}.

A \emph{load state} is a tuple $(\status_1,\dots,\status_m)$ of vectors where  $\status_j \in \mathbb{N}^t$ satisfies $\status_j \leq \capl_j$ for each $\loc_j \in \LOC$. 
Define $\Q$ as the set of all load states, and observe that $|\Q|\leq (\maxcap)^{mt}$.
Define also $\FA_i=\{\fa_1,\dots,\fa_i\}$ for each $i \in [n]$.

For each possible load state~$\statQ=(\status_1,\dots,\status_m) \in \Q$ 
and each family~$i \in [n]$ we compute the maximum utility~$T(\statQ,i)$ that can be achieved by some assignment $\asg:\FA_i \rightarrow \LOC$ with $\satur(\loc_j,\asg)=\status_j$ for each $j \in [m]$; if no such assignment exists, then we will write $T(\statQ,i)=-\infty$.

We start by computing the values $T(\statQ,i)$, for $i=1$ and for each load state $\statQ=(\status_1,\dots,\status_m) \in \Q$ as follows:
\begin{compactitem}
\item
if $\status_j=\reqf_1$ for some $j \in [m]$ 
and $\status_{j'}=0$ for each $j' \in [m] \setminus \{j\}$, then 
$T(\statQ,1)=\profit_1[j]$;
\item if $\status_j=0$ for each $j \in [m]$, then 
$T(\statQ,1)=0$;
\item otherwise $T(\statQ,1)=-\infty$.
\end{compactitem}
The correctness of these values can be seen by observing that there are exactly $m+1$ possible assignments for $\FA_1$: either we leave family~$\fa_i$ unassigned, or we assign it to one of the places in~$\LOC$. 
Notice that we do not require feasibility in the definition of~$T(\statQ,i)$.

After the above initialization for the case $i=1$, we compute the values $T[\statQ,i]$  for each $i=2,\dots,n$ and for all $\statQ \in \Q$ using the following recursive formula. Let $\statQ=(\status_1,\dots,\status_m)$; then
\begin{linenomath*}
\begin{equation}
\label{eqn:DP-recursion}
T(\statQ,i)=
\max \left \{T(\statQ,i-1), \max_{\substack{j \in [m] \\ \status_j\geq \reqf_i}}
 \left\{T(\statQ^{(j)},i-1)+\profit_i[j] \right\} 
 \right\}
\end{equation}
\end{linenomath*}
where 
\begin{linenomath*}
\begin{equation}
\label{eqn:Q'-def-for-DP}
\statQ^{(j)}=(\status_1,\dots,\status_{j-1},\status_j-\reqf_i,\status_{j+1},\status_m).
\end{equation} 
\end{linenomath*}
Finally, the algorithm returns the maximum utility achievable by some feasible load state, that is, its output is
\begin{linenomath*}
\begin{equation*}
\max \{T((\status_1,\dots,\status_m),n):\caplmin_j \leq \status_j \leq \capl_j \textrm{ for each }j \in [m]\}.
\end{equation*} 
\end{linenomath*}

\begin{claim}
\label{clm:DP-recursion-correct}
Equation (\ref{eqn:DP-recursion}) correctly computes $T(\statQ,i)$ for each $\statQ \in \Q$ and $i \geq 2$.  
\end{claim}
\begin{claimproof}
We prove the claim by recursion on~$i$, building on the fact that the algorithm computes the values for $T(\statQ,1)$  for each $\statQ \in \Q$ correctly. Thus, suppose that the computations are correct for $i-1$, and consider the formula~(\ref{eqn:DP-recursion}) for~$i$.

Consider some load state $\statQ=(\status_1,\dots,\status_m) \in \Q$. 
Assume first that $\asg:\FA_i \rightarrow \LOC$ is an assignment with maximum utility among those that satisfy
$\satur(\loc_j,\asg)=\status_j$ for each $j \in [m]$; this means $T(\statQ,i)=\util(\asg)$. 
Let $\asg'$ denote the restriction of~$\asg$ to~$\FA_{i-1}$.
If $\asg$ leaves $\fa_i$ unassigned, then the load of every place is the same under~$\asg'$ and under~$\asg$, which implies $T[\statQ,i-1] \geq \util(\asg)$ by our inductive hypothesis for $i-1$. 
If $\asg$ assigns $\fa_i$ to some place~$\loc_j \in \LOC$, then 
the load of $\loc_j$ under~$\asg'$ is $\status_j - \reqf_i$, while the load of every other place~$\loc_{j'} \in \LOC \setminus \{\loc_j\}$ is the same as under~$\asg$,  that is, $\status_{j'}$. This means that $T[\statQ^{(j)},i-1] \geq \util(\asg')=\util(\asg)- \profit_i[j]$
where $\statQ^{(j)}$ is defined by~(\ref{eqn:Q'-def-for-DP}), again using our hypothesis.
Hence, irrespective of the value~$\asg(\fa_i)$, the value on the right-hand side of~(\ref{eqn:DP-recursion}) is at least $\util(\asg)=T(\statQ,i)$. 

It remains to show the reverse direction, so assume that the right-hand side of~(\ref{eqn:DP-recursion}) is $\hat{u}$; we are going create an assignment~$\hat\asg$ which satisfies $\satur(p_j,\hat\asg)=\status_j$ for each $j \in [m]$ and has utility at least~$\hat{u}$.
First, if $\hat{u}=T(\statQ,i-1)$, then by our inductive hypothesis we know that there exists an assignment $\asg:\FA_{i-1} \rightarrow \LOC$ with $\util(\asg)=\hat{u}$ and satisfying $\status_j=\satur(\loc_j,\asg)$ for each $j \in [m]$. Then setting $\hat\asg=\asg$ is sufficient.
Otherwise, the right-hand side of~(\ref{eqn:DP-recursion}) must be defined by $\hat{u}=T(\statQ^{(j)},i-1)$ for some index~$j \in [m]$.
By our inductive hypothesis, there exists some assignment~$\asg':\FA_{i-1}\rightarrow \LOC$ with utility $u^\star-\profit_i[j]$ that yields the load state~$\statQ^{(j)}$ as defined by~(\ref{eqn:Q'-def-for-DP}), i.e., for which the load of~$\loc_j$ is $\status_j-\reqf_i$ while the load of each remaining place~$\loc_{j'}$ is~$\status_{j'}$. In this case, we can extend~$\asg'$ by assigning~$\fa_i$ to~$\loc_j$, giving us the assignment~$\hat\asg$  with the required properties and utility~$\hat{u}$.
\end{claimproof}

Due to \cref{clm:DP-recursion-correct} and by the definitions of the function~$T(\cdot)$, the correctness of our algorithm follows. Each computation step described by the recursion~(\ref{eqn:DP-recursion}) requires $O(m)$ time, and we perform it at most~$|\Q|n$ times. By $|\Q| \leq (1+\maxcap)^{mt}$, the overall running time is $(1+\maxcap)^{mt}O(nm)$.

By using standard techniques, we can not only compute the utility of an optimal, feasible assignment, but also determine a maximum-utility feasible assignment itself.
}

We close this section by mentioning that a simple \XP\ algorithm exists for the case when the parameter is the desired total utility~$\finalutil$, and there are no lower quotas (cf. \cref{prop:wtbinpacking} stating the intractability of the case $\finalutil=0$ when lower quotas are allowed).

\begin{restatable}[\appsymb]{proposition}{xpfinalutil}
\label{xp:finalutil}
\feasible-, \pareto- and \maxutil-\RR\ are in \XP\ w.r.t.\ the desired utility~$\finalutil$ if there are no lower quotas.
\end{restatable}

\appendixproofwithstatement{xp:finalutil}{\xpfinalutil*}{
Since there are no lower quotas, we have no reason to match a family $\fa_i \in \FA$ to a place $\loc_j \in \LOC$ if $\profit_i[j] \leq 0$.
Since utilities are integral, every pair $\fa_i \in \FA, \loc_j \in \LOC$ such that $\asg(\fa_i) = \loc_j$ and $\profit_i[j] > 0$ contributes at least $1$ to the final utility. Thus we can iterate over all subsets of family-place pairs of size at most $\finalutil$ such that every pair has a positive utility, and verify whether any of them gives raise to a feasible assignment $\asg$ with $\util(\asg) \geq \finalutil$. There are at most $O((nm)^{\finalutil})$ subsets of pairs. 
}

\section{Conclusion}
\label{sec:conclusion}
We provided a comprehensive parameterized complexity analysis for three variants of \RefRes, which focus on ensuring feasibility, maximizing utility, and achieving Pareto optimality.
There remain some interesting parameter combinations for which the complexity of these problems is open, e.g., is \maxutil-\RR\  FPT with parameter~$m+t+\maxreq$ for arbitrary utilities? 
 Another exciting line of future research is to explore the possibilities of tailoring the proposed algorithms to efficiently solve practical instances, and determining which parameterizations are the most relevant in different real-world applications.
We believe that our ILP algorithms could perform substantially faster then their theoretical bounds, as the current day ILP solvers are efficient. However, it could also be that a straightforward ILP formulation with a variable for each family and each place outperforms our specialized formulations.

\clearpage

\begin{ack}
Jiehua Chen and Sofia Simola are supported by the Vienna Science and Technology Fund (WWTF)~[10.47379/VRG18012].
Ildik\'o Schlotter is supported by the Hungarian Academy of Sciences under its Momentum Programme (LP2021-2) and its J\'anos Bolyai Research Scholarship.
\end{ack}



\bibliography{bib}

\clearpage

\ifshort

\else
\begin{appendices}

\appendixtext

\section{Further parameterizations}\label{sec:moreparameters}

\newcommand{\refrmmtmaxrequmax}{[R\ref{rm:mtmaxrequmax}]}
\newcommand{\refrmparetosatwithties}{[R\ref{rm:paretosatwithties}]}
\newcommand{\refrmfinalutilmaxreq}{[R\ref{rm:finalutilmaxreq}]}
\newcommand{\equtil}{\;\;\scriptsize eq.\ util.}

\newcommand{\algpos}{$^+$}

\begin{table}[t!]
  \centering
  \extrarowheight=.5\aboverulesep
  \aboverulesep=1pt
  \belowrulesep=1pt
  \addtolength{\extrarowheight}{\belowrulesep}
    \begin{tabular}{@{}l@{\,} | @{\,} c@{\,}l @{\,}c@{\,} c@{\,}l @{}c@{\,}} 
\toprule
      {\small Other param.}  & \multicolumn{2}{@{\;}c@{\;}}{$\maxutilf$} &  \multicolumn{2}{@{\;}c@{\;}}{$\finalutil$} \\
&  \multicolumn{1}{c}{\nolower~/~\withlower}& &\multicolumn{1}{c}{\nolower~/~\withlower}&\\
      \midrule
       $ - $ &  \tNPh\tid/\tNPh\tid  &\cite{gurski2019knapsack}  &  \our{\tFPT}/\our{\tNPh}\tid & \reffptfinalutiltone/\refpropwtbinpacking   \\  \hline 
              $m$ & \our{\tWoneh\tid}/\our{\tWoneh\tid}& \refpropwtbinpacking & \our{\tFPT}/\our{\tWoneh\tid}& \reffptfinalutiltone/\refpropwtbinpacking \\[-.8ex]   
       & \our{\tXP}/\our{\tXP} & \refxpmt & \our{\tFPT}/\our{\tXP} & \reffptfinalutiltone/\refxpmt  \\ \hline
       $\maxreq$ &\our{\tNPh}/\our{\tNPh}&\refthmsatwithties  & \our{\tFPT}/\OQ, \our{\tXP}\algpos & \reffptfinalutiltone/\refrmfinalutilmaxreq\\ [-.8ex] 
       \equtil & \our{\tFPT}/\our{\tFPT} & \refthmonetboundedrmaxilp & \our{\tFPT}/\our{\tFPT} & \refthmonetboundedrmaxilp\\ \hline
        \hline
\toprule
$-$ &  \tNPh\tid/\tNPh\tid  &\cite{gurski2019knapsack} & \tWoneh\tid/\our{\tNPh}\tid & \cite{gurski2019knapsack}/\refpropwtbinpacking \\[-.8ex]
  & &  & \our{\tXP}/\our{\tNPh}\tid & \refxpfinalutilmodel/\refpropwtbinpacking \\ \hline \hline
 $\noser$ &  \tNPh\tid/\tNPh\tid  &\cite{gurski2019knapsack} & \OQ,\our{\tXP}/\our{\tNPh}\tid & \refxpfinalutilmodel/\refpropwtbinpacking \\ \hline
 $m$ & \tNPh\tid/\tNPh\tid & \cite{gurski2019knapsack} & \tWoneh\tid/\tNPh\tid & \cite{gurski2019knapsack}/\refpropfeashardm \\[-.8ex]
& &  & \our{\tXP}/\tNPh\tid & \refxpfinalutilmodel/\refpropfeashardm\\ \hline
 $\maxreq$  & \tNPh\tid/\tNPh\tid & \cite{gurski2019knapsack} & \tWoneh\tid/\our{\tNPh}\tid & \cite{gurski2019knapsack}/\refpropfeashardm \\[-.8ex]
& & &\our{\tXP}/\our{\tNPh}\tid & \refxpfinalutilmodel/\refpropfeashardm \\ \hline \hline
$m + \noser$ & \our{\tWoneh}\tid/\our{\tWoneh}\tid & \refpropwtbinpacking  & \OQ/\our{\tWoneh}\tid & \refpropwtbinpacking  \\ [-.8ex]
 & \our{\tXP}/\our{\tXP} & \refxpmt & \our{\tXP}/\our{\tXP} & \refxpmt\\ \hline
    $m + \maxreq$ & \tNPh\tid/\tNPh\tid & \cite{gurski2019knapsack}  & \tWoneh\tid/\our{\tNPh}\tid & \cite{gurski2019knapsack}/\refpropfeashardm \\[-.8ex]
   & & & \our{\tXP}/\our{\tNPh}\tid & \refxpfinalutilmodel/\refpropfeashardm \\ \hline 
    $\noser + \maxreq$ & \our{\tNPh}/\our{\tNPh} & \refthmsatwithties & \OQ,\our{\tXP}/\OQ & \refxpfinalutilmodel \\ [-.8ex]
    \equtil   & \our{\tFPT}/\our{\tFPT} & \refthmfptservmaxreq & \our{\tFPT}/\our{\tFPT} & \refthmfptservmaxreq \\ \hline \hline
    $m + \noser + \maxreq$ & \our{\tFPT}/\our{\tFPT}\algpos & \refrmmtmaxrequmax & \our{\tFPT}/\our{\tFPT}\algpos & \refrmmtmaxrequmax\\ \bottomrule
    \end{tabular}
  \caption{
  Additional results regarding parameters $\finalutil$ and $\maxutilf$ for \maxutil-\RR.
Above: Results for the single-service case ($t = 1$).
We skip the parameterization by $n$ since for this case since it is FPT for the more general case. 
Bold faced ones are obtained in this paper. \nolower\ (resp.\ \withlower) refers to the case when lower quotas are zero (resp.\ may be positive).
Here, \tNPh\ means that the problem remains NP-hard even if the corresponding parameter(s) are constant.
All hardness results hold for \binary\ utilities.
Additionally, \tid\ means hardness results hold even for equal utilities, and \algpos\ means the algorithm only works when the utilities are non-negative.
 \label{table:utilparams}}
\end{table}

\begin{table}[t!]
  \centering
  \extrarowheight=.5\aboverulesep
  \aboverulesep=1pt
  \belowrulesep=1pt
  \addtolength{\extrarowheight}{\belowrulesep}
    \begin{tabular}{@{}l@{\,} | c@{\,}c @{\,}c@{\,} c@{\,} } 
\toprule
      {\small Other param.}  & \multicolumn{2}{@{\;}c@{\;}}{$\tieno$} \\
& \multicolumn{1}{c}{\nolower~/~\withlower}& & & \\
      \midrule
$-$ &  \our{\tFPT}/\our{\tNPh} & \refpropfpttieno/\refpropwtbinpacking \\ \hline 
$m$ & \our{\tFPT}/\our{\tWoneh} & \refpropfpttieno/\refpropwtbinpacking \\ [-.8ex] 
& \our{\tFPT}/\our{\tXP} & \refpropfpttieno/\refxpmt  \\ \hline
$\maxreq$ &  \our{\tFPT}/\our{\tNPh} & \refpropfpttieno/\refrmparetosatwithties \\ \hline \hline \toprule

$ - $ & \our{\tFPT}/\our{\tNPh} & \refpropfpttieno/\refpropfeashardm  \\ \hline \hline
$m$ & \our{\tFPT}/\our{\tNPh} & \refpropfpttieno/\refpropfeashardm \\ \hline
$\noser $  & \our{\tFPT}/\our{\tNPh} & \refpropfpttieno/\refpropwtbinpacking \\ \hline 
$\maxreq$ & \our{\tFPT}/\our{\tNPh} & \refpropfpttieno/\refpropfeashardm \\ \hline \hline
$m + \noser$ & \our{\tFPT}/\our{\tWoneh} & \refpropfpttieno/\refpropwtbinpacking \\ [-.8ex]
 & \our{\tFPT}/\our{\tXP} & \refpropfpttieno/\refxpmt \\ \hline
 $m + \maxreq$ & \our{\tFPT}/\our{\tNPh} & \refpropfpttieno/\refpropfeashardm \\ \hline
 $t + \maxreq$ & \our{\tFPT}/\our{\tNPh} & \refpropfpttieno/\refrmparetosatwithties \\ \hline \hline
 $m + \noser + \maxreq$ & \our{\tFPT}/\our{\tFPT} & \refthmfptservmaxcaploc \\ \bottomrule
\end{tabular}
  \caption{
  Additional results regarding parameter $\tieno$ for \pareto-\RR. Since the problem is \FPT\ w.r.t.\ $\tieno$ when there are no lower quotas by \cref{prop:fpt_tieno}, we only look at the general case.
Above: Results for the single-service case ($t = 1$).
We skip the parameterization by $n$ since for this case since it is FPT for the more general case. 
Here, \tNPh\ means that the problem remains NP-hard even if the corresponding parameter(s) are constant.
 \label{table:tieparams}}
\end{table}

In this section we discuss our results regarding parameters that were not included in detail in the main paper. We show \FPT\ results for the parameters $\sumreqs, \sumcaps$, and $\sumutil$, and present in \cref{table:utilparams,table:tieparams} a nearly complete complexity picture regarding the parameters $\maxutilf, \finalutil$, and~$\tieno$.

We define $\tieno$ as the number of families who have ties in their preferences. That is, the number of families $\fa_i \in \FA$ such that there are two places $\loc_j, \loc_{j'} \in \LOC$ such that $\fa_i$ finds both $\loc_j$ and $\loc_{j'}$ acceptable and is indifferent between them. Many of the hardness-results for this parameter follow from the fact that \pareto-\RR\ is \NP-hard even when there is only one place, which is why we did not include it in the main paper. We do however discover that \pareto-\RR\ is \FPT\ w.r.t.\ $\tieno$ when there are no lower quotas (\cref{prop:fpt_tieno}).

Note that families with all-zero requirement vectors can be assigned arbitrarily without affecting the feasibility, and thus we may disregard such families and assume that $n \leq \sum_{\fa_i \in F} \sum_{s_k \in \SER} \reqf_i[k]=\sumreqs$. 
Therefore, \cref{prop:fpt-n} implies fixed-parameter tractability of our problems when parameterized by the sum of all requirements.

\begin{corollary}
\feasible-, \pareto-, and \maxutil-\RR\ are \FPT\ w.r.t.~$\sumreqs$.
\end{corollary}

Another simple observation yields fixed-parameter tractability for parameter~$\sumutil$ in the case when there are no lower quotas. Indeed, in such a case we can discard all families with all-zero utility vectors, as they can be deleted from any assignment without changing its feasibility, acceptability, or total utility. Therefore, we may assume that $n \leq \sum_{\fa_i \in F} \sum_{\loc_j \in \LOC} \profit_i[j]=\sumutil$, leading to the following consequence of \cref{prop:fpt-n}.

\begin{corollary}
\maxutil-\RR\ is \FPT\ w.r.t.~$\sumutil$ if there are no lower quotas.
\end{corollary}

Finally, it is also not hard to see that setting $\sumcaps$ as the parameter also yields fixed-parameter tractability.

\begin{proposition}
\feasible-, \pareto-, and \maxutil-\RR\ are \FPT\ w.r.t.~$\sumcaps$.
\end{proposition}
\begin{proof}
By definition, we have $\maxcap \leq \sumcaps$. 
It is also clear we can discard all places with all-zero upper quotas, which means that we may assume $m \leq \sumcaps$. 
Finally, we can also discard all services for which every place has zero upper quota (removing also all families that require such a service), implying that we may suppose $t \leq \sumcaps$. 
Hence we have $m+t+\maxcap \leq 3\sumcaps$, and the result follows from \cref{xp:mt}.
\end{proof}

\subsection{Additonal Results}
\label{subsec:additonal}
In this section, we make some remarks on how existing results can be modified to show parameterized results regarding $\maxutilf, \finalutil$ and $\tieno$.

\begin{remark}
\label{rm:paretosatwithties}
\textit{\pareto-\RR\ for $ \noser = 1$ is \NPh\ even when $\maxreq = \maxcap = 2$ and no family has ties in its preferences.}
\end{remark}

\begin{proof}
We modify the proof of \cref{thm:satwithties}.
Observe that the total service requirements of the families are $R = 6|X|$, and the total upper quotas of the localities are $Q = 2|C| + 4|X|$.  Since every literal appears in two clauses, and each clause has three literals, there are $\frac{4|X|}{3}$ clauses. Thus $Q = 2 \frac{4|X|}{3} + 4|X| \geq 6|X| = R$. 
We create $Q - R$ many dummy families which each require one unit of service. They find all places acceptable, and rank them in an arbitrary order.
We also modify the preferences of the existing families so that they rank all the families they find acceptable in an arbitrary order. 
Now, $\tieno = 0$.
We set the lower quotas of the places equal to their upper quotas.

Now there is a feasible and acceptable assignment if and only if there is an assignment that assigns every family to a place it finds acceptable. Since the dummy families find every place acceptable, it is sufficient to look into finding an assignment that places the original families to a place they find acceptable. The proof of \cref{thm:satwithties} shows this is \NP-hard. Since a Pareto-optimal, feasible, and acceptable assignment must assign every family to a place they find acceptable, \pareto-\RR\ must be \NP-hard even when $ \noser = 1$, $\maxreq = \maxcap = 2$ and no family has ties in its preferences.
\end{proof}

\begin{remark} \label{rm:mtmaxrequmax}
\textit{If utilities are non-negative, \maxutil-\RR\ is \FPT\ w.r.t.\ $m + \noser + \maxreq + \maxutilf$. If there are no lower bounds, \maxutil-\RR\ is \FPT\ w.r.t.\ $m + \noser + \maxreq + \maxutilf$ regardless of the utilities.}
\end{remark}
\begin{proof}
If there are no lower bounds, we can set any negative utility to $-1$. There is never a reason to match a family to a place where it has negative utility.

We can observe the statement from the proof of \cref{thm:fptservmaxcaploc}. We know that the number of different utility vectors families may have is at most $(\maxutilf + 1)^m$ (or $(\maxutilf + 2)^m$ when there are no lower bounds) and the number of different requirement vectors is, as before, at most $(\maxreq + 1)^{\noser}$. Thus the number of different family types is at most $\maxutilf^m (\maxreq + 1)^{\noser}$ (or $(\maxutilf + 2)^m (\maxreq + 1)^{\noser} $ when there are no lower bounds), which is a function of the parameter.
Rest of the proof proceeds as in the proof of \cref{thm:fptservmaxcaploc}.
\end{proof}

\begin{remark} \label{rm:finalutilmaxreq}\textit{
\maxutil-\RR\ is \XP\ w.r.t. $\maxreq + \finalutil$ when $\noser = 1$ and the utilities are non-negative. 
}
\end{remark}

\begin{proof}
Similarly to the proof of \cref{xp:finalutil}, we iterate over all the size at most $\finalutil$ subsets of family-place pairs, and see if assigning them according to the pairs gives sufficient utility. If yes, we assign the families in the pairs to their places and update the upper and lower quotas accordingly. As we have already obtained sufficient utility, we can ignore the utilities of the families and use \cref{thm:onetboundedrmaxilp} to find a feasible assignment for the updated instance in \FPT-time w.r.t.~$\maxreq$.
\end{proof}
\end{appendices}
\fi
\end{document}